\documentclass{article}

\usepackage{amsfonts,amsmath,amsthm,amssymb}
\usepackage{mathtools}
\usepackage{geometry}
\usepackage{float}
\usepackage{xcolor}
\usepackage{url}
\usepackage{circuitikz}
\usepackage{subfigure}
\usepackage{subcaption}
\usepackage{authblk}
\usepackage{natbib}

\usetikzlibrary{arrows.meta, decorations.pathreplacing}

\usepackage{hyperref}
\usepackage{enumitem}
\renewcommand{\P}{\mathbb{P}}
\renewcommand{\d}{\mathrm{d}}

\makeatletter
\renewenvironment{proof}[1][\proofname] {\par\pushQED{\qed}\normalfont\topsep6\p@\@plus6\p@\relax\trivlist\item[\hskip\labelsep\bfseries#1\@addpunct{.}]\ignorespaces}{\popQED\endtrivlist\@endpefalse}
\makeatother

\usepackage{xr}
\title{
When Hedging Changes the Payoff: Option Replication with Price Impact and Execution Costs
}

\author[$\mathsection$]{\normalsize David Itkin}
\author[$\dag$,$\ddag$]{\normalsize Leandro S\'anchez-Betancourt}

\affil[$\mathsection$]{Department of Statistics, London School of Economics and Political Science\newline
\href{mailto:d.itkin@lse.ac.uk}{d.itkin@lse.ac.uk}}
\affil[$\dag$]{Mathematical Institute, University of Oxford}
\affil[$\ddag$]{Oxford-Man Institute of Quantitative Finance, University of Oxford\newline
\href{mailto:leandro.sanchezbetancourt@maths.ox.ac.uk}{leandro.sanchezbetancourt@maths.ox.ac.uk}}

\newcommand{\R}{\mathbb{R}}
\newcommand{\N}{\mathbb{N}}
\newcommand{\E}{\mathbb{E}}
\newcommand{\Q}{\mathbb{Q}}

\newcommand{\cash}{C}
\newcommand{\val}{\mathfrak{X}}
\newcommand{\valm}{\mathcal{Z}}
\newcommand{\pos}{\Delta}

\newcommand{\manip}{\mathcal{C}}

\newtheorem{theorem}{Theorem}
\newtheorem{lemma}[theorem]{Lemma}
\newtheorem{proposition}[theorem]{Proposition}
\newtheorem{corollary}[theorem]{Corollary}

\newtheorem{assumption}[theorem]{Assumption}

\newtheorem{definition}[theorem]{Definition}

\numberwithin{equation}{section}
\numberwithin{theorem}{section}

\begin{document}
\maketitle

\begin{abstract}
  Hedging a derivative by trading the underlying asset changes the payoff that the hedging intended to replicate. We study this phenomenon when trading generates  price impact and execution costs. 
   In a binomial model, we characterize replication through a fixed point equation. In continuous time, we derive a nonlinear pricing PDE whose implicit terminal condition captures the nature of the moving target problem of the hedger. 
   For monotone convex Lipschitz payoffs (such as   calls and puts) we establish exact replication under midpoint execution costs.  Numerical experiments illustrate: (i) how price impact shifts the effective strike, (ii) the nonlinear dependence of the option price on the number of contracts, (iii) how execution costs smooth terminal holdings, (iv) the extent to which the hedger's own trading can bring an otherwise worthless option into the money, and (v) how the spread and the shape of the limit order book in the options market arise from the price impact and the shape of the limit order book of the underlying. In particular, calibrating the one-period binomial model to trade and quote data on near-expiry SPY call options, we find that our model explains 85\% of the observed option bid-ask spread for near-the-money options.
\end{abstract}


\section{Introduction}

Since the seminal papers of \cite*{black1973pricing} and \cite*{merton1973theory}, derivative pricing has been built around the powerful principle of payoff replication. The classical framework treats the hedger as a price taker where the act of replication does not alter the payoff. 
For a sufficiently large trader, or a sufficiently illiquid market, this approximation of reality breaks down. Hedging trades move market prices, and when a derivative's payoff is determined at those prices, the hedge itself changes the payoff. Replication then becomes a moving target problem in which the trading strategy, the settlement value, and the derivative price must be determined jointly.

The practical relevance of this feedback is highlighted by the 2025 allegations from the Securities and Exchange Board of India that Jane Street,  through the price impact of large volumes of trading,  manipulated the price of an underlying stock index to make large profits from derivatives linked to that index.\footnote{For details see  \cite*{SEBI2025JaneStreet}.}
Moreover, the rapid growth of short-dated options, such as zero-days-to-expiry (0DTE) options,\footnote{0DTEs accounted for 59\% of S\&P 500 index option volume in 2025 and averaged 2.3 million contracts per day; see \cite*{cboe2026volume}.} makes the moving target problem particularly relevant. To handle settlement it is important to distinguish the maturity date,  when the payoff is fixed, from the settlement date, when the payoff is paid. For example, the end-of-month S\&P 500 index options offered by
\citeauthor{cboeSPXsettlement} determine their payoff from closing prices
at maturity and deliver cash on the following business day.
The distinction between maturity ($T$ in our model) and settlement date ($T+$ in our model representing a time shortly after $T$) is a key modelling choice. This allows us to capture the unwinding of positions after the option's expiry, which the hedger undertakes to meet the cash settlement obligation associated with the impacted payoff.

We develop a framework for exact replication that incorporates price impact and execution costs while distinguishing between maturity and settlement dates. The observed non-dividend-paying asset price consists of a fundamental price and the price impact generated by the hedger's trading activity. We account separately for  price impact and execution costs, including the initial acquisition and final liquidation of the hedge. The payoff is fixed at maturity, while the outstanding position is liquidated for subsequent cash settlement. 
Consequently, replication requires a strategy whose liquidation value matches the payoff induced by that same strategy.

We first study a one-period binomial model with general impact and execution costs. This  serves as an important building block for the subsequent multiperiod and continuous-time sections. Here, replication reduces to a scalar fixed point equation, for which we establish existence for Lipschitz payoffs and a sharp sufficient condition for uniqueness.  
This one-period model explains the shape of the limit order book in the options market (including the spread). We show that the nonlinear shape of the book is inherited from the underlying asset (see Figure~\ref{fig:LOB-lambdas}). Regarding the spread, a calibrated model explains 85\% of the observed bid-ask spread in SPY call options trading one day before expiry whose strikes are within 0.5\% of the spot price.\footnote{The one-period binomial model is well-suited for this experiment because widely traded American options coincide with European options in such a simple setup, provided it is not optimal to exercise immediately. }
Under linear impact and execution costs, we then construct a multiperiod replication scheme by backward induction. Despite the path dependence hedging introduces into observed prices, the strategy remains numerically tractable as it can be computed using the recombining fundamental-price tree (even though the impacted-price tree is not necessarily recombining, see Figure~\ref{fig:tree-impact-2}). 

In continuous time, we derive the wealth dynamics as a limit of discrete rebalancing and obtain a new nonlinear pricing partial differential equation (PDE) \eqref{eq:PDE} that is consistent with the formal limit of the binomial model. Its terminal condition \eqref{eq:terminal_condition} is itself an implicit ordinary differential equation (ODE), expressing the dependence of the payoff on the hedge. The verification result of Theorem~\ref{thm:verification} justifies the formal derivation and shows that sufficiently regular solutions to these equations lead to strategies that perfectly replicate the payoff. A central contribution is Theorem~\ref{thm:terminal}, which solves this terminal problem for monotone convex Lipschitz payoffs (such as calls and puts).
For the case of midpoint execution costs, Theorem~\ref{thm:PDE} proves classical well-posedness of this novel pricing PDE and obtains a stochastic representation. Together with the verification theorem, these results yield an admissible strategy that perfectly replicates the claim accounting for price impact and execution costs. This analysis thus provides a procedure for solving the replication problem: one first resolves the terminal feedback condition by solving the implicit ODE, and then determines the dynamic hedge by solving the PDE.

These results lead to financial insights on the pricing and hedging of options. The cost of replicating depends nonlinearly on contract size and for monotone payoffs always exceeds the frictionless cost (see Corollary~\ref{cor:BS_bound} for the discrete-time case and Theorem~\ref{thm:PDE}\ref{item:bbounds} for continuous-time).
For calls, the hedge shifts the effective strike, while execution costs above the midpoint smooth the discontinuity in the terminal holdings that is present in the frictionless case (see Figure~\ref{fig:terminal-condition-call}). Additionally, with explicit solutions for quadratic payoffs and numerical experiments for call options, we illustrate the effects of impact and execution costs on option prices and hedging strategies.

Our work builds on the theory of replication under price impact previously studied by \cite*{frey1998perfect,bank2004hedging,loeper2018option,bouchard2016almost} and by \cite*{becherer2024hedging}. 
Closest to our work is \cite*{bouchard2016almost}, who study a continuous-time superhedging problem with state-dependent permanent price impact and characterize the value function, in the viscosity sense, through a different, quasi-linear pricing PDE. 
Our formulation disentangles fundamental prices from impact. In particular, the fundamental price, which is the price that would prevail in the absence of trading, has its own dynamics. Moreover, we model execution costs and price impact separately.  In contrast to \cite*{bouchard2016almost}, we fix the payoff before the accumulated hedged position is liquidated, which happens shortly after maturity.  
This leads to a novel terminal condition in our setting, as the value and terminal hedge must jointly satisfy an implicit ODE. We develop its construction alongside discrete- and continuous-time replication results. We show how the feedback between hedging and payoff changes replication costs, creates a nonlinear dependence on contract size, and may cause the hedger's own trading to bring an option that would have otherwise expired worthless into the money (see Figure~\ref{fig:S<K<P}). By contrast, in \cite*{bouchard2016almost} a cash-settled hedge is unwound before the payoff is fixed, so with constant impact and constant (arithmetic) volatility the impact washes out at expiry.

The effect that hedging has on the payoff also appears in \cite*{gueant2017option}, where price impact affects the cash payoff within a utility-based pricing and partial-hedging problem. A similar feedback appears in client-dealer transactions. \cite*{baldauf2022principal} study contracts written on outcomes that the dealer might be able to influence through trading (similar to our price impact mechanism), 
while \cite*{muhlekarbe2026prehedging} examine how dealers pre-hedging changes clients' outcomes. These studies highlight the importance of accounting for the effect of hedging.

A complementary literature studies the trade-off between hedging risk and trading costs. \cite*{BertsimasKoganLo2001} develop a recursive approach to minimizing replication error in incomplete markets, while \cite*{KennedyForsythVetzal2009} balance jump risk and transaction costs in a dynamic hedging strategy. Execution costs and market impact also motivate the hedging approaches of \cite*{rogers2010cost,almgren2016option,bank2017hedging,cartea2020hedging}, and the recent signature-based method of \cite*{abijaber2025signature}.\footnote{Related attempts include the duality approach for superlinear frictions in \cite*{GuasoniRasonyi2015},  the robust superreplication in \cite*{OblojWiesel2021}, the incomplete market hedging  in \cite*{bouzianis2020optimal}, and the learning strategy in \cite*{Buehler03082019,shi2023deep}.} Our results complement these approaches by identifying conditions under which hedging risk can be eliminated entirely (even when trading changes the payoff) and characterizing the resulting replication costs and strategies.

The organization of the paper is as follows. Sections~\ref{sec:one_period} and \ref{sec:multi_period} develop binomial models (one-period and multiperiod, respectively) with price impact and execution costs. Section~\ref{sec:cont_time} presents the continuous-time framework and main replication results. Section~\ref{sec:numerics} provides explicit solutions and numerical experiments, and Section~\ref{sec:conclusion} concludes. Appendix~\ref{sec: appendix CalculationSpread} discusses the methodology of our calibration exercise and Appendix~\ref{sec:ec_examples} extends the analysis of Section~\ref{sec:one_period} to put options and square-root impact. All proofs are deferred to the electronic companion for better readability.

\section{One-period model} \label{sec:one_period}
\subsection{Setup}
We begin with the simplest setting of a one-period binomial model. As described in the introduction, a key feature of our model is that it separates the maturity and settlement times of contingent claims. As such, we consider three times: the initial time $t=0$, the \emph{maturity time} $t=T$ and the \emph{settlement time} $t=T+$, which occurs shortly after time $T$. We assume the sample space is binary: $\Omega = \{u,d\}$, and that there is one risk-free and one risky asset available for investment. The risk-free asset has rate of return $r >-1$, while the risky asset has initial price $S_0>0$ and \emph{fundamental} terminal value $S_T(\omega)$ for $\omega \in \Omega$. This is the value of the security that will prevail in the absence of the investor's trading, and we assume it satisfies the standard no-arbitrage condition, $S_T(d) < S_0(1+r) < S_T(u)$. In the frictionless case,  the unique risk-neutral probability $\Q$ is given by \[\Q(\{u\}) = q = 1- \Q(\{d\}) \quad  \text{with} \quad  q = \tfrac{S_0(1+r)-S_T(d)}{S_T(u)-S_T(d)}.\] Between the maturity time $T$ and the settlement time $T+$, there is no additional randomness and, consequently, we have $S_{T+} = S_T$.

We assume the investor's trades generate price impact. Concretely, if the investor enters into a position $\pos \in \R$  at time $t = 0$,  then for each $\omega\in\Omega$ the \emph{observed price} at the maturity time  $t=T$ is given by
\[P^\pos_T(\omega) = S_T(\omega) + I(\pos)\] for some nondecreasing function $I:\R \to \R$ with $I(0) = 0$. For example, we may consider the case of linear impact, $I(x)= \lambda\,x$, or square-root price impact, $I(x) = \lambda\,\mathrm{sign}(x)\,\sqrt{|x|}$, for $\lambda > 0$.  Figure \ref{fig:tree-impact-1} shows a diagram with these variables in the classical one-period tree for the case $\pos>0$.

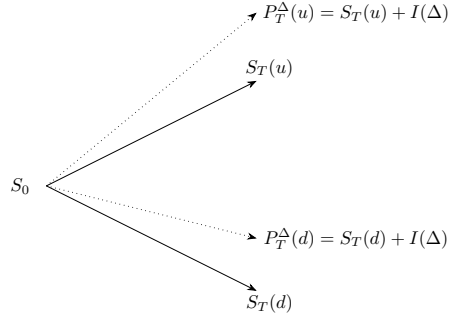
\begin{figure}
\centering
\resizebox{0.4\textwidth}{!}{%
\begin{circuitikz}
\tikzstyle{every node}=[font=\normalsize]

\draw [->, >=Stealth] (5,9) -- (9,11);
\draw [->, >=Stealth] (5,9) -- (9,7);

\draw [->, >=Stealth, dotted] (5,9) -- (9,12.3);
\draw [->, >=Stealth, dotted] (5,9) -- (9,8.);

\node at (4.5,9) {$S_0$};
\node at (9.25,11.25) {$S_T(u)$};
\node at (9.25,6.75) {$S_T(d)$};

\node at (10.9,12.3) {$P^\pos_T(u) = S_T(u) + I(\pos)$};
\node at (10.9,8.) {$P^\pos_T(d) = S_T(d) + I(\pos)$};

\end{circuitikz}
}%
\caption{One-period model price evolution with and without impact.}
\label{fig:tree-impact-1}
\end{figure}

\subsection{Wealth equation and absence of price manipulation}
We now discuss investment and the cost of trading.
To this end, we denote by $\cash_t$ the cash position at time $t \in \{0,T,T+\}$, which between time $0$ and time $T$ grows at rate $r$ with no interest accumulating between time $T$ and time $T+$. We assume that the investor starts with  a cash position of $\cash_0 = \val_0 \in \R$ and no holdings in the risky asset (throughout the paper, the risky asset does not pay dividends).
The \emph{signed execution cost} when trading $x\in\mathbb{R}$ shares will be denoted by $\phi(x)$, so that the cash needed to finance the transaction is $x(S_0 + \phi(x))$. The function $I$ governs how the observed price differs from the fundamental price, whereas $\phi$ governs the effective execution price of the investor's trade. 
 We require $\phi(0) = 0$ and for $\phi$ to be nondecreasing, with the midpoint specification $\phi(x) = \frac{1}{2}I(x)$ being particularly tractable. 

To facilitate the analysis to come, we need to specify how we compare the value of different portfolios. In this paper we will work with the \emph{liquidation value} process $\val$,
which aggregates the investor's cash holdings with the value of their stock position upon immediate liquidation. This gives
\begin{align} 
\val_T & = \cash_T + \pos\,\big(P_T^\pos + \phi(-\pos)\big)  = \Big(\val_0 - \pos\big(S_0 + \phi(\pos)\big)\Big)(1+r) + \pos\,\big(P_T^\pos + \phi(-\pos)\big). \label{eq:value1+}
\end{align} 
In this paper we study the replication problem, so we require the investor to liquidate their holdings for the cash settlement that occurs at time $T+$. Under this convention, at time $T+$ their wealth is held entirely in cash so that 
\begin{equation} \label{eq:terminal_cash_relationship}
\val_{T+} = \cash_{T+} = \val_T.
\end{equation}
This means all trades are round-trips, since the investor starts with no holdings, has $\pos$ shares at time $T$, and then again liquidates their position at time $T+$.

We now derive necessary and sufficient conditions on $\phi$ and $I$ so that the model does not allow for \emph{price manipulation}. The following definition for price manipulation is stated in a general form that encompasses the multiperiod and continuous-time frameworks of Sections~\ref{sec:multi_period} and \ref{sec:cont_time}. In the present one-period model an admissible strategy is simply any $\pos\in\mathbb{R}$, as the information available at time zero is $\mathcal{F}_0=\{\emptyset,\Omega\}$. However, in subsequent sections the definition of admissible strategy will require adaptivity to the information available. 
\begin{definition}
    The market is said to admit price manipulation if starting from initial value $\val_0 = 0$, there exists an admissible trading strategy $\pos$ such that $\E^\Q[\val_T] > 0$. 
\end{definition}

Under $\Q$, the discounted asset price is a martingale so in the absence of price impact one would have $\E^\Q[\val_T] = 0$ for all admissible strategies $\pos$, whenever $\val_0 = 0$. As such, positive expected terminal wealth can only be attributed to the trader's own price impact,  which is price manipulation. A direct calculation shows that $\E^\Q[\val_T] = - \pos(\phi(\pos)(1+r) - \phi(-\pos) - I(\pos))$ for any $\pos \in \R$. 
Hence, in this one-period model, there is no price manipulation if and only if
\begin{equation} \label{eq:no_price_manipulation}
    x\,\manip(x) \geq 0, \quad  \text{for all } x \in \R, \qquad \text{where} \quad  \manip(x) =  \phi(x)(1+r) - \phi(-x) - I(x).
\end{equation}

\subsection{Replication of European claims} \label{sec:1period_replication}

The goal of this section is to investigate the replication of a European option with \emph{payoff function} $V:\R \to \R$ (for example $V(x) = N(x-K)^+$ for $N >0$ contracts and strike $K > 0$). 
Here the realized payoff depends on the observed price at the maturity time; that is, the realized payoff is $V(P^\pos_T)$. We see that, unlike the frictionless case, the payoff the holder of this security receives is influenced by the position $\pos$ taken. In other words, the replication problem is a fixed point problem, since the target replication value changes with the trading strategy. 
 In line with the wealth process definition of the previous section, we will focus on cash settlement, where the payoff value is determined at the maturity time $T$ and the seller delivers the payoff in cash at the settlement date $T+$.

To value the option $V$, we seek to find an initial value $\val_0$ and replicating strategy $\pos$ such that on the settlement date the payoff obligation is met. In our notation, and recalling that $\val_T = \val_{T+}$, this corresponds to requiring
\begin{equation}\label{eq: fixed point equation}
\val_T = V(P_T^\pos), \qquad \Q\text{-a.s.}
\end{equation}
Since the sample space is binary this leads to two equations for $\val_T$; one in the $u$ state and one in the $d$ state.
Substituting \eqref{eq:value1+} and subtracting the two equations leads to
\begin{equation} \label{eq:pos_fixed_point}
    \pos = \frac{V(P_T^\pos(u)) - V(P_T^\pos(d))}{P_T^\pos(u) - P_T^\pos(d)} = \frac{V(S_T(u) + I(\pos)) - V(S_T(d) + I(\pos))}{S_T(u)-S_T(d)},
\end{equation}
where in the final equality we used that $P_T^\pos (u) - P_T^\pos(d) = S_T(u) - S_T(d)$. Unlike the classic frictionless case, the right-hand side also depends on $\pos$ so that this is a fixed point equation. This captures the central feedback mechanism of the model, where the hedge changes the observed terminal price, which in turn changes the payoff that needs to be hedged.

Equation \eqref{eq:pos_fixed_point} exhibits different behaviour depending on the particular impact function $I$ and payoff function $V$, but is independent of the execution cost function $\phi$. 
The following result shows that if the gross payoff function $V$ is Lipschitz continuous then there exists at least one solution to the fixed point equation above. We also provide sufficient conditions for uniqueness. 
\begin{proposition}
    \label{thm: Lipschitz implies existence}
    If $I$ is continuous and $V$ is Lipschitz continuous with constant $L_V$
    then there exists a  solution to the fixed point equation \eqref{eq:pos_fixed_point}, and every solution $\pos^*$ satisfies $|\pos^*| \le L_V$. If, in addition, (i) $V$ is monotone, (ii) $I$ is Lipschitz continuous with constant $L_I$, and (iii) the relationship
    \begin{equation} \label{eq:uniqueness_condition}
        L_VL_I < S_T(u) - S_T(d)
    \end{equation} holds, then the solution to \eqref{eq:pos_fixed_point} is unique. 
\end{proposition}

    The condition \eqref{eq:uniqueness_condition} is tight in the sense that if $L_VL_I = S_T(u) - S_T(d)$ then the solution may not be unique. As a concrete example, consider linear impact $I(x) = \lambda x$ and $N > 0$ call option contracts with strike $K = S_T(u)$; that is, $V(x) = N(x-S_T(u))^+$. Then $L_V = N$, $L_I = \lambda$ and if $N\lambda = S_T(u) - S_T(d)$, then any $\pos^* \in [0,N]$ replicates the payoff $V$. 

We conclude this subsection by noting that there is a unique initial value for the replicating portfolio, $\val_0^*$, corresponding to any solution $\pos^*$.
It turns out that $\val_0^*$ has an interpretable representation, from which we can deduce that the cost of replicating any option with a monotone payoff increases in this setting relative to the frictionless case. 
\begin{proposition} \label{prop:replicating_value} For any solution $\pos^*$ to \eqref{eq:pos_fixed_point}, the associated cost of the replicating portfolio is uniquely determined and given by 
    \begin{equation}\label{eq: general val_0 1 period case}
        \val_0^* = \tfrac{1}{1+r} \big(\mathbb{E}^{\mathbb{Q}}[V(P_T^{\pos^*})] + \pos^*\manip(\pos^*)\big),
    \end{equation}
    where $\manip$ is as in  \eqref{eq:no_price_manipulation}.
    If \eqref{eq:no_price_manipulation} holds, then $\val_0^* \ge \mathbb{E}^{\mathbb{Q}}[V(P_T^{\pos^*})]/(1+r)$. If, in addition,  $V$ is monotone,  then $\val_0^* \geq  \E^\Q[V(S_T)]/(1+r)$, where the right-hand side is the unique no-arbitrage price of the claim $V$ in the frictionless setting. 
\end{proposition}
The representation \eqref{eq: general val_0 1 period case} decomposes the replication cost into the discounted expected payoff at the impacted price plus the  discounted expected round-trip cost of the hedge.
When the fixed point equation \eqref{eq:pos_fixed_point} admits multiple solutions, their corresponding replicating costs may differ, in which case we focus on the one with lowest cost. For monotone payoffs this corresponds to the \emph{minimal-in-absolute-value} position at the maturity time,
\begin{equation}  \label{eq:smallest_solution}
    \pos^*_{\min} \in \arg\min\{|\pos^*|: \pos^* \text{ solves } \eqref{eq:pos_fixed_point}\},
\end{equation}
provided the expected round-trip cost function $x \mapsto x\manip(x)$ is nondecreasing in $|x|$ (when $I$ and $\phi$ are linear functions, this is equivalent to the absence of price manipulation condition \eqref{eq:no_price_manipulation}).
 Indeed, when $V$ is monotone any solution $\pos^*$ to \eqref{eq:pos_fixed_point} shares the same sign $\varepsilon \in \{-1,1\}$ as the slope of $V$. Since $x \mapsto V(S_T(\omega) + I(x))$ is monotone in the same direction as $V$ when $\mathrm{sign}(x) = \varepsilon$, we see from \eqref{eq: general val_0 1 period case} that $\pos^*_{\min}$ leads to the minimal replicating cost. 
For this reason, in the examples below, we focus on identifying $\pos^*_{\min}$.

\subsection{Examples}
We now consider a call option and digital call option with linear price impact. In Appendix~\ref{sec:ec_examples} we additionally study put options and square-root price impact.

\subsubsection{Call options and linear impact} \label{sec:call_linear}
First,  we study the call option  case $V(x) = N(x-K)^+$ for strike price $K  >0$ and number of contracts $N \in \R$. In this subsection we consider the linear impact specification $I(x) = \lambda\,x$ for impact parameter $\lambda > 0$. 
\begin{proposition}\label{prop:call_linear} The fixed point equation \eqref{eq:pos_fixed_point} has at least one solution, and this solution is unique if $\lambda N < S_T(u)-S_T(d)$. The minimal-in-absolute-value solution \eqref{eq:smallest_solution} is given by 
\[ \pos^*_{\min} = 
\begin{cases}
   0,
& K\ge S_T(u), \\
\frac{N(S_T(u)-K)}
      {S_T(u)-S_T(d)-\lambda N},
& S_T(d)+\lambda N\le K<S_T(u), \\
N,
& K < \min\{S_T(d)+\lambda N, S_T(u)\},
\end{cases}
\]
and $\val_0^*$ is given by \eqref{eq: general val_0 1 period case} with $\pos^*_{\min}$ in place of $\pos^*$.
\end{proposition}
The following numerical example illustrates the result. We take $r = 0.01$ and for the execution cost specification we take 
\begin{equation} \label{eq:phi_square_root} 
\phi(x) = \begin{cases}
    \varphi\, \mathrm{sign}(x) \sqrt{|x|}, & |x| \le \frac{4\varphi^2}{\lambda^2}, \\
    \frac{\lambda}{2}x, & \text{otherwise},
\end{cases}
\end{equation} for parameter $\varphi > 0$. Here, the execution cost scales with the square root of the trade size, but for very large positions is chosen to be linear so that the no price manipulation condition \eqref{eq:no_price_manipulation} is globally satisfied. Figure~\ref{fig:LOB-lambdas} illustrates the average per unit execution price of the underlying asset $S$ and the call option $V$ as we vary the total number of contracts $N \in \R$. Observe that the concavity of the execution costs in the underlying market, prescribed by  $\phi$ via \eqref{eq:phi_square_root}, is inherited by the call option price.

\begin{figure}
    \centering
    \includegraphics[width=0.8\linewidth]{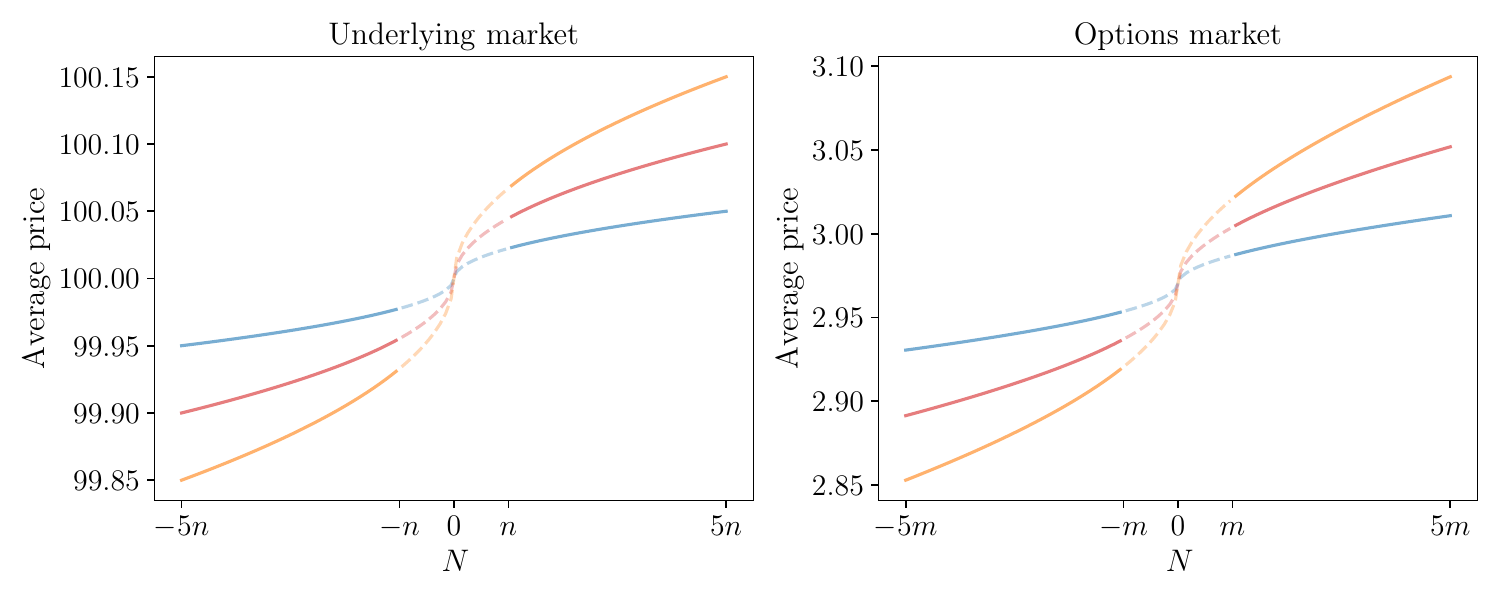}
    \includegraphics[width=0.6\linewidth]{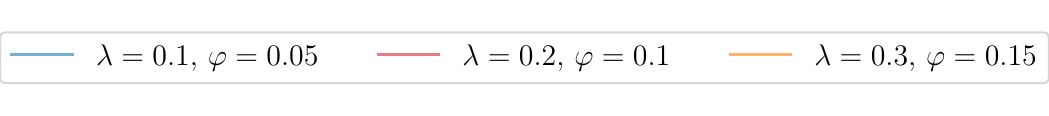}
    \caption{
    The left panel shows the average amount per share that was paid (positive $N$) or received (negative $N$) when executing $N$ shares of the underlying asset. The execution price is given by $100 + \varphi\,\mathrm{sign}(N)\,\sqrt{|N|}$ and $\varphi=0.5\,\lambda$. The right panel shows the prices, using Proposition \ref{prop:call_linear}, for a call option with  $S_0=100$, $S_T(u) = 105$, $S_T(d) = 95$, $r =0.01$, $K= 100$. Here $I(x) = \lambda\,x$ and $\phi(x)$ is given by \eqref{eq:phi_square_root}. 
    In the plots, $n$ and $m$ represent illustrative transaction sizes in the underlying and options markets, respectively. The dashed lines show the finite-size differences between buying and selling these quantities, reflecting a hypothetical spread.  The parameter $n=0.2$ is chosen so that the region plotted satisfies $|x|\leq 4\varphi^2\lambda^{-2}$.
    \label{fig:LOB-lambdas}}
\end{figure}

\subsubsection{Market data}
We employ market data for the SPY exchange-traded fund (ETF) trading on the NYSE Arca, and American call options written on this ETF trading on Cboe. Our data spans January to June 2025. We use the TAQ database via WRDS for the prices and quoted volumes of the underlying ETF, and we use Databento for the prices and quoted volumes in the options market. We calibrate the one-period model above with linear impact and execution costs to this data focusing on options with one day remaining until maturity. We then compare the predictions of our calibrated model to the observed spread. In Appendix \ref{sec: appendix CalculationSpread} we provide details of how we calibrate each of the model parameters.  We summarize the results in Table \ref{tab:aapl_spread_moneyness}.

\begin{table}[htbp]
    \centering
    \caption{SPY call spread results by moneyness band. The explained spread column reports median model implied spread as a percentage of the median observed spread. Data spans  01/2025--06/2025. }
    \label{tab:aapl_spread_moneyness}
    \small
    \begin{tabular}{l|r|r}
        $\big|\frac{S_0}{K} -1\big|\leq \delta$  & $n_{\mathrm{options}}$ 
        & Explained spread (\%) \\
        \hline\hline\vspace{0.1cm}
        $\delta=$0.5\% & 712   & 85.02  \\\vspace{0.1cm}
        $\delta=$1\% & 1420     & 58.53 \\
        \vspace{0.1cm}
        $\delta=$2\% & 2796    & 35.72 \\
        \vspace{0.1cm}
        $\delta=$5\% & 5790    & 9.04  \\
        \vspace{0.1cm}
        All strikes & 15570  & 4.49   \\
        \hline
    \end{tabular}
\end{table}
The high values of explained spread when $\delta$ (the percentage deviation of the spot from the strike) is small suggest that the one-period model is able to capture most of the spread observed in near-the-money options close to expiry. On the other hand, bid-ask spreads for deep out-of-the-money options are less well explained by the one-period model. One explanation for this is that close to expiry deep out/in-the-money options are affected by illiquidity (leading to wider observed spreads than what the model predicts), which is beyond the scope of this one-period model.

Next, we use our model to compute the number of contracts $N$ for which the spread produced by our model matches the observed spread. Similarly, we also compute the  execution cost parameter $\varphi$ for which the spread produced by our model matches the observed spread. Both of these calculations are carried out leaving the remaining model parameters fixed at their calibrated value.  The results are summarized in Table~\ref{tab:implied}.

\begin{table}[htbp]
    \centering
    \caption{Median implied number of contracts compared to median volume at best quotes (left). Median model implied $\varphi$, median model calibrated $\varphi$ and their ratio (right). Data spans 01/2025--06/2025. }
    \label{tab:aapl_spread_implied_N}
    \small
    \begin{tabular}{l|rrr|rrr}
        $\big|\frac{S_0}{K} -1\big|\leq \delta$  & Implied $N$
        & Bid $N$ & Ask $N$  
        & Implied $\varphi$
        & Calibrated $\varphi$ & Explained (\%)\\
        \hline\hline\vspace{0.1cm}
        $\delta=$0.5\% & 32.59   & 50.0 & 47.0 & 7.78e--06 & 1.26e--05 & 61.65  \\\vspace{0.1cm}
        $\delta=$1\% & 42.57   & 50.0 & 50.0 & 9.24e--06 & 1.26e--05 & 73.21  \\
        \vspace{0.1cm}
        $\delta=$2\% & 50.68   & 50.0 & 50.0 & 1.14e--05 & 1.17e--05 & 97.70  \\
        \vspace{0.1cm}
        $\delta=$5\% & 94.78   & 39.0 & 40.0 & 3.45e--05 & 1.26e--05 & 274.63  \\
        \vspace{0.1cm}
        All strikes & 326.44   & 15.0 & 14.0 & 2.89e--04 & 1.29e--05 & 2241.66 \\
        \hline
    \end{tabular}
    \label{tab:implied}
\end{table}

We see that the model fit is best when the spot is near the strike. The table shows that the number of contracts $N$ that produces the observed spread is consistent with the typical observed volumes at best quotes up to $\delta=2\%$.
Moreover, the implied execution cost parameter $\varphi$ exhibits a similar behaviour to that of the implied volatility; that is, for out-of-the-money options, the implied execution cost parameter is much higher than those obtained from near-the-money options. Again, the implied values fit well for $\delta\leq 2\%$.

\subsubsection{Digital call option} \label{sec:digital}

Here we work with linear price impact, $I(x) = \lambda x$, and study the payoff $V(x) = N1_{\{x \ge K\}}$ for strike $K > 0$ corresponding to $N$ units of a digital call option. This payoff is discontinuous, so the existence result of Proposition~\ref{thm: Lipschitz implies existence} does not apply. In fact, this claim is not always replicable.
\begin{proposition} \label{prop:digital_call}
There exists a solution to  \eqref{eq:pos_fixed_point} if and only if 
\begin{equation} \label{eq:bad_strikes} K \not \in J:= \begin{cases} 
\Big(S_T(d),\, \min\big\{S_T(d) + \frac{N\lambda}{S_T(u)-S_T(d)},\,S_T(u)\big\}\Big], & \text{if } N \ge 0, \\
\Big(\max\big\{S_T(d),\,S_T(u) + \frac{N\lambda}{S_T(u)-S_T(d)}\big\},\, S_T(u)\Big], & \text{if } N < 0.
\end{cases} 
\end{equation} Outside of this range of strikes the minimal-in-absolute-value solution is given by
\begin{equation} \label{eq:pos_digital} \pos^*_{\min} = \begin{cases}
    0, & \text{if }  K > S_T(u) \text{ or } K \leq S_T(d), \\
    \frac{N}{S_T(u)-S_T(d)}, & \text{if } K \in (S_T(d),S_T(u)] \setminus J,
\end{cases} 
\end{equation}
and $\val_0^*$ is given by \eqref{eq: general val_0 1 period case} with $\pos^*_{\min}$ in place of $\pos^*$.
 \end{proposition}

Proposition~\ref{prop:digital_call} shows that not all contingent claims can be replicated; that is, the model is incomplete. Indeed, for any $\lambda > 0$ and target number of contracts $N \ne 0$, there exists a range of strikes $J$ where no solution to \eqref{eq:pos_fixed_point} exists. The failure in replication is driven by the payoff discontinuity. 
For $N >0$ and strikes $K$ just above $S_T(d)$, the trade attempting to hedge the claim can itself push the observed price $P^\pos_T(d)$ above the strike causing the payoff of the option to jump from zero to $N$ in state $d$. As such, the replicating attempt itself causes the hedger to miss the mark.

\section{Multiperiod model} \label{sec:multi_period} 
\subsection{Setup} 
Here we tackle the multiperiod case. To provide a single unifying setup for multiperiod binomial models of arbitrary size, we will use the continuous time index set $ \mathcal{T} =  [0,T]\cup \{T+\}$ for a finite time horizon $T > 0$. As before, the time $T+$ is a discrete time point shortly after time $T$, which is the settlement time of contingent claims we consider. For an integer $M  > 0$ all market activity happens on the grid of times $\mathbb{T}:= \{t_m\}_{m =0,\dots,M}$, where $t_m = mh$ for $h = T/M$. We work on a probability space $(\Omega,\mathcal{F},(\mathcal{F}_t)_{t \in \mathcal{T}},\P)$, where the sample space is the discrete binomial model $\Omega = \{u,d\}^M$ with $0<d <u$.\footnote{With a slight abuse of notation, we use $u$ and $d$ both as the up and down outcomes, as well as the gross returns of the risky asset.} We will canonically write $\omega = \omega_1\dots\omega_M$, where $\omega_m \in \{u,d\}$ for each $m$. Additionally, with a superscript, we denote partially observed outcomes, $\omega^m = \omega_1\,\dots\,\omega_m$. The filtration satisfies  $\mathcal{F}_t = \mathcal{F}_{t_m}$ for $t_m  \le t < t_{m+1}$, where  \[\mathcal{F}_0 = \{\emptyset,\Omega\} \qquad \text{and} \qquad \mathcal{F}_{t_m} = \sigma(\{\omega \in \Omega: \omega^m = a\}: a \in \{u,d\}^m), \quad \text{for }m=1,\dots,M-1.\]
 All of the randomness is revealed by the maturity date leading to $\mathcal{F} = \mathcal{F}_T = \mathcal{F}_{T+} = \mathcal{P}(\Omega)$. Finally, as is standard, the physical probability measure $\P$ satisfies $\P(\{\omega\}) >0$ for each $\omega \in \Omega$, but will play a minor role in the analysis to come. 

The model has one risk-free and one risky asset available for investment. The risk-free asset pays interest rate $rh > -1$ over any time interval $[t_m,t_{m+1})$.  Consistent with the one-period model, we view the settlement date $T+$ as occurring very soon after the maturity time $T$, so that no interest accumulates between time $T$ and $T+$. The fundamental price process of the risky asset is denoted by $(S_t)_{t \in \mathcal{T}}$. It starts at some $S_0 > 0$, is piecewise constant, $(\mathcal{F}_t)_{t \in \mathcal{T}}$-adapted, right-continuous and jumps only at times $t_m \in \mathbb{T}\setminus\{0\}$ according to the relationship
\[S_{t_m}(\omega) = u\times S_{t_{m-1}}(\omega)1_{\{\omega_m = u\}} + d\times S_{t_{m-1}}(\omega)1_{\{\omega_m = d\}}, \quad \text{and} \quad S_{T+}(\omega) = S_T(\omega) \qquad \text{for all } \omega \in \Omega.
\] We denote by $\mathcal{S}_m$ the set of size $m+1$ consisting of all possible values that $S_{t_m}$ can take.
This setup encodes the frictionless binomial model with $M$ time steps, so we impose the standard no-arbitrage condition $d < 1+rh < u$. The unique risk-neutral measure is then given for all $\omega \in \Omega$ by 
\[\Q(\{\omega\}) =q^{n(\omega)}(1-q)^{M-n(\omega)}, \qquad \text{where} \quad q = \tfrac{1+rh-d}{u-d} \quad \text{and} \quad n(\omega) = {\textstyle \sum}_{m=1}^M 1_{\{\omega_m = u\}}.\]
	
Next, we introduce the price impact model, which generalizes the $M=1$ setup of Section~\ref{sec:one_period}. For simplicity, we restrict to linear permanent price impact and execution cost functions
\[I(x) = \lambda x, \qquad \phi(x) = \varphi x,\]
 for some parameters $\varphi,\lambda > 0$.\footnote{Restricting to linear permanent price impact is more than just a convenient choice, as the celebrated results of \citet{huberman2004price} and \citet{gatheral2010no} stipulate that permanent price impact must be linear for models with unlimited trading periods to be free of price manipulation.}  An admissible trading strategy $(\pos_{t})_{t\in\mathcal{T}}$ is assumed to be an adapted,  left-continuous, piecewise constant process that jumps only at trading times and additionally satisfies $\pos_0 = \pos_{T+} = 0$. Under the left-continuity convention, $\pos_{t_m}$ is the position the investor holds at time $t_m$, while the position held over the time period $(t_m,t_{m+1}]$ incorporating the new information revealed at time $t_m$ is $\pos_{t_m+} = \pos_{t_{m+1}}$, where we write $t-$ and $t+$ for the left and right limits of a time $t \in [0,T]$. We denote by $(P^\pos_t)_{t \in \mathcal{T}}$ the observed price process when utilizing the strategy $\pos$, and it is assumed to satisfy
 \begin{equation} \label{eq:observed_price}
    P^\pos_t = S_t + \lambda \pos_t, \qquad t \in \mathcal{T}.
\end{equation}
Although the binomial tree for $S$ is recombining, the same cannot be said for $P^\pos$; see Figure~\ref{fig:tree-impact-2} for a two-period illustration.
Additionally, due to our right-continuity convention on $S$, and left-continuity on $\pos$, the affected price process $P^\pos$ has left and right limits, but is not necessarily left- or right-continuous.

\begin{figure}
\centering
\resizebox{0.5\textwidth}{!}{%
\begin{circuitikz}
\tikzset{every node/.style={font=\normalsize}}

\coordinate (S0)   at (0,0);

\coordinate (S1u)  at (3.5,1.8);
\coordinate (S1d)  at (3.5,-1.8);

\coordinate (S2uu) at (7,3.6);
\coordinate (S2m)  at (7,0);
\coordinate (S2dd) at (7,-3.6);

\coordinate (P1u)  at (3.5,2.8);
\coordinate (P1d)  at (3.5,-0.8);

\coordinate (P2uu) at (7,5.1);
\coordinate (P2ud) at (7,1.5);
\coordinate (P2du) at (7,0.7);
\coordinate (P2dd) at (7,-2.9);

\draw[->, >=Stealth] (S0) -- (S1u);
\draw[->, >=Stealth] (S0) -- (S1d);

\draw[->, >=Stealth] (S1u) -- (S2uu);
\draw[->, >=Stealth] (S1u) -- (S2m);

\draw[->, >=Stealth] (S1d) -- (S2m);
\draw[->, >=Stealth] (S1d) -- (S2dd);

\draw[->, >=Stealth, dotted] (S0) -- (P1u);
\draw[->, >=Stealth, dotted] (S0) -- (P1d);

\draw[->, >=Stealth, dotted] (P1u) -- (P2uu);
\draw[->, >=Stealth, dotted] (P1u) -- (P2ud);

\draw[->, >=Stealth, dotted] (P1d) -- (P2du);
\draw[->, >=Stealth, dotted] (P1d) -- (P2dd);

\node[left] at (S0) {$S_{t_0}$};

\node[below] at (3.5,1.7) {$S_{t_1}(u)$};
\node[below] at (3.5,-1.9) {$S_{t_1}(d)$};

\node[right] at (S2uu) {$S_{t_2}(uu)$};
\node[right] at (S2m)  {$S_{t_2}(ud)=S_{t_2}(du)$};
\node[right] at (S2dd) {$S_{t_2}(dd)$};

\node[above] at (3.5,2.9)
{$P^\pos_{t_1}(u)$};

\node[above] at (3.5,-0.7)
{$P^\pos_{t_1}(d)$};

\node[right] at (P2uu) {$P^\pos_{t_2}(uu)$};
\node[right] at (P2ud) {$P^\pos_{t_2}(ud)$};
\node[right] at (P2du) {$P^\pos_{t_2}(du)$};
\node[right] at (P2dd) {$P^\pos_{t_2}(dd)$};

\node[below] at (0,-4.8)   {$t_0$};
\node[below] at (3.5,-4.8) {$t_1$};
\node[below] at (7,-4.8)   {$t_2$};

\end{circuitikz}
}%
\caption{Solid lines represent the recombining fundamental price process $S$, while dotted lines represent a possible affected price process $P^\pos$, which may not recombine because the investor's inventory $\pos$ may depend on the path taken. }
\label{fig:tree-impact-2}
\end{figure}
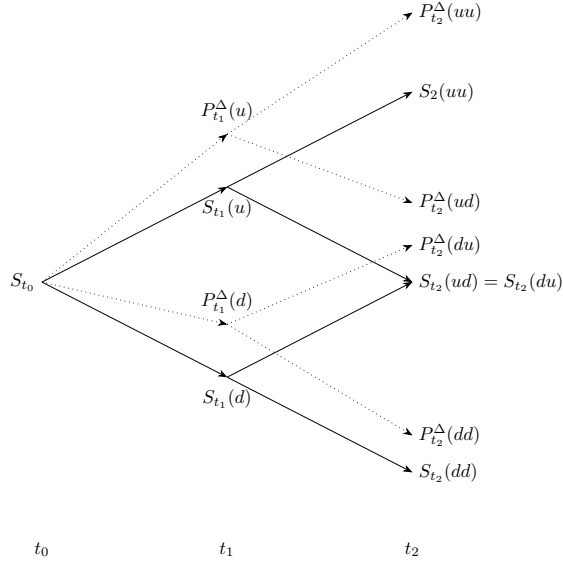

The cash process $(\cash_t)_{t \in \mathcal{T}}$ starts at some given value $\cash_{0}=\val_0$, is constant on the intervals $(t_m,t_{m+1})$,
and evolves according to
\begin{align*}  
\cash_{t_m} & = \cash_{t_m-}(1+rh), & m=1,\dots,M, \\
\cash_{t_m+} & = \cash_{t_m} - (\pos_{t_{m+1}} - \pos_{t_m})(P_{t_m}^\pos+\varphi(\pos_{t_{m+1}} - \pos_{t_m})), & m=0,\dots,M-1, \\
\cash_{T+} &  = \cash_T + \pos_T(P_T^\pos - \varphi \pos_T),
\end{align*}
as depicted in Figure~\ref{fig:cash}.
Under this convention, the interest on the bank balance is paid prior to the rebalancing of the position.
\begin{figure}
    \centering
    \begin{tikzpicture}[
    >=Latex,
    dot/.style={circle, fill=black, inner sep=1.7pt},
    opendot/.style={circle, draw=black, fill=white, inner sep=1.7pt},
    lab/.style={font=\small, align=center}
]
\draw[->, thick] (-0.5,0) -- (6.4,0) node[right] {time};
\draw[->, thick] (0,-0.3) -- (0,3.6);
\coordinate (leftlim)  at (3.0,1.0);
\coordinate (leftlimLabel)  at (3.1,0.7);
\coordinate (point)    at (3.0,2.0);
\coordinate (pointLabel)    at (2.9,1.86);
\coordinate (rightlim) at (3.0,3.0);
\coordinate (rightlimLabel)  at (4.0,3.3);
\draw[ultra thick] (0.4,1.0) -- (leftlim);
\draw[ultra thick] (rightlim) -- (6.0,3.0);
\draw[dashed] (3.0,0) -- (3.0,3.1);
\node[below] at (3.0,0) {$t_m$};
\node[opendot] at (leftlim) {};
\node[dot]     at (point) {};
\node[opendot] at (rightlim) {};
\node[left=2pt] at (leftlimLabel) {$\cash_{t_m-}$};
\node[right=4pt] at (pointLabel) {$\cash_{t_m}=\cash_{t_m-}\times (1+r\,h)$};
\node[left=-213pt] at (rightlimLabel){$\cash_{t_m+} = \cash_{t_m} -  (\pos_{t_{m+1}}-\pos_{t_m})(P^\pos_{t_m} + \varphi\,(\pos_{t_{m+1}}-\pos_{t_m}))$};
\draw[decorate, decoration={brace, amplitude=5pt}, thick]
    (2.75,1.25) -- (2.75,1.95)
    node[midway, left=4pt, lab]
    {interest payment};
\draw[decorate, decoration={brace, amplitude=5pt, mirror}, thick]
    (3.25,2.15) -- (3.25,2.85)
    node[midway, right=8pt, lab]
    {rebalancing};
\end{tikzpicture}
    \caption{Open circles represent limiting values and the filled circle is the value assigned at time $t_m$.}
    \label{fig:cash}
\end{figure}
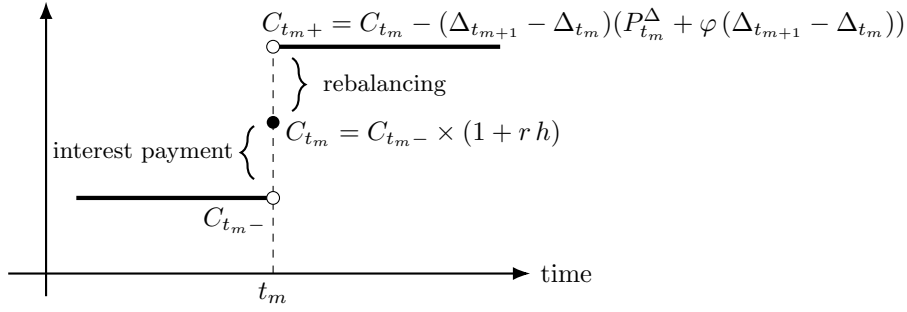
\noindent The liquidation value $(\val_t)_{t \in \mathcal{T}}$ of the position is then defined to be 
\begin{equation} \label{eq:value_process_defn}
    \val_t  = \cash_t + \pos_t(P^\pos_t-\varphi \pos_t), \qquad t \in \mathcal{T},
\end{equation}
where the second term encodes the amount of cash the investor would receive if they liquidate their position $\pos_t$ at time $t$.
In particular, we again have the relationship \eqref{eq:terminal_cash_relationship} since no interest accrues and no additional randomness occurs between times $T$ and $T+$.
We close out this subsection by establishing conditions on the parameters that rule out price manipulation.
\begin{proposition}\label{prop:linear_no_manip}
Suppose that the parameters satisfy
\begin{equation}\label{eq:linear_price_manip} r \ge 0 \qquad \text{and} \qquad 
  2\varphi \ge \lambda .
\end{equation}
Then the market does not admit price manipulation.
\end{proposition}
In the present discrete-time setup, the conditions \eqref{eq:linear_price_manip} are merely sufficient for the absence of price manipulation, but in the continuous-time limit they will also be necessary; see Proposition~\ref{prop:price_manip_cont_time} of Section~\ref{sec:cont_time}.

\subsection{Replication of European claims}
We now seek to replicate the payoff of a European claim. As in the one-period model, we suppose that a payoff function $V:\R \to \R$ is given and we seek to find an admissible strategy $(\pos^*_t)_{t \in \mathcal{T}}$ so that its associated liquidation value process satisfies $\val^*_T = V(P_T^{\pos^*})$. We proceed by backward induction, noting that the base case corresponds precisely to the one-period model considered in Section~\ref{sec:one_period}. Indeed, at each penultimate node the one-period replication equation determines the terminal inventory $\pos^*_{t_M}$ held over the final interval $(t_{M-1},t_M]$, as well as the corresponding terminal wealth $\val^*_{t_M}$. To carry out the inductive step we first let $m \in \{1,\dots,M\}$ and a path $\omega^{m-1}$ observed up to the time $t_{m-1}$ be given. By evaluating \eqref{eq:value_process_defn} at the outcome $\omega_{m} = d$ and subtracting it from the same equation evaluated at $\omega_{m} = u$, we obtain 
\begin{equation} \label{eq:multiperiod_position_equation} 
\pos^*_{t_{m}}(\omega^{m-1}) = \frac{\val^*_{t_{m}}(\omega^{m-1}u) - \val^*_{t_{m}}(\omega^{m-1}d)}{S_{t_{m}}(\omega^{m-1}u) - S_{t_{m}}(\omega^{m-1}d)}.
\end{equation}
As in \eqref{eq:pos_fixed_point}, this is again a fixed point equation since $\val^*_{t_{m}}(\omega^{m})$ is typically a function of $\pos^*_{t_{m}}(\omega^{m-1})$. By the inductive hypothesis we assume that a solution $\pos_{t_{m}}^*(\omega^{m-1})$ exists, we substitute \eqref{eq:multiperiod_position_equation} back into \eqref{eq:value_process_defn} and after simplification (see the proof of Theorem~\ref{thm:binomial_replication} in \ref{sec:multi_period_proofs} for the detailed derivation) arrive at an expression for the liquidation value at the previous time, 
\begin{equation} \label{eq:cash_value_discrete}
\val_{t_{m-1}}^* = \frac{1}{1+rh}\E^\Q_{t_{m-1}}[\val^*_{t_{m}}] + \pos^*_{t_{m}}\bigg(\frac{2\varphi-\lambda +\varphi rh}{1+rh} \pos^*_{t_{m}} -(2\varphi-\lambda) \pos^*_{t_{m-1}}\bigg).
\end{equation}
Substituting \eqref{eq:cash_value_discrete} into \eqref{eq:multiperiod_position_equation}, with $t_{m-1}$ in place of $t_m$  leads to a linear equation for $\pos^*_{t_{m-1}}(\omega^{m-2})$, which has a unique solution if and only if  
\begin{equation} \label{eq:nonzero_denom}
S_{t_{m-1}}(\omega^{m-2}u)  - S_{t_{m-1}}(\omega^{m-2}d) + (2\varphi-\lambda) \big(\pos^*_{t_m}(\omega^{m-2}u) -\pos^*_{t_m}(\omega^{m-2}d)\big) \ne 0.
\end{equation} In this case the solution is explicitly given by $\pos^*_{t_{m-1}}= f_{m-1}(S_{t_{m-2}})$, where $f_{m-1}$ is defined in Theorem~\ref{thm:binomial_replication} below. With $\pos^*_{t_{m-1}}$ determined we see that \eqref{eq:cash_value_discrete} determines the liquidation value process at the time $t_{m-1}$, completing the inductive step. As such, backward induction leads to replication if (i) the initial fixed point equation \eqref{eq:pos_fixed_point} has a solution at each penultimate outcome $\omega^{M-1}$ and (ii) the positions obtained from the backward propagation through the binomial tree satisfy \eqref{eq:nonzero_denom}.
This analysis leads us to the following result.
\begin{theorem} \label{thm:binomial_replication} Suppose the absence of price manipulation condition \eqref{eq:linear_price_manip} holds and that a Lipschitz continuous payoff function $V:\R \to \R$ is given. For each $s \in \mathcal{S}_{M-1}$ the fixed point equation
\[x = \frac{V( su + \lambda x) - V( sd + \lambda x)}{ su -  sd}
\]
admits at least one solution. We select any such solution, denote it by $f_M(s)$ and define the function $g_M: \mathcal{S}_{M-1} \times \mathcal{S}_M \to \R$ via $g_M(s,s') = V(s' + \lambda f_M(s))$.
Next, we define $f_m:\mathcal{S}_{m-1} \to \R$ and $g_m:\mathcal{S}_{m-1} \times \mathcal{S}_{m} \to \R$ for $m=M-1,\dots,1$ recursively as
\begin{align} f_m(s) & = \frac{\E^\Q[g_{m+1}(su,S_{t_{m+1}})|S_{t_m}=su] - \E^\Q[g_{m+1}(sd,S_{t_{m+1}})|S_{t_m}=sd] }{(1+rh)(su - sd + (2\varphi-\lambda)(f_{m+1}(su) - f_{m+1}(sd)))} \nonumber \\
& \qquad  + \frac{(2\varphi - \lambda + \varphi rh)(f^2_{m+1}(su) - f^2_{m+1}(sd))}{(1+rh)(su - sd + (2\varphi-\lambda)(f_{m+1}(su) - f_{m+1}(sd)))}, \label{eq:fm}
\\
g_m(s, s') & = \frac{\E^\Q[g_{m+1}(s',S_{t_{m+1}})|S_{t_m} = s']}{1+rh} + f_{m+1}(s')\bigg(\frac{2\varphi - \lambda + \varphi rh}{1+rh}f_{m+1}(s') - (2\varphi - \lambda)f_m(s)\bigg). \label{eq:gm}
\end{align}
Finally, we set $f_0(S_0) = 0 $ and $g_0(S_0) = \frac{1}{1+rh}(\E^\Q[g_1(S_0,S_{t_1})] + (2\varphi - \lambda + \varphi rh)f_{1}(S_0)^2).$ If 
\begin{equation}\label{eq:denominator_finite}
    su-sd + (2\varphi-\lambda)\big(f_{m+1}(su) - f_{m+1}(sd)\big) \ne  0, \qquad \text{for all } m = 1,\dots,M-1, \quad s \in \mathcal{S}_{m-1}
\end{equation}
holds then \eqref{eq:fm} is well-defined and the trading strategy  $\pos^*_{t_m} = f_m(S_{t_{m-1}})$ for $m=1,\dots,M$ together with $\pos_0 = \pos_{T+} = 0$ replicates the payoff $V$ with $\val^*_0 = g_0(S_0)$ and $\val^*_{t_m} = g_m(S_{t_{m-1}},S_{t_m})$ for $m =1,\dots,M$, as its associated liquidation value process.
\end{theorem}
The nondegeneracy condition \eqref{eq:denominator_finite} is always satisfied for the midpoint cost regime $2\varphi = \lambda$. 
The upshot of this result is that despite a generic trading strategy $\pos$ introducing path dependence of the realized payoff $V(P_T^{\pos})$ (see Figure~\ref{fig:tree-impact-2}),
the replicating strategy $\pos^*$ and liquidation value process $\val^*$ do not inherit this complexity. Indeed, the replicating strategy at a given node depends only on the fundamental price at that node, while the associated liquidation value depends only on the current and previous fundamental price values. In the midpoint case $2\varphi = \lambda$, the dependence on the previous price vanishes and the situation resembles the frictionless case, albeit with the more involved expressions \eqref{eq:fm} and \eqref{eq:gm}.  This structure provides significant computational savings when it comes to numerically computing $(\pos^*,\val^*)$ relative to a truly path dependent situation. Indeed, the replication problem can be solved on the recombining fundamental binomial tree, avoiding the need to use the exponentially growing observed price process tree.

As a consequence of Theorem~\ref{thm:binomial_replication} we obtain the following corollary, which shows that replication of monotone payoffs is always more expensive in the present setting than the frictionless setting.
\begin{corollary} \label{cor:BS_bound}
Let $V$ be as in Theorem~\ref{thm:binomial_replication} and let $(\pos^*,\val^*)$ be an associated
replicating pair.  Setting $t_{M+1} = T+$, we then have
\begin{equation} \label{eq:multiperiod_value_decomposition}
  \val^*_0 = \E^\Q\bigg[\frac{V(P^{\pos^*}_T)}{(1+rh)^M}\bigg]
  + \E^\Q\bigg[\frac{2\varphi-\lambda}{2}\sum_{m=0}^{M}\frac{
      (\pos^*_{t_{m+1}}-\pos^*_{t_m})^2}{(1+rh)^m}
  + \frac{\lambda r h}{2}\sum_{m=1}^{M}\frac{(\pos^*_{t_m})^2}{(1+rh)^m}\bigg].
\end{equation}
In particular, under the no price manipulation condition \eqref{eq:linear_price_manip} we have
$\val^*_0 \ge \E^\Q[V(P^{\pos^*}_T)]/(1+rh)^M$. If, in addition, $V$ is monotone, then
$
  \val^*_0 \ge\E^\Q[V(S_T)]/(1+rh)^M$, where 
the right-hand side is the unique no-arbitrage price of $V$ in the frictionless $M$-period model.
\end{corollary} 

\subsection{Heuristics for the limit of the replicating strategy $\pos^*$}\label{sec:heuristics}
We will now formally study the behaviour of the strategy $\pos^*$ given by Theorem~\ref{thm:binomial_replication} in the continuous-time limit as $h \to 0$. To this end, we assume the standard scaling $u =  e^{\sigma \sqrt{h}}$ and $d  = e^{-\sigma \sqrt{h}}$ for some constant $\sigma > 0$. Next, we introduce the \emph{midpoint liquidation value process} 
\begin{equation}
\label{eq:modified_cash_value}
  \valm_t^* := \val^*_{t} + \tfrac{2\varphi-\lambda}{2}\,(\pos^*_t)^2 ,\qquad  t \in \mathcal{T},
\end{equation}
which admits a clear financial interpretation. Indeed, substituting the liquidation value definition
\eqref{eq:value_process_defn} and the observed price \eqref{eq:observed_price} for the value process $\val^*$ in \eqref{eq:modified_cash_value}, the
execution cost parameter $\varphi$ cancels out and we obtain the representation
  $\valm_t^* = \cash_t + \pos_{t}^*(S_{t} + \tfrac{\lambda}{2}\pos_{t}^*). $
It follows that $\valm^*_t$ marks the inventory neither at the fundamental price $S_t$, nor at the observed
price $P^{\pos^*}_t = S_t + \lambda\pos^*_t$, but at the midpoint of the fundamental and observed prices. The
gap $\valm^*_t - \val^*_t = \frac{2\varphi-\lambda}{2}(\pos_t^*)^2$ measures how far the realizable
liquidation value falls short of this average-price mark. This quantity is nonnegative 
under the no price manipulation condition $2\varphi\ge\lambda$, and it vanishes in the
midpoint case $2\varphi=\lambda$.

We now connect the strategy $\pos^*$ to the midpoint liquidation value process $\valm^*$. To this end, for any $\mathcal{F}_{t_m}$-measurable random variable $X$ we write $\delta [X]$ for the $\mathcal{F}_{t_{m-1}}$-measurable random variable  $\delta [X](\omega^{m-1}) = X(\omega^{m-1} u) - X(\omega^{m-1} d)$.
From the definition of $f_m$ in \eqref{eq:fm} and the fact that $\val^*_{t_m} = g_m(S_{t_{m-1}},S_{t_m})$, we have
\begin{align*} 
\pos^*_{t_m}  = f_m(S_{t_{m-1}}) & = \frac{\delta[\E^\Q_{t_m}[\val^*_{t_{m+1}}] + (2\varphi-\lambda)(\pos^*_{t_{m+1}})^2]}{\delta [S_{t_m} + (2\varphi-\lambda)\pos^*_{t_{m+1}}]} + O(h) \\
& = \frac{\delta[\E^\Q_{t_m}[\val^*_{t_{m+1}}] + (2\varphi-\lambda)(\pos^*_{t_{m+1}})^2]}{\delta [S_{t_m}]} \times \frac{\delta [S_{t_m}]}{\delta [S_{t_m} + (2\varphi-\lambda)\pos^*_{t_{m+1}}]} + O(h).
\end{align*}
In the limit one expects, based on the Markovian structure observed in Theorem~\ref{thm:binomial_replication}, that $\val^*_t$ converges to $v(t,S_t)$ and $\pos^*_t$ converges to $w(t,S_t)$ for some functions $v$ and $w$ as $h \to 0$. This suggests the limiting relationship 
\[w(t,x) = \frac{\partial_x (v(t,x) + (2\varphi-\lambda) w(t,x)^2)}{\partial_x(x+(2\varphi-\lambda)w(t,x))} = \frac{\partial_x v(t,x) + 2(2\varphi-\lambda)w(t,x) \partial_{x}w(t,x)}{1 + (2\varphi-\lambda)\partial_{x}w(t,x)}\]
for $ (t,x) \in [0,T] \times (0,\infty)$. Clearing the denominator and collecting like terms leads to 
\[w(t,x) = \partial_x v(t,x) + (2\varphi-\lambda)w(t,x)\partial_xw(t,x) = \partial_x\big(v(t,x)+ \tfrac{2\varphi-\lambda}{2}w^2(t,x)\big).\]
From \eqref{eq:modified_cash_value}, the function $u:= v + \frac{2\varphi-\lambda}{2}w^2$ evaluated at $(t,S_t)$ is the formal limit of $\valm^*_t$ as $h \to 0$. This suggests the relationship $\pos^*_t = \partial_x u(t,S_t)$ and $\valm^*_t = u(t,S_t)$, which represents the replicating strategy in terms of the midpoint liquidation value process $\valm^*_t$. 
This heuristic motivates the approach we take in the next section, where we derive a pricing PDE for the function $u$.

\section{Continuous-time model} \label{sec:cont_time}
In this section we directly study the replication problem with price impact in continuous-time. We start by deriving the continuous-time liquidation value dynamics and the pricing PDE \eqref{eq:PDE} together with an implicit terminal condition \eqref{eq:terminal_condition}, consistent with the formal limit of Section~\ref{sec:heuristics}. We then analyze this terminal condition and PDE, providing a stochastic representation for the solution and establishing classical solutions and corresponding replicating strategies in the midpoint regime $2\varphi = \lambda$. 

\subsection{Setup and preliminaries}
 We fix a maturity time $T > 0$, define the time set $\mathcal{T} = [0,T] \cup \{T+\}$, and work on a filtered probability space $(\Omega, \mathcal{F},(\mathcal{F}_t)_{t \in \mathcal{T}}, \P)$ satisfying the usual hypotheses. We assume that the fundamental price satisfies
\begin{equation} \label{eq:local_vol}
\d S_t = rS_t \d t + \sigma(t,S_t)S_t \d W_t, \quad t \in [0,T], \qquad \text{and} \quad  S_{T+} = S_T,
\end{equation}
for some initial value $S_0 > 0$
under the risk-neutral probability $\Q$, where $W$ is a standard Brownian motion under $\Q$, $r \in \R$ is the interest rate parameter and $\sigma:[0,T] \times (0,\infty) \to (0,\infty)$ specifies a local volatility model. We make the following technical assumption on $\sigma$.
	\begin{assumption} \label{ass:sigma} Set $y = \log x$ and $\widetilde \sigma(t,y) = \sigma(t,e^y)$. We assume there exist constants $ 0 < \underline \sigma \le \overline \sigma < \infty$, $C_\sigma > 0$ and $\beta \in (0,1)$ such that for all $(t,y) \in [0,T]\times \R$, 
	\[\underline \sigma \le \widetilde \sigma(t,y) \le \overline \sigma, \qquad |\partial_y \widetilde \sigma(t,y)|  \le C_\sigma, \qquad |\widetilde \sigma(t,y) - \widetilde \sigma(s,y)| \le C_\sigma|t-s|^{\beta/2}.\]
	\end{assumption}
    We now define trading strategies, keeping in mind that our goal is to study the pricing and hedging of European claims with cash settlement under linear price impact with linear execution costs.
\begin{definition}
An admissible trading strategy 
$(\pos_t)_{t \in \mathcal{T}}$ is required to be adapted, c\`agl\`ad, satisfy
 $\pos_0 = \pos_{T+} = 0$ and
    $\E^\Q[\int_0^T \pos_t^2 \sigma^2(t,S_t) S_t^2 \d t] < \infty$. If $2\varphi \ne \lambda$, we additionally require its right-continuous version $\pos^+$ to have finite quadratic variation.
\end{definition}
To be consistent with the discrete-time framework of Section~\ref{sec:multi_period}, we have chosen to make  the holdings process $\pos$ c\`agl\`ad.
 In the liquidation value process dynamics \eqref{eq:cash_value_dynamics} to come, $\pos$ appears as an integrand, but $\pos^+$ also appears in other roles when $2\varphi \ne \lambda$, which is why we require $\pos^+$ to have finite quadratic variation.
 In the frictionless case $\val$ is continuous, due to the continuity of the price process, but here it can experience jumps due to bulk trades prescribed by $\pos$. By convention, we enforce the same one-sided continuity conventions for $\val$ and $\pos$; that is, liquidation value processes $\val$ are also c\`agl\`ad. 

As in discrete-time, the observed price $P^\pos$ is given by  the right-hand side of \eqref{eq:observed_price} when the strategy $\pos$ is used.
To obtain the dynamics of the liquidation value process we take a limit of the discrete-time model as the number of periods $M$ tends to infinity, which leads to the dynamics \begin{equation} \label{eq:cash_value_dynamics}
e^{-rt}\val_t  = \val_0 
+\int_0^t e^{-ru}\pos_u\sigma(u,S_u)S_u\d W_u - \frac{r\lambda}{2}\int_0^te^{-ru}\pos_u^2 \d u  -\frac{2\varphi -\lambda}{2}\bigg(e^{-rt}\pos_t^2 +\int_{[0,t)}e^{-rs} \d[\pos^+]_s\bigg)
\end{equation}
under $\Q$ for all $t \in  \mathcal{T}$ and given initial value $\val_0 \in \R$. 
See Proposition~\ref{prop:limit} in \ref{sec:proof_cont} of the electronic companion for a precise statement and proof of this limiting expression.
Here, and in what follows, we use the convention $\int_{\{0\}}e^{-rs}d[\pos^+]_s  = \pos_{0+}^2$ and when $t = T+$, we interpret $[0,t)$ as $[0,T]$, so that the final term in \eqref{eq:cash_value_dynamics} incorporates both initial and terminal bulk trades.
In the midpoint regime $2\varphi = \lambda$ the final term in \eqref{eq:cash_value_dynamics} vanishes, so no assumption on the quadratic variation of $\pos^+$ is required.
 As in the discrete-time case, it is easy to see from \eqref{eq:cash_value_dynamics} that $\val_{T+} = \val_T$. 
 We now establish necessary and sufficient conditions on the parameters that ensure absence of price manipulation.
\begin{proposition}[Absence of price manipulation]\label{prop:price_manip_cont_time}
    The model does not admit price manipulation opportunities if and only if \eqref{eq:linear_price_manip} holds.
\end{proposition}
\subsection{The pricing PDE}
We now consider the question of replicating the payoff of a European claim $V:\R \to \R$. As before, the task is to find an admissible strategy $\pos^*$ such that $\val_T^* = V(P^{\pos^*}_T)$ holds, $\Q$-a.s.
From the analysis in discrete-time, we formally expect that the midpoint liquidation value process $\valm^*$ given by \eqref{eq:modified_cash_value} satisfies $\valm^*_t = u(t,S_t)$ for some function $u$ to be determined. To derive a PDE for $u$ we, on the one hand, directly compute the discounted value dynamics for $t \in \mathcal{T}$,
\begin{align}
    e^{-rt}\valm_t^* & = e^{-rt}\big(\val_t^* + \tfrac{2\varphi - \lambda}{2}(\pos_t^*)^2\big) \nonumber  \\
    & = \val_0^* + \int_0^t e^{-ru}\pos^*_u \sigma(u,S_u)S_u \d W_u - \frac{r\lambda}{2}\int_0^t e^{-ru}(\pos^*_u)^2 \d u - \frac{2\varphi-\lambda}{2}\int_{[0,t)}e^{-rs}\d[(\pos^*)^+]_s\label{eq:modified_cash_value_dynamics}
\end{align}
from \eqref{eq:cash_value_dynamics} while, on the other hand, we use It\^o's formula to obtain
\begin{align} 
\d\big(e^{-rt}u(t,S_t)\big) 
& = e^{-rt}\big(-ru(t,S_t) + \partial_t u(t,S_t) + rS_t\partial_x u(t,S_t) + \tfrac{1}{2}\sigma^2(t,S_t)S_t^2 \partial_{xx}u(t,S_t)\big)\d t \nonumber \\
& \quad + e^{-rt}\sigma(t,S_t)S_t\partial_x u(t,S_t)\d W_t \label{eq:Ito_cash_value} 
\end{align}
for $t \in (0,T)$. Comparing the Brownian terms in \eqref{eq:modified_cash_value_dynamics} and \eqref{eq:Ito_cash_value} leads us for $t \in (0,T)$ to the choice $\pos^*_t = \partial_x u(t,S_t)$ as suggested by the discrete-time heuristics in Section~\ref{sec:heuristics}. This comparison applies on $(0,T)$, since there is a bulk trade that occurs at time $t = 0+$ to establish the initial hedge. As such, $\pos^*$ jumps at zero, leading to a jump in the midpoint liquidation value process $\valm^*$ via the quadratic variation term in \eqref{eq:modified_cash_value_dynamics}. Hence, we obtain \begin{align} 
e^{-rt}\valm_t^* & = \val_0^* - \tfrac{2\varphi-\lambda}{2}\big(\partial_xu(0,S_0)\big)^21_{\{t>0\}} + \int_0^t e^{-rs}\sigma(s,S_s)S_s\partial_x u(s,S_s)\d W_s \nonumber \\
& \qquad  - \int_0^te^{-rs}\Big(\tfrac{r\lambda}{2}\big(\partial_x u(s,S_s)\big)^2 + \tfrac{2\varphi-\lambda}{2}\big(\partial_{xx}u(s,S_s)\big)^2 \sigma^2(s,S_s)S_s^2\Big)\d s. \label{eq:modified_cash_value_simplified}
\end{align}
 Matching the drift terms in \eqref{eq:modified_cash_value_simplified} to those in \eqref{eq:Ito_cash_value} leads us to the nonlinear PDE
\begin{equation}
\label{eq:PDE}
\begin{aligned} 
0 & = -ru(t,x) + \partial_t u(t,x) + rx\partial_xu(t,x) + \tfrac{1}{2}\sigma^2(t,x)x^2\partial_{xx}u(t,x) \\
& \quad + \tfrac{r\lambda}{2}\big(\partial_x u(t,x)\big)^2 + \tfrac{2\varphi-\lambda}{2}\sigma^2(t,x)x^2\big(\partial_{xx}u(t,x)\big)^2,
\end{aligned} \quad \text{for }(t,x) \in [0,T) \times (0,\infty).
\end{equation}
With the drift terms matched, comparison of \eqref{eq:Ito_cash_value} and \eqref{eq:modified_cash_value_simplified} also shows that $\valm_t^*=u(t,S_t)$ for $t>0$ provided
\begin{equation}\label{eq:initial_cash}
\val_0^* = u(0,S_0)
+\tfrac{2\varphi-\lambda}{2}\big(\partial_xu(0,S_0)\big)^2.
\end{equation}
The replication condition requires $\val_T^* = V(P_T^{\pos^*}) = V(S_T + \lambda \partial_x u(T,S_T))$. Combining with the identity \eqref{eq:modified_cash_value} we obtain an implicit ODE for the terminal condition
\begin{align} \label{eq:terminal_condition}
    u(T,x) & = V\big(x+\lambda \partial_x u(T,x)\big) + \tfrac{2\varphi-\lambda}{2}\big(\partial_x u(T,x)\big)^2, \qquad x >0.
\intertext{From \eqref{eq:modified_cash_value} and the identity $\valm_t^* = u(t,S_t)$ for $t > 0$, the liquidation value process is recovered via}
 \label{eq:cash_value_recovery}
\val_t^* & = u(t,S_t) - \tfrac{2\varphi-\lambda}{2}\big(\partial_x u(t,S_t)\big)^2, \qquad t \in (0,T].
\end{align}
In particular, $\val_0^*$ is given by \eqref{eq:initial_cash}, while $
    \val^*_{0+}  = u(0,S_0) - \frac{2\varphi-\lambda}{2}(\partial_x u(0,S_0))^2$; that is, the bulk trade establishing the hedger's initial position instantaneously reduces the liquidation value by $\val^*_0 - \val^*_{0+} = (2\varphi-\lambda)(\partial_x u(0,S_0))^2$.

We now verify that if a sufficiently regular solution can be found to the PDE \eqref{eq:PDE} and terminal implicit ODE \eqref{eq:terminal_condition} then the replication problem is solved. 
\begin{theorem}[Verification] \label{thm:verification}
    Let a function $V:\R \to \R$ be given. Suppose there exists a function $u:[0,T] \times (0,\infty) \to \R$ satisfying the following conditions: 
        \begin{enumerate}
        \item \label{item:C12}the function $u$ belongs to $C^{1,2}([0,T) \times (0,\infty)) \cap C([0,T] \times (0,\infty))$ and satisfies \eqref{eq:PDE},
         \item \label{item:polynomial} there exist $C > 0$ and $p \ge 1$ such that
        \begin{equation}
            \label{eq:polynomial_bound}
        |\partial_x u(t,x)| \le C(1+x^p), \qquad \forall t \in [0,T), \,  x \in (0,\infty),
        \end{equation}
        \item \label{item:absolute_continuity}the function $u(T,\cdot)$ is absolutely continuous and satisfies \eqref{eq:terminal_condition} for a.e.\ $x \in (0,\infty)$,
        \item \label{item:T_continuity} $\lim_{t  \uparrow T} \partial_x u(t,S_t) = \partial_x u(T,S_T)$, $\Q$-a.s. 
    \end{enumerate}
    Then the strategy $\pos_t^* = \partial_x u(t,S_t)1_{(0,T]}(t)$ is admissible,  the liquidation value process is given by \eqref{eq:cash_value_recovery}, $\valm_t^* = u(t,S_t)$ for $t \in (0,T]$, and the strategy $\pos^*$ replicates $V$ when initiated with cash position $\val^*_0 = u(0,S_0) + \frac{2\varphi-\lambda}{2}(\partial_x u(0,S_0))^2$.
\end{theorem}

\subsection{Terminal condition} \label{sec:terminal_condition}
As in the discrete-time case, we study the replication problem backward in time starting with the implicit terminal condition \eqref{eq:terminal_condition}. 
To facilitate the analysis we will restrict to monotone, Lipschitz and convex payoffs $V:\R \to \R$, such as calls and puts. For increasing payoffs, the replicating strategy will take positive holdings and, as such, the value of $V$ on $(-\infty,0)$ will not affect the pricing or hedging problem. However, for decreasing payoffs, we expect to have $\pos^*_T = \partial_x u(T,S_T) < 0$ and from \eqref{eq:terminal_condition} we see that the value of $V$ on $(-\infty,0)$ may enter and affect $u(T,x)$ for some $x > 0$. This reflects a modelling asymmetry --- although the fundamental price $S$ is strictly positive, there is nothing stopping the observed price $P^\pos$ of \eqref{eq:observed_price} from becoming negative, with small but positive probability, for sufficiently negative positions $\pos$. For the well-posedness study of the terminal condition \eqref{eq:terminal_condition} and the PDE \eqref{eq:PDE}, it will be convenient to take the following convention on the payoff:
\begin{equation}\label{eq:payoff_negative} V(x) = V(0),  \qquad x \le 0.  
\end{equation}
For the put option this simply specifies that the payoff the holder receives can never exceed the strike price $K$. 

We now turn to the implicit ODE \eqref{eq:terminal_condition}. For  many important payoffs, like calls and puts, $V$ is only Lipschitz continuous but not continuously differentiable. Nevertheless, the quadratic term on the right-hand side of \eqref{eq:terminal_condition} has a regularizing effect and leads to $C^1$ solutions. In the case $2\varphi = \lambda$ however, that term vanishes and we cannot expect everywhere differentiable solutions to this implicit ODE; indeed, for the call option the solution in that case inherits the characteristic hockey-stick shape of the call. For this reason we work with the weaker notion of Carath\'eodory solutions to ODEs, which require the solution to be absolutely continuous and satisfy the ODE almost everywhere (see \cite[Section~I.5]{hale2009ordinary}). 

\begin{theorem} \label{thm:terminal}
    Let $V:[0,\infty) \to \R$ be a Lipschitz continuous, monotone and convex function. We extend $V$ to $\R$ by \eqref{eq:payoff_negative}. If $2\varphi \ge \lambda$ then there exists a Carath\'eodory solution $\widetilde V:\R \to \R$ to the implicit ODE
    \begin{equation}     \label{eq:implicit_ODE}
 \widetilde V(x) = V\big(x+\lambda \widetilde V'(x)\big) + \tfrac{2\varphi - \lambda}{2}\big(\widetilde V'(x)\big)^2, \qquad x \in \R, 
 \end{equation} which satisfies the following minimality property: suppose $f$ is another  Carath\'eodory solution to \eqref{eq:implicit_ODE}, and if $V$ is nonincreasing additionally suppose that $f(x) = V(0)$ for $x \le 0$ and $f'(x) \ge -x/\lambda$ for a.e.\ $x > 0$; then $\widetilde V \le f$ and $|\widetilde V'(x)| \le |f'(x)|$ at all common points $x$ of differentiability.
    
    Additionally, $\widetilde V \ge V$ and it is monotone in the same direction as $V$. If $2\varphi > \lambda$ then $\widetilde V$ is  $C^1$.
\end{theorem}
In \ref{sec:terminal_ODE}, as part of the proof of this theorem we additionally derive a more explicit representation for $\widetilde V$ in terms of a unique bounded solution to the standard first-order ODE \eqref{eq:P_ODE}. Since the implicit ODE \eqref{eq:implicit_ODE} is not equipped with a natural initial condition, it may have multiple solutions. The minimal solution $\widetilde V$, as described in Theorem~\ref{thm:terminal}, is then the natural financial choice. Indeed, since the replicating strategy satisfies $\pos^*_T = \partial_x u(T,S_T) = \widetilde V'(S_T)$, the property $|\widetilde V'| \le |f'|$ is the analogue of the minimal-in-absolute-value solution of \eqref{eq:smallest_solution} discussed in the one-period model.
The upshot of Theorem~\ref{thm:terminal} is that the nonlinear PDE \eqref{eq:PDE} can now be appended with the terminal condition 
\begin{equation} \label{eq:terminal_condition_explicit}
    u(T,x) = \widetilde V(x), \qquad x  > 0,
\end{equation}
instead of the implicit version in \eqref{eq:terminal_condition}.
\subsection{Solving the PDE \eqref{eq:PDE}}\label{sec:PDE}

We now focus on the pricing PDE \eqref{eq:PDE} with terminal condition \eqref{eq:terminal_condition_explicit}. As in Theorem~\ref{thm:terminal}, we focus on monotone, convex and Lipschitz continuous payoff functions $V$ satisfying \eqref{eq:payoff_negative}.  Then, formally the equations \eqref{eq:PDE} and \eqref{eq:terminal_condition_explicit} correspond to the Hamilton--Jacobi--Bellman (HJB) equation of the stochastic control problem
\begin{equation} \label{eq:u_stoch_rep}
u(t,x) = \sup_{\nu \in \mathcal{A}}  \E\bigg[e^{-r(T-t)}\widetilde V(S_T^\nu) - \int_t^T e^{-r(s-t)}\bigg(\frac{\alpha_s^2}{2r\lambda} + \frac{\gamma_s^2}{2(2\varphi-\lambda)\sigma^2(s,S^\nu_s)(S_s^\nu)^2}\bigg)\d s\, \bigg|\, S_t^\nu = x\bigg], 
\end{equation}
where the controlled state dynamics are
\[\d S^\nu_s = (rS_s^\nu + \alpha_s)\d s + 1_{\{S_s^\nu > 0\}} \sqrt{\sigma^2(s,S_s^\nu)(S_s^\nu)^2 + 2\gamma_s}\d W_s, \quad S^\nu_t = x, \qquad s \in [t,T].\] Here, $\mathcal{A}$ consists of all adapted square-integrable controls $\nu = (\alpha,\gamma)$ with $2\gamma_s \ge -\sigma^2(s,S_s^\nu)(S_s^\nu)^2$,  whenever $S_s^
\nu > 0$, and $\gamma_s = 0$, whenever $S_s^\nu \le 0$, with the convention that both the $\gamma$-dependent running cost in \eqref{eq:u_stoch_rep} and the diffusion coefficient of $S^\nu$ are zero in this case.
If $r = 0$, then \eqref{eq:u_stoch_rep} is understood with $\alpha \equiv 0$, while if $2\varphi = \lambda$ then \eqref{eq:u_stoch_rep} is understood with $\gamma \equiv 0$ so that the first or second term in the integrand vanishes respectively. Indeed, the HJB associated with  \eqref{eq:u_stoch_rep} is
\begin{equation}
\label{eq:HJB}
\begin{aligned} 
0 & = -ru(t,x) + \partial_t u(t,x) + rx\partial_xu(t,x) + \tfrac{1}{2}\sigma^2(t,x)x^2\partial_{xx}u(t,x) \\
& \quad +\sup_{\alpha \in \R} \{\alpha \partial_x u- \tfrac{\alpha^2}{2r\lambda}\} + \sup_{\gamma \ge -\frac{1}{2}\sigma^2(t,x)x^2}\{\gamma \partial_{xx}u - \tfrac{\gamma^2}{2(2\varphi-\lambda)\sigma^2(t,x)x^2}\}
\end{aligned} \quad \text{for }(t,x) \in [0,T) \times (0,\infty).
\end{equation}
The maximum in $\alpha$ is always achieved by $\alpha^* = r\lambda\partial_x u$, while the maximum for $\gamma$ is achieved by the unconstrained optimizer $\gamma^* = (2\varphi - \lambda)\sigma^2(t,x)x^2\partial_{xx}u$ if and only if 
\begin{equation} \label{eq:gamma_constraint}
    (2\varphi - \lambda)\partial_{xx} u(t,x) \ge -\tfrac{1}{2}.
\end{equation}
In this case, the HJB \eqref{eq:HJB} becomes \eqref{eq:PDE}.

The condition \eqref{eq:gamma_constraint} is always satisfied in the midpoint regime $2\varphi = \lambda$, which is the case we focus on in the following theorem establishing smooth classical solutions to the pricing PDE.

	 \begin{theorem} \label{thm:PDE}
        Let the payoff function $V$ be as in Theorem~\ref{thm:terminal} and suppose $2\varphi = \lambda$. We set $\varepsilon = 1$ if $V$ is nondecreasing or constant and $\varepsilon = -1$ if $V$ is nonincreasing.
		\begin{enumerate}[label = (\roman*)]
			\item  \label{item:u} Then there exists a unique function $u \in C^{1,2}([0,T) \times (0,\infty)) \cap C([0,T] \times (0,\infty))$, which satisfies \eqref{eq:PDE} and \eqref{eq:terminal_condition_explicit}, and the bounds 
			\begin{equation} \label{eq:dxu_bounds}
				\varepsilon \partial_x u(t,x) \in [0,L_V] \qquad \text{and, if }\varepsilon = -1,  \quad \partial_x u(t,x) \ge -\Lambda x, \qquad (t,x) \in [0,T)\times (0,\infty) 
			\end{equation}
			 for some $\Lambda > 0$, and where $L_V$ is the Lipschitz constant of $V$.
		\item \label{item:stochastic_rep}The function $u$  of \ref{item:u} admits the stochastic representation \eqref{eq:u_stoch_rep} with $\gamma \equiv 0$ (and $\alpha \equiv 0$ if $r =0$), and we have $\lim_{t \uparrow T} \partial_x u(t,S_t) = \widetilde V'(S_T), \mathbb{Q}$-a.s.		
		\item We have the estimates \label{item:bbounds}
		\begin{equation} \label{eq:u_bounds}
			\widetilde u^{BS}(t,x) + \tfrac{\lambda L_V^2}{2}(1-e^{-r(T-t)}) \ge u(t,x) \ge \widetilde u^{BS}(t,x) \ge u^{BS}(t,x),
			\end{equation} 
		where $\widetilde u^{BS}(t,x) = \E^\Q[e^{-r(T-t)}\widetilde V(S_T) \mid S_t=x]$ and $u^{BS}(t,x) = \E^\Q[e^{-r(T-t)} V(S_T)\mid S_t=x]$.	\end{enumerate}
	\end{theorem} 
This result together with Theorem~\ref{thm:verification} ensures that any such payoff $V$ can be perfectly replicated in the midpoint cost case.
\begin{corollary} \label{cor:replication}
    Let the payoff function $V$ be as in Theorem~\ref{thm:terminal} and suppose $2\varphi = \lambda$. This payoff can be replicated with initial capital $\val^*_0 = u(0,S_0)$ and admissible trading strategy $\pos^*_t = \partial_x u(t,S_t)1_{(0,T]}(t)$, where $u$ is the function obtained in Theorem~\ref{thm:PDE}.
\end{corollary}

 As in the frictionless case, the first bound in \eqref{eq:dxu_bounds} guarantees that the replicating strategy never holds more shares of the underlying than the Lipschitz constant of the payoff; that is, for calls and puts the investor continues to never hold more than a single share per contract. Additionally, the bounds in \eqref{eq:u_bounds} again illustrate that it is always more expensive to replicate the contract in the present setting with frictions than in the frictionless setting. In fact, the representation \eqref{eq:u_stoch_rep} suggests that this ordering should persist in the general $2\varphi \ge \lambda$ regime, by taking $\alpha \equiv \gamma \equiv 0$ and recalling from Theorem~\ref{thm:terminal} that $\widetilde V \ge V$. The bounds \eqref{eq:u_bounds} also provide an upper bound on how much more expensive the replicating cost can be. Indeed, evaluating at $t = 0$ stipulates that it can never exceed the frictionless cost of replicating the \emph{modified} payoff $\widetilde V$ plus a premium of $\frac{1}{2}\lambda L_V^2(1-e^{-rT})$. In the case of call and put options this premium is proportional to the impact parameter $\lambda$ and the square of the number of contracts $|N| = L_V$.

 When  $r = 0$, the PDE is solved by $\widetilde u^{\mathrm{BS}}$, which is  the frictionless price of the \emph{modified} payoff $\widetilde V$. Since the replicating strategy is $\pos^*_t  = \partial_x \widetilde u^{\mathrm{BS}}(t,S_t)$, the hedger replicating $V$ in the price impact setting trades as if they were hedging $\widetilde V$ without any impact or execution costs.

We finish off this section with a scaling property relating the
price impact coefficient $\lambda$ to the number of contracts $N$.
Suppose the \emph{gross payoff} is $V(x)=Nv(x)$ for a fixed payoff function
$v$ and $N>0$. In the midpoint regime $2\varphi=\lambda$, the
liquidation value and hedging position per contract (divided by $N$) are determined
by the function $u_N = u/N$ and its spatial derivative. Dividing
\eqref{eq:PDE} and \eqref{eq:terminal_condition} by $N$ shows that
$u_N$ satisfies the corresponding single contract equations with
$\lambda$ replaced by the effective impact parameter $N\lambda$.
Thus, both the
liquidation value and the hedging position per contract depend on
$N$ and $\lambda$ through their product $N\,\lambda$. 
That is, in the midpoint regime $2\,\varphi=\lambda$, doubling the number of contracts has the same effect (on these normalized
quantities) as doubling $\lambda$.

More generally, when $2\,\varphi>\lambda$, the normalized equations depend on the two effective parameters $N\lambda$ and
$N(2\varphi-\lambda)$. Holding $N\lambda$ fixed does not necessarily 
preserve the replication cost or hedge per 
contract. The scaling equivalence extends to this regime if both
$\lambda$ and $\varphi$ are varied proportionally. 

\section{Numerical results} \label{sec:numerics}
In the numerical examples below we take constant volatility $\sigma > 0$, as in the Black--Scholes model.

\subsection{Quadratic payoffs: a closed-form example} \label{sec:quadratic_payoffs}

As a first exercise, here we explicitly solve the replication problem for a quadratic payoff $V(x) = N(\alpha x^2 +\beta x +\gamma)$ for $N > 0$ contracts and parameters $\alpha \in (0,(16\varphi N)^{-1})$ and $\beta,\gamma \in \R$, in the general parameter regime $2\varphi \ge \lambda$. The key observation is that quadratic ansatzes for $\widetilde V$ and $u(t,\cdot)$ satisfy the terminal condition \eqref{eq:terminal_condition} and PDE \eqref{eq:PDE}, respectively. Indeed, conjecturing $\widetilde V(x) = a(T)x^2 + b(T)x + c(T)$, substituting into \eqref{eq:implicit_ODE} and collecting like terms leads to a quadratic equation for $a(T)$, and linear equations for $b(T)$ and $c(T)$. For $a(T)$ we take the root that among the two quadratic solutions makes the magnitude of the implied terminal inventory the smallest, akin to the minimality approach in Theorem~\ref{thm:terminal}, leading to the nested solutions
\begin{align*}
    a(T)&= \frac{1-4\,N\alpha\lambda - \sqrt{1-16\,N\alpha\,\varphi}}{4\,(2N\,\alpha\,\lambda^2 +2\,\varphi-\lambda)}
    \\
    b(T) &= 
    \frac{-N\beta -2N\,\beta\,\lambda\,a(T)}{-1+2\,N\alpha\,\lambda-2\,\lambda\,a(T)+ 4N\,\alpha\,\lambda^2\,a(T)+ 4\,\varphi\,a(T)},
   \\
   c(T)&= N\gamma + N\beta\,\lambda\,b(T) + N\alpha\,\lambda^2 \,b(T)^2+ \tfrac{1}{2}( 2\,\varphi-\lambda)\,b(T)^2\,.
\end{align*}
For the PDE, using the quadratic ansatz $u(t,x) = a(t)x^2 + b(t)x + c(t)$ and collecting like terms leads to 
\begin{align*}
&x^2\Big\{(r+\sigma^2)\,a(t) + 2\,r\,\lambda\,a(t)^2 + 2\,\sigma^2(-\lambda+2\,\varphi)a(t)^2 +a'(t) \Big\}
\\
&\quad + x\Big\{ 2\,r\,\lambda\,a(t)\,b(t) + b'(t) \Big\} 
+ \Big\{ \frac{1}{2}r\,\lambda \,b(t)^2 - r\,c(t) +  c'(t) \Big\} =0\,.
\end{align*}
The solution to the Riccati equation is
\begin{equation} \label{eq:at}
    a(t)=\frac{a(T)(r+\sigma^2)}{(r+\sigma^2) - 2a(T)(r\lambda + \sigma^2(2\varphi-\lambda))(e^{(r+\sigma^2)(T-t)}-1)}e^{(r+\sigma^2)(T-t)}.
\end{equation}
We impose the condition
\[
    2a(T)\big(r\lambda + \sigma^2(2\varphi-\lambda)\big)\big(e^{(r+\sigma^2) T}-1\big)<r + \sigma^2, 
\]
on the maturity time $T$, which ensures that the denominator in \eqref{eq:at} is positive throughout $[0,T]$. 
 With $a(t)$ given we obtain the unique solution to the linear ODE for $b(t)$ with terminal condition $b(T)$, which in turn we use to solve the ODE for $c(t)$ with terminal condition $c(T)$. As the solution is smooth on $[0,T] \times \R$ and $\partial_x u(t,x)$ is affine in $x$, by Theorem~\ref{thm:verification} we have that $\pos^*_t = \partial_x u(t,S_t)1_{(0,T]}(t)$ and $\val^*$ given by \eqref{eq:initial_cash} and \eqref{eq:cash_value_recovery} are a replicating strategy and liquidation value process for this quadratic payoff.

Figure~\ref{fig:LQ price and delta} illustrates the effect of price impact on the 
replicating value and initial position for different numbers of derivative contracts.
The top panels show the derivative price per contract, while the bottom panels show 
the initial position per contract, as functions of the initial fundamental price $S_0$. 
For all the parameter combinations and initial prices shown the values in the presence of price impact lie above their corresponding frictionless benchmarks, although for $\lambda = 0.01$ and $N=1$ the difference is barely visible.      
The difference becomes more pronounced as the number  of contracts increases. In particular, even after normalizing by the number of contracts $N$, both the replication cost and the size of the hedge increase with $N$. This illustrates  the nonlinear dependence of the replication problem on the number of contracts, which  arises because a larger hedging position generates a larger displacement of the observed price. Note, as per the discussion at the end of Section~\ref{sec:cont_time}, that the $N=1$ and $\lambda = 0.1$ curves are exactly equal to the $N =10$ and $\lambda = 0.01$ curves in Figure~\ref{fig:LQ price and delta} as we are in the midpoint regime $2\varphi = \lambda$ and they share the same effective parameter $\lambda N$.

\begin{figure}
    \centering
    \includegraphics[width=0.8\linewidth]{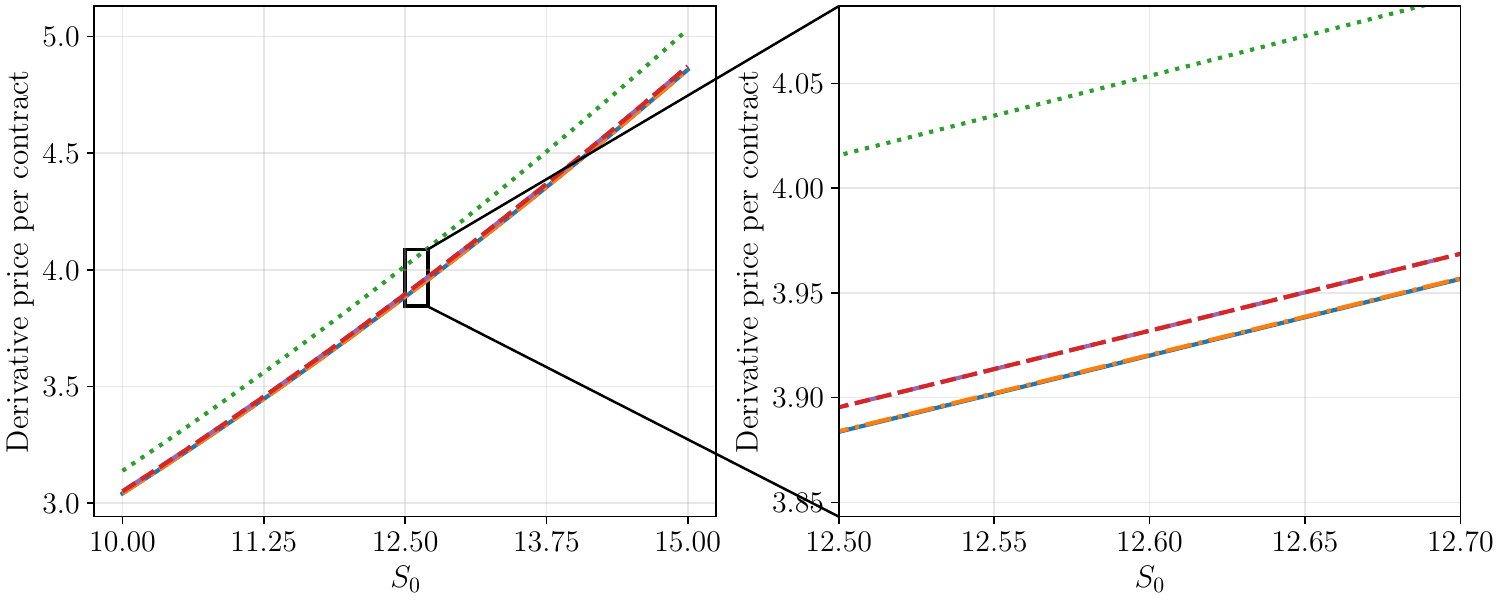}
    \includegraphics[width=0.8\linewidth]{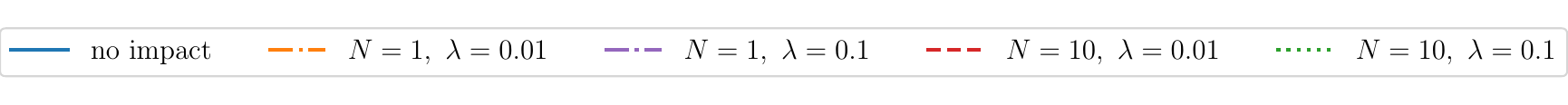}
    \includegraphics[width=0.8\linewidth]{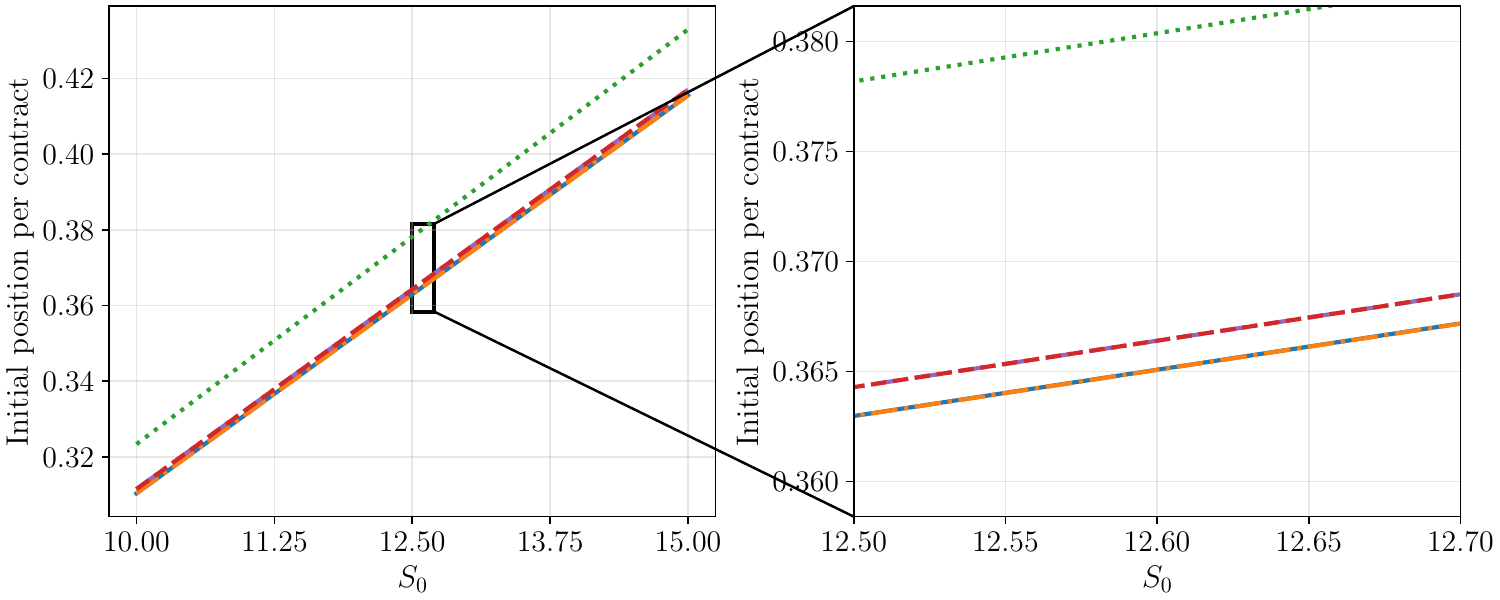}
    \caption{Quadratic option price per contract (top) and initial replicating position per 
    contract (bottom) in the midpoint regime as functions of the initial fundamental price $S_0$. The solid line 
    corresponds to the no-impact benchmark, while the dash-dotted and dotted lines 
    correspond respectively to $N=1$ and $N=10$ contracts with price-impact parameters 
    $\lambda=0.01$ and $\lambda=0.1$.}
    \label{fig:LQ price and delta}
\end{figure}

\subsection{Call options} \label{sec:call_numerics}
Next, we tackle the  challenging problem of replicating a call option with price impact and execution costs. 
For $N >0$ call option contracts, having payoff $V(x) = N(x-K)^+$, the modified terminal payoff of Theorem~\ref{thm:terminal} is explicitly given by
\begin{equation}\label{eq:tildeVcall}
\widetilde V(x) = \begin{cases}
    0, & x < K - 2\varphi N, \\
    \frac{(x-K+2\varphi N)^2}{2(2\varphi-\lambda)}, & K - 2\varphi N \le x < K - \lambda N, \\
    N(x-K) + \frac{2\varphi+\lambda}{2} N^2, & x \ge K-\lambda N.
\end{cases}
\end{equation}
We refer the reader to \ref{sec:terminal_ODE} for the verification that \eqref{eq:tildeVcall} is the minimal solution of \eqref{eq:implicit_ODE} in the sense of Theorem~\ref{thm:terminal}. 
From \eqref{eq:tildeVcall} we see that the terminal replicating position $\pos_T^*$ is determined by
\[\widetilde V'(x) = \frac{x-K+2\varphi N}{2\varphi-\lambda}1_{\{x \in (K-2\varphi N,K-\lambda N)\}} + N1_{\{x \ge K-\lambda N\}}.\]

 Figure~\ref{fig:terminal-condition-call} plots  $\widetilde{V}$ (top panels) and $\widetilde V'$ (bottom panels) in a $2\varphi > \lambda$ regime, as well as the modified call payoff $N(x-\widetilde K)^+$ and its derivative, where $\widetilde K := K - \lambda N$. 
\begin{figure}
    \centering
    \includegraphics[width=0.7\linewidth]{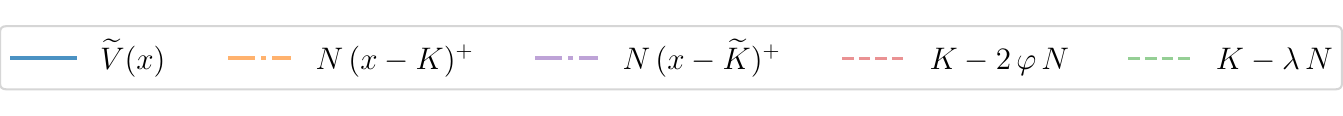}
    \includegraphics[width=0.4\linewidth]{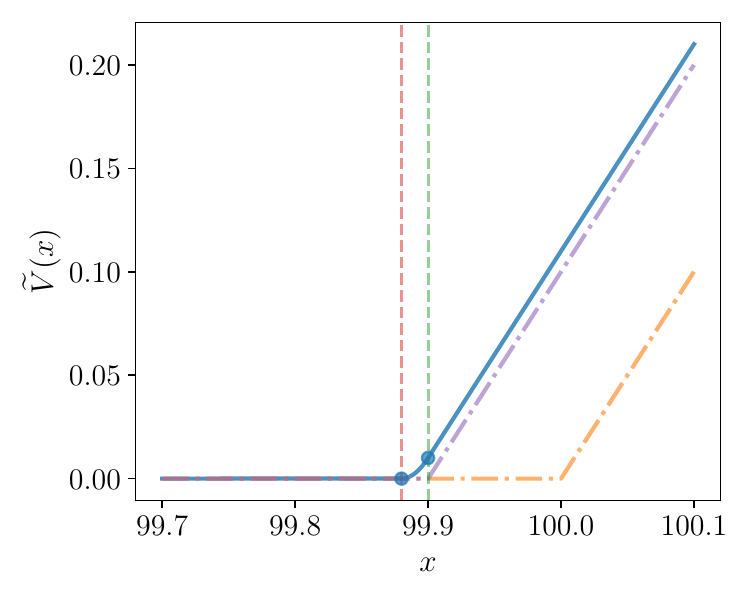}
    \includegraphics[width=0.4\linewidth]{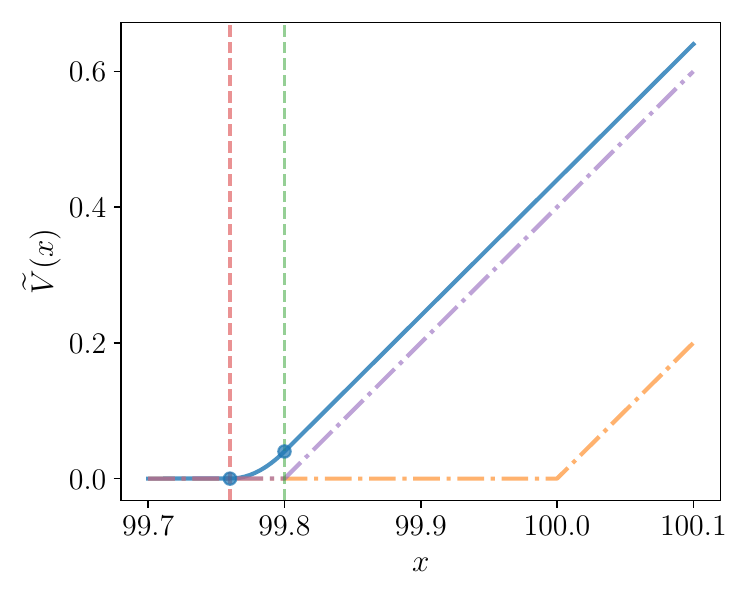}\\
    \includegraphics[width=0.5\linewidth]{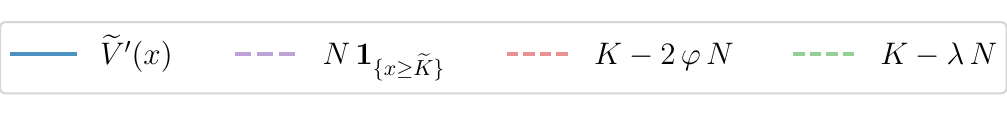}\\
    \includegraphics[width=0.4\linewidth]{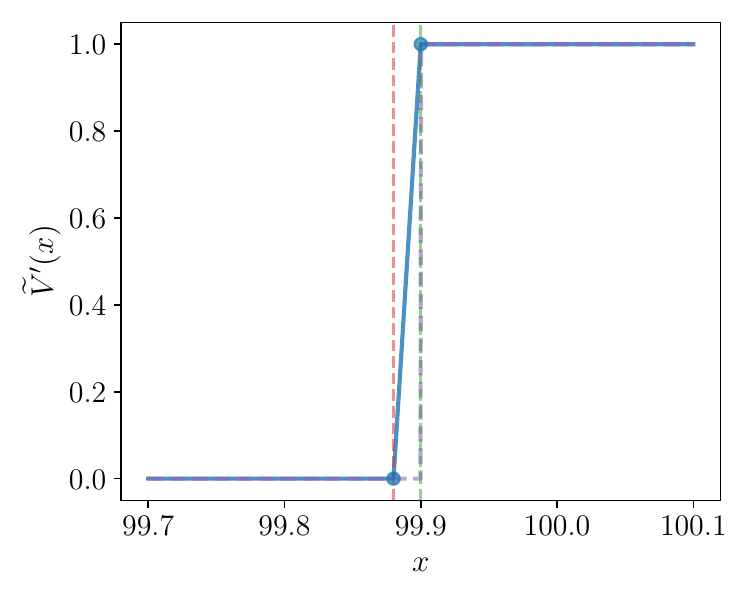}
    \includegraphics[width=0.4\linewidth]{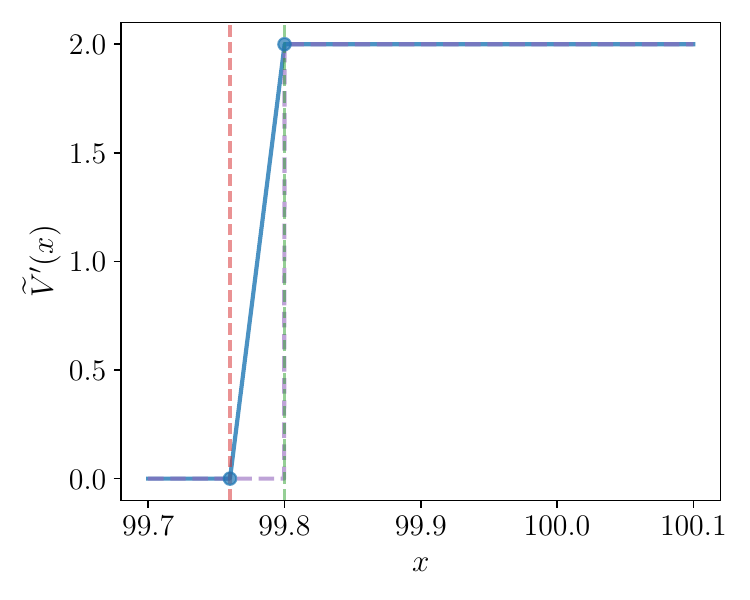}
    \caption{Top panels: Terminal condition $\widetilde{V}(x)$ for the case of a call option with  $K=100$, $\lambda= 0.1$, $\varphi = 0.06$, $r=0.05$, and $\sigma=0.10$.
    Bottom panels: Derivative $\widetilde{V}'(x)$ representing the hedger's position.
    The left panels are for $N=1$ and the right panels are for $N=2$. The dashed vertical lines represent the boundaries $K-2\,\varphi\,N$ and $K-\lambda\,N$ as in  \eqref{eq:tildeVcall}. }
    \label{fig:terminal-condition-call}
\end{figure}
 We see two effects that differ from the frictionless case: a shift in the effective strike and, because $2\varphi - \lambda > 0$, a linear interpolation of the terminal holdings as the option approaches the in-the-money region. For the former effect, if the strike is $K$, any unaffected price above $\widetilde K$ would lead to the call option ending up in the money due to the price impact of the replicating position. Naively, one may expect the hedger to instantaneously transition their position from $0$ to $N$ at the strike $K$, as in the case of the frictionless model.  With price impact, however, this position
would make the observed terminal price jump from $K$ to
$K+\lambda N$, producing a discontinuity in the resulting payoff
as a function of the fundamental price. Such discontinuities
can obstruct replication, as illustrated by the digital call
in Section~\ref{sec:digital}. The terminal condition
\eqref{eq:terminal_condition} instead accounts for the feedback
between the hedging position and the payoff, and provides a continuous
modified terminal payoff $\widetilde V$. 
 
 In the case of midpoint execution costs $2\varphi = \lambda$, the hedger's delta still instantaneously jumps from $0$ to $N$ as in the frictionless case, albeit at the threshold $\widetilde K$. For the $2\varphi > \lambda$ case, the hedger linearly interpolates between these two extremes on the interval $(K-2\varphi N, K-\lambda N)$, whose length is proportional to the price impact discrepancy parameter $2\varphi - \lambda$ (see the bottom panels of Figure~\ref{fig:terminal-condition-call}). When execution is more expensive than the midpoint price, the replicator pre-positions their hedge and accumulates inventory before the effective strike $\widetilde K$. This leads to the effective terminal payoff $\widetilde V$ being a continuously differentiable version of the familiar call option payoff. 

Figure \ref{fig:discretisation_error_hedge_call} shows a scatter plot of the impacted terminal price $P^{\pos^*}_T$ against $\val_T^*$. The strategy $\pos^*$ is the one obtained from $u(t,x)$ when solving the PDE  \eqref{eq:PDE} numerically. The replication identity $V(P^{\pos^*}_T) = \val_T^*$ a.s.~is illustrated as the number of time steps grows large.
\begin{figure}
    \centering
    \includegraphics[width=0.95\linewidth]{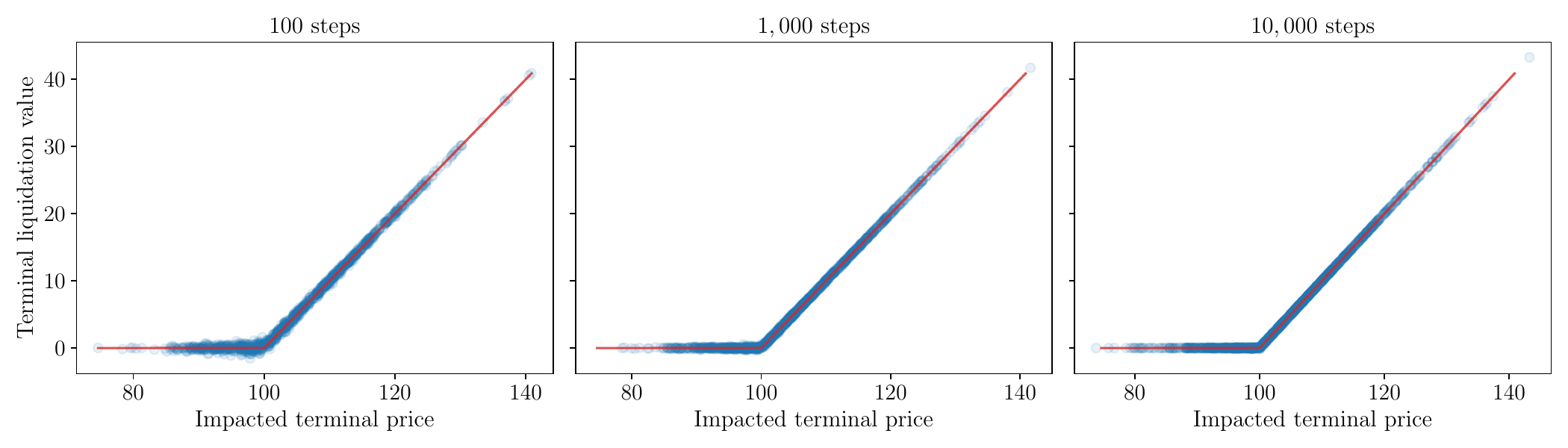}
    \caption{Comparison between $P^{\pos^*}_T$ and  $\val_T^*$. The payoff $V(x) = \,(x-100)^+$ is in a red solid line. The panels are for the cases when $[0,T]$ is discretized in: 100 steps (left), 1,000 steps (middle), and 10,000 steps (right). This illustrates convergence to zero 
    of the well-known discretisation error. }
    \label{fig:discretisation_error_hedge_call}
\end{figure}

Next, Figure \ref{fig:N-call-stress} shows the change in the option price and the initial position $\pos^*_{0+}$ when $N\in\{1,10\}$ and $\lambda \in\{0.01,0.1\}$ for a fixed value of $\varphi = 0.06$. We see that the per-contract impact of increasing the notional by a factor of ten dominates that of the analogous tenfold increase in the price-impact parameter $\lambda$. Moreover, for the parameter values shown, the hedging strategy in our setting is always greater than the Black--Scholes delta. For a call option, the purchases used to hedge the payoff raise
the observed price and thereby increase the liability being hedged.
It follows that the hedger holds more shares
at a given fundamental price than in the frictionless model.\footnote{In the midpoint regime $2\varphi = \lambda$, this relationship can be readily deduced from the PDE \eqref{eq:PDE}. Differentiating the PDE leads to $w = \partial_x u$ solving $\partial_t w + \frac{\sigma^2}{2}x^2\partial_{xx}w + (r + \sigma^2)x\partial_xw + r\lambda w\partial_x w = 0$. The Black--Scholes delta satisfies the same equation without the last term. Since both the delta and gamma for the call are nonnegative the final term is nonnegative, and since the terminal conditions are also ordered, $N1_{\{x \geq \widetilde K\}} \ge N1_{\{x \geq K\}}$, comparison yields the pointwise relationship $\partial_x u \ge \partial_x u^{\mathrm{BS}}$.  }

\begin{figure}
    \centering
    \includegraphics[width=0.8\linewidth]{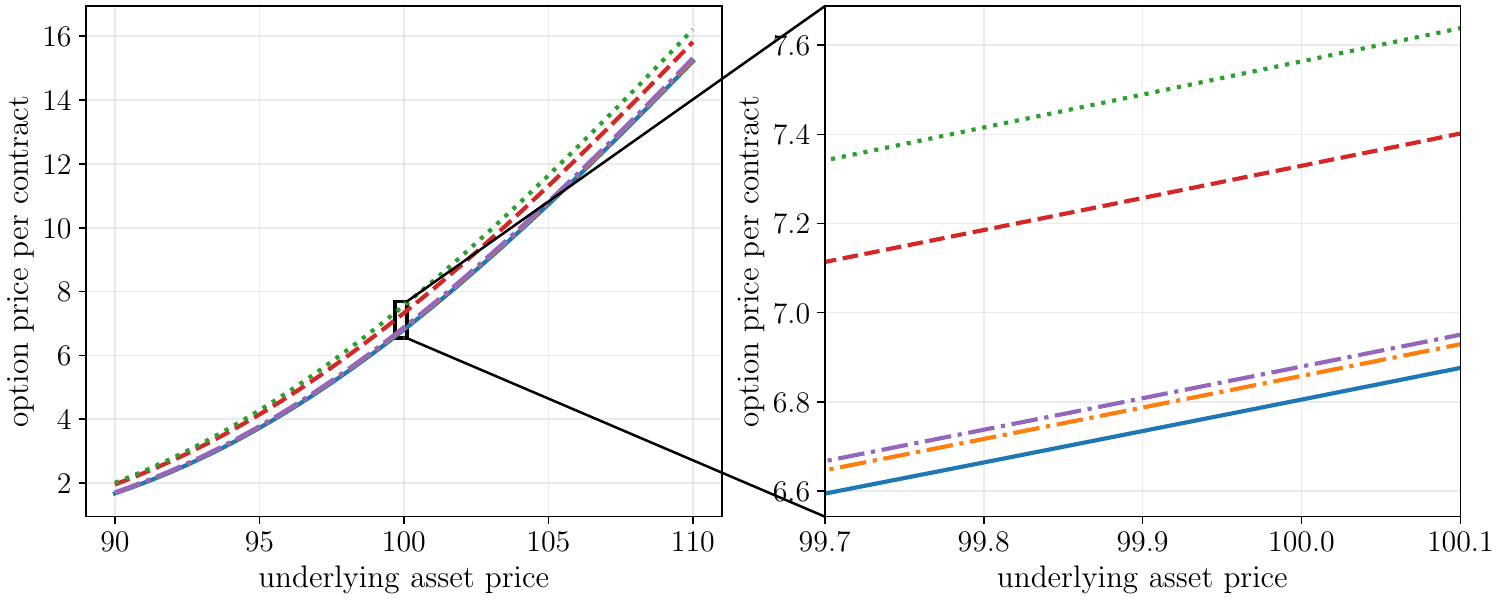}
    \includegraphics[width=0.8\linewidth]{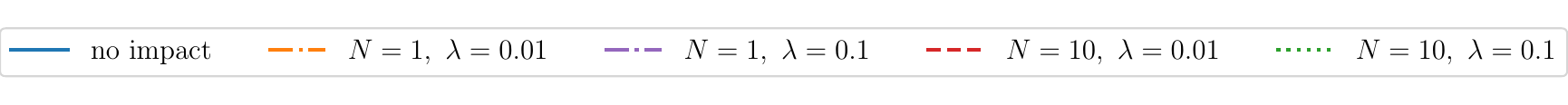}
    \includegraphics[width=0.8\linewidth]{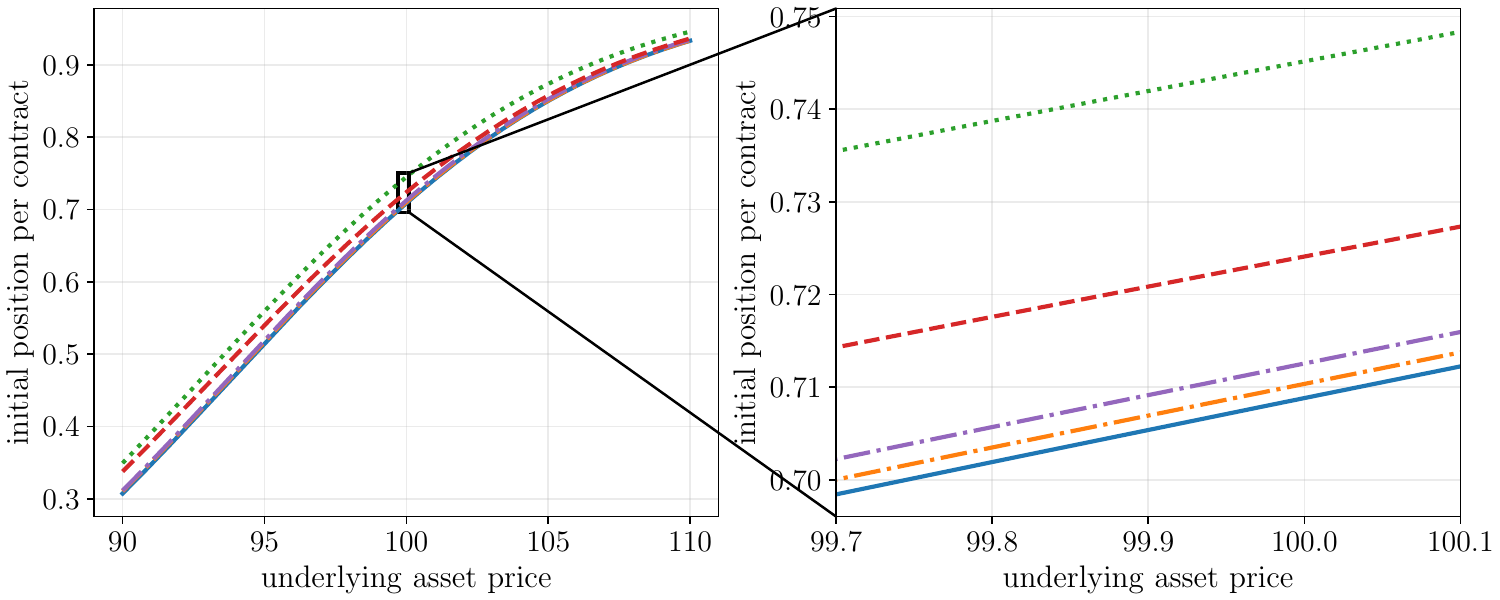}
    \caption{Call option price per contract (top panels) and initial replicating position per 
    contract (bottom panels) as functions of the initial fundamental price. The solid line 
    corresponds to the no-impact benchmark, while the dash-dotted and dashed lines 
    correspond respectively to $N=1$ and $N=10$ contracts with cost parameter $\varphi = 0.06$ and price impact parameter 
    $\lambda=0.01$ or  $\lambda=0.1$.}
    \label{fig:N-call-stress}
\end{figure}

Lastly, among 1,000 simulated paths, we found 9 for which the trading implied by the hedging strategy $\pos^*$ had the property that $S_T<K<P_T^{\pos^*}$; i.e., that the option expired in the money, but would have expired out of the money if the hedger did not trade in the underlying market. Figure~\ref{fig:S<K<P} depicts one of these paths, with model parameters as above and $N=1$. 
This illustrates how the hedger's actions changed the derivative's payoff.

\begin{figure}
    \centering
        \includegraphics[width=0.6\linewidth]{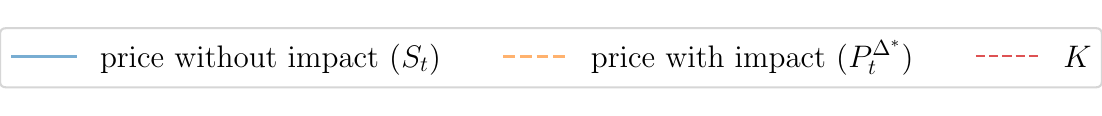}    
\includegraphics[width=0.8\linewidth]{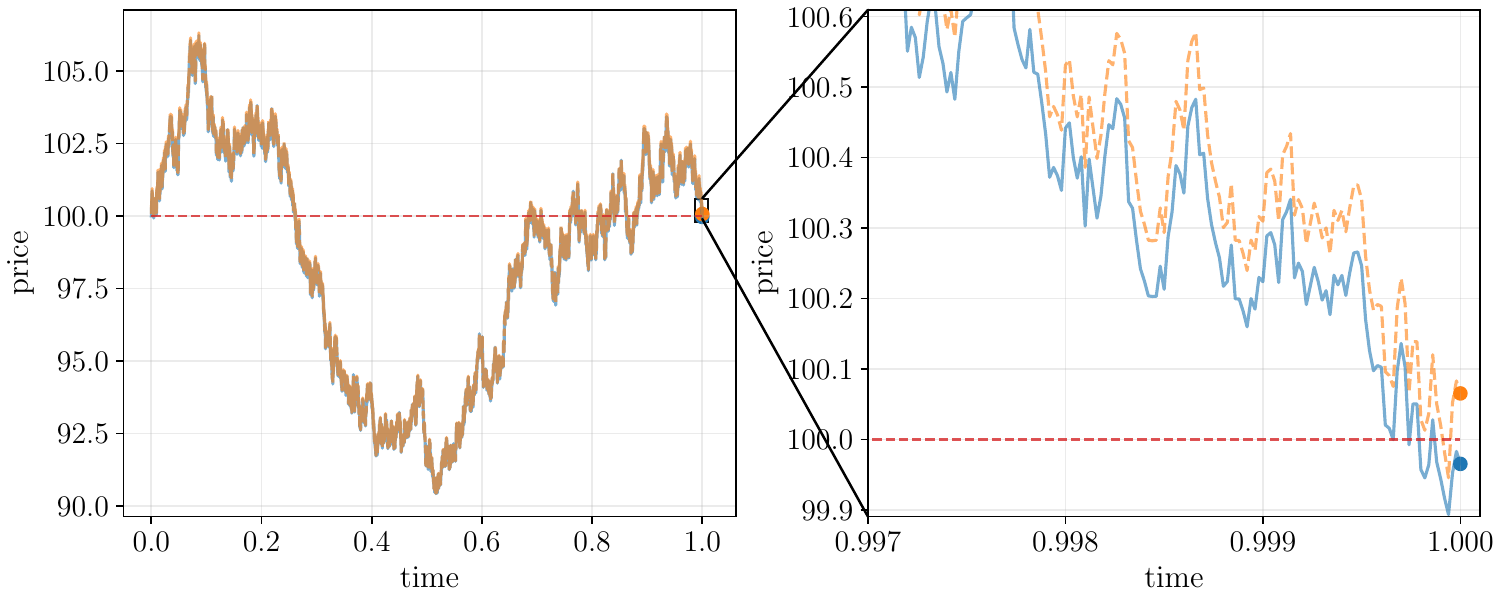}
    \includegraphics[width=0.8\linewidth]{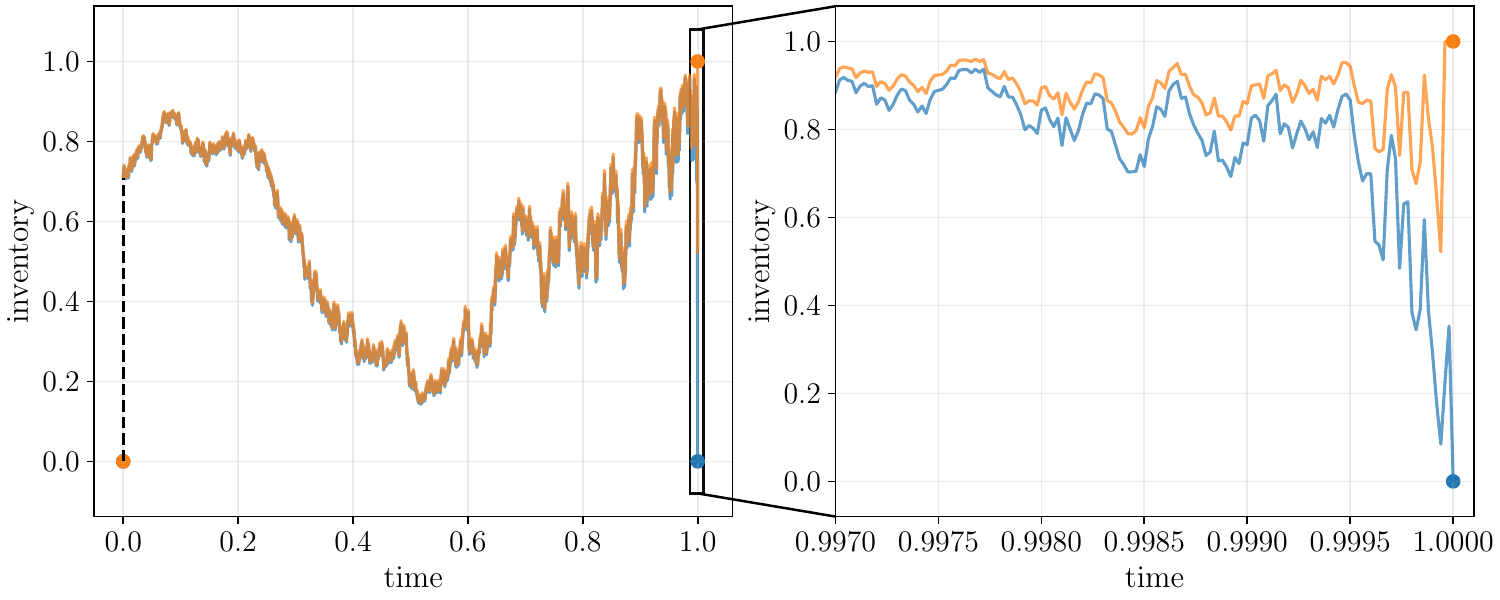}    
        \includegraphics[width=0.6\linewidth]{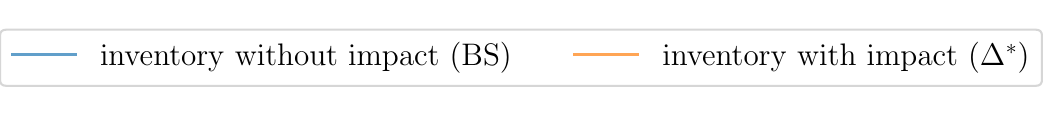}    

    \caption{Top panels: price process with and without the trading activity of the hedger. Bottom panels: inventory in the Black--Scholes model (no impact) against the inventory $\pos^*$ from our model.}
    \label{fig:S<K<P}
\end{figure}

\section{Conclusion} \label{sec:conclusion}

In this paper we studied the replication and pricing of European claims in the presence of price impact and execution costs. We covered the one-period and multiperiod binomial models as well as the  continuous-time problem. In continuous time, the replication problem gives rise to a nonlinear pricing PDE together with an implicit terminal condition that captures the effect of the hedger's terminal position on the settlement price (illustrating the moving target nature of the problem).  The hedger proceeds as follows: for a given payoff and number of contracts, they solve the implicit ODE \eqref{eq:implicit_ODE} for a modified payoff $\widetilde V$. They then solve the nonlinear PDE \eqref{eq:PDE} to obtain the modified liquidation value function $u$, with the replicating strategy given in feedback form by $\partial_x u$. With this in hand,  the hedger initiates their position through a bulk trade of $\partial_x u(0,S_0)$ shares and subsequently trades continuously according to $\pos^*_t = \partial_x u(t,S_t)$ until the maturity time $T$. At maturity, the terminal liquidation value of the hedging portfolio, accounting for both permanent price impact and execution costs, matches the gross payoff. The hedger then liquidates their terminal position $\pos_T^* = \partial_x u(T,S_T) = \widetilde V'(S_T)$ before settlement at time $T+$, and settles the contract by paying the resulting cash amount to the buyer.

In the midpoint regime $2\varphi=\lambda$, we established the well-posedness of the pricing problem and obtained the corresponding replicating strategy.
Our numerical examples illustrate that price impact introduces a nonlinear dependence on the number of contracts and substantially alters both option prices and hedging strategies. For the quadratic payoff case, we obtained all formulae in closed form, which served as a useful sanity check for the calculations in the paper. We envisage a variety of related directions for future work, ranging from more general models of price impact (e.g., propagator models) to multidimensional extensions and non-European claims.

\subsection*{Acknowledgements} We thank Alvaro Cartea, Johannes Muhle-Karbe and Thaleia Zariphopoulou for comments that improved the manuscript.

\subsection*{Code and Data Disclosure}\label{sec:Code and Data Disclosure}The code to support the numerical experiments in this paper is available upon request. The data that support the findings of this study are available from the TAQ database via WRDS and Databento. Restrictions apply to the availability of these data, which were used under license for this study.

\subsection*{AI Statement}  The authors used ChatGPT-6 Astra and Claude Opus 5.5 to assist with writing and coding. We also ran checks for the mathematical proofs in the paper. The authors take full responsibility for the accuracy, integrity, and originality of the work.

{ \small
\bibliographystyle{abbrvnat}
\setlength{\bibsep}{2.5pt}
\bibliography{references}
}

\appendix 

\section{Details on the market data calibration}\label{sec: appendix CalculationSpread}
For each option on SPY expiring between January 2025 and June 2025 we collect data for the preceding trading session $t_i = T_i-1$ (one day before expiry). We set $S_{t_i}$ to be the closing price. We compute the median spreads of the underlying and the options during the trading session of the preceding day, and the median volumes sitting at the best quotes for the underlying and the options during the trading session of the preceding day.

The up and down prices are constructed from the standard deviation
$\sigma_d$ of the preceding 60 daily log returns; more precisely, 
\[
    S_T(u)=S_0 e^{\sigma_d},
    \qquad
    S_T(d)=S_0 e^{-\sigma_d}.
\]

The risk-neutral probability and one-period net risk-free return are
\[
    q=\frac{S_0-S_T(d)}{S_T(u)-S_T(d)},
\]
and $r=0$ for simplicity (given the options expire in one day). The execution cost parameter $\varphi$ is calibrated from the median spread in the underlying and the median volumes sitting at best quotes as follows. Let $V^{S,a}$ and $V^{S,b}$ be the median volumes sitting at best quotes and let $\mathfrak{S}^S$ be half the spread. We take the midpoint regime $\lambda=2\,\varphi$ and we choose $\varphi$ such that it extrapolates the transaction costs associated with executing half of the volume sitting at best quotes (50\% weight for ask and 50\% weight for bid); that is, we set 
$$
\varphi = \frac{\mathfrak{S}^S}{2}\left(\frac{1}{V^{S,a}} + \frac{1}{V^{S,b}}\right).
$$
Let $V^{O,a}$ and $V^{O,b}$ be the median volume sitting at best quotes in the options market. Given that one option contract involves 100 shares of the underlying, in this experiment we wish to explain the spread associated with an option from our model written on 
$$
100\times \frac{1}{2}\left(V^{O,a}+V^{O,b}\right)
$$
shares of the underlying.
This corresponds to trying to explain the spread implied by the median volume sitting at the best quotes in the options market.

\section{Further examples for the one-period binomial model} \label{sec:ec_examples}

\subsection{Call option and square-root impact} 

In this section we derive an explicit replication strategy for the call option $V(x)= N(x-K)^+$ when the impact function is of square-root type, $I(x) = \lambda\, \mathrm{sign}(x)\sqrt{|x|}$ 
for parameter $\lambda > 0$. 
\begin{proposition} \label{prop:call_sqrt} For $K < S_T(u)$ let 
\[c_* = \frac{\lambda N + \sqrt{\lambda^2 N^2 + 4|N|(S_T(u)-S_T(d))(S_T(u)-K)}}{2(S_T(u)-S_T(d))}.\] Then the trading strategy
\begin{equation} \label{eq:pos_call_square_root}
\pos^*_{\min} = \begin{cases}
   0, & K \ge S_T(u), \\
   \mathrm{sign}(N)c_*^2, &  S_T(d) + \lambda\, \mathrm{sign}(N)\sqrt{|N|} \le K < S_T(u), \\
   N, & K < \min\{S_T(d) + \lambda\, \mathrm{sign}(N)\sqrt{|N|},S_T(u)\}
\end{cases}
\end{equation}
is the minimal-in-absolute-value solution to the fixed point equation \eqref{eq:pos_fixed_point}, and $\val_0^*$ is given by \eqref{eq: general val_0 1 period case} with $\pos^*_{\min}$ in place of $\pos^*$.
\end{proposition}

 \subsection{Put option} 

We now record the corresponding results for the European put, $V(x) = N(K-x)^+$. These results follow from the formulae obtained for the call. Indeed, suppose $I$ is  an odd function, as is the case for both the linear and square-root impact functions considered above. Then
$\pos$ solves \eqref{eq:pos_fixed_point} for the put payoff $x\mapsto N(K-x)^+$ with data
$(S_T(u),S_T(d),K)$ if and only if $-\pos$ solves \eqref{eq:pos_fixed_point} for the
call payoff $ x \mapsto N(x-\widetilde K)^+$ with data
$(\widetilde S_T(u),\widetilde S_T(d),\widetilde K)=(-S_T(d),-S_T(u),-K)$. The correspondence
maps minimal-in-absolute-value solutions to
minimal-in-absolute-value solutions, so that  Propositions~\ref{prop:call_linear} and
\ref{prop:call_sqrt} for the call option transfer directly to the put option.
\begin{proposition}
Let $I(x)=\lambda x$ for $\lambda>0$ and let $V(x)=N(K-x)^+$.  The minimal-in-absolute-value solution
to \eqref{eq:pos_fixed_point} is given by
\[
\pos^*_{\min}=
\begin{cases}
0, & K \le S_T(d), \\
-\frac{N(K-S_T(d))}{S_T(u)-S_T(d)-\lambda N},
  & S_T(d) < K \le S_T(u)-\lambda N,\\
-N, & K > \max\{S_T(u)-\lambda N,\ S_T(d)\},
\end{cases}
\]
and $\val_0^*$ is given by \eqref{eq: general val_0 1 period case} with $\pos^*_{\min}$ in place of $\pos^*$.
\end{proposition}

\begin{proposition}
Let $I(x)=\lambda\,\mathrm{sign}(x)\sqrt{|x|}$ for $\lambda>0$ and let $V(x)=N(K-x)^+$. For $K>S_T(d)$ let
\[
c_\star=\frac{\lambda N+\sqrt{\lambda^2N^2+4|N|(S_T(u)-S_T(d))(K-S_T(d))}}{2(S_T(u)-S_T(d))}.
\]
Then the trading strategy
\[
\pos^*_{\min}=
\begin{cases}
0, & K\le S_T(d),\\
-\mathrm{sign}(N)c_\star^2, & S_T(d)<K\le S_T(u)-\lambda\,\mathrm{sign}(N)\sqrt{|N|},\\
-N, & K>\max\{S_T(u)-\lambda\,\mathrm{sign}(N)\sqrt{|N|},\ S_T(d)\},
\end{cases}
\]
is the minimal-in-absolute-value solution to the fixed point equation \eqref{eq:pos_fixed_point}, and $\val_0^*$ is given by \eqref{eq: general val_0 1 period case} with $\pos^*_{\min}$ in place of $\pos^*$.
\end{proposition}

\newpage 
\setcounter{section}{0}
\renewcommand{\thesection}{EC.\Alph{section}}
This electronic companion contains proofs of the theoretical results in the manuscript.
\section{Proofs for Section~\ref{sec:one_period} and Appendix~\ref{sec:ec_examples}}
For the proofs in this section we note that if $V$ is monotone then any solution $\pos^*$ to \eqref{eq:pos_fixed_point} must have the same sign as the slope of $V$.

\begin{proof}[Proof of Proposition \ref{thm: Lipschitz implies existence}]
Define 
    $$
    f(x) = \frac{V(S_T(u) + I(x)) - V(S_T(d) + I(x))}{S_T(u)-S_T(d)} -x\,,
    $$
and note that $\pos^*$ satisfies \eqref{eq:pos_fixed_point} if and only if $f(\pos^*) = 0$. 
Using that $V$ is Lipschitz, we obtain the estimates 
\begin{align*}
f(x) & \leq -x + L_V \frac{|S_T(u)+I(x)-S_T(d)-I(x)|}{S_T(u)-S_T(d)}=-x+L_V\,.\\
f(x) & \geq -x - L_V \frac{|S_T(u)+I(x)-S_T(d)-I(x)|}{S_T(u)-S_T(d)}=-x-L_V\,.
\end{align*}
 These inequalities imply that $f(x) \ne 0$ if $|x| > L_V$, so that any solution must lie in $[-L_V,L_V]$, and that
$\lim_{x\to \pm \infty}f(x) = \mp \infty.$ Since $f$ is continuous it follows that there exists $\pos^* \in \R$ such that $f(\pos^*) =0$, which establishes the existence claim, as well as the property $|\pos^*| \le L_V$.

We now turn to the uniqueness claim. Without loss of generality, we assume that $V$ is nondecreasing, as the nonincreasing case is analogous. To establish uniqueness it suffices to show that $f$ is strictly decreasing. To this end, we let $x > y$ be given and we estimate that
\begin{align*}
f(x) - f(y)
& =   \frac{V(S_T(u) + I(x)) - V(S_T(u) + I(y))}{S_T(u) - S_T(d)}  - \frac{V(S_T(d) + I(x)) - V(S_T(d) + I(y))}{S_T(u) - S_T(d)} -\!(x-y) \\ 
& \le L_V\frac{S_T(u) + I(x) -(S_T(u) + I(y))}{S_T(u) - S_T(d)} - (x-y) \le (x-y)\bigg(\frac{L_VL_I}{S_T(u)-S_T(d)}-1\bigg) < 0.
\end{align*}
In the first inequality we applied the Lipschitz property of $V$ to the first term and we dropped the second term, which has a negative contribution because $V$ and $I$ are both nondecreasing. The final two inequalities followed from the Lipschitz property of $I$ and the condition \eqref{eq:uniqueness_condition}, respectively. This completes the proof. 
\end{proof}

    \begin{proof}[Proof of Proposition~\ref{prop:replicating_value}]
    To obtain \eqref{eq: general val_0 1 period case}, we take a solution $\pos^*$ and evaluate \eqref{eq:value1+} separately at state $u$ and state $d$. Multiplying the state $u$ equation by $\alpha$ and the state $d$ equation by $1-\alpha$, for some $\alpha \in \R$, we obtain 
\begin{align*}
    \val_0^* = \tfrac{1}{1+r}&\Big(\alpha\,V\big(P^{\pos^*}_T(u)\big) + (1-\alpha)V\big(P^{\pos^*}_T(d)\big)\\
    &\quad - \alpha\,\pos^*\,P_T^{\pos^*}(u) - (1-\alpha) \,\pos^*\,P_T^{\pos^*}(d) - \pos^*\,\phi(-\pos^*) + (1+r)\,\pos^*\big(S_0 + \phi(\pos^*)\big)\Big).
\end{align*}
The value of $\alpha$ that makes the second line equal zero is
\begin{align*}
    \alpha &= \frac{- \phi(-\pos^*) - P_T^{\pos^*}(d) + (1+r)(S_0 + \phi(\pos^*)) }{P_T^{\pos^*}(u) - P_T^{\pos^*}(d)} 
    = q + \frac{\manip(\pos^*)}{S_T(u)-S_T(d)}.
\end{align*}
Substituting back in for $\val^*_0$ and recalling that $\pos^*$ satisfies 
\eqref{eq:pos_fixed_point} establishes \eqref{eq: general val_0 1 period case}.

Next, we note that the no price manipulation condition \eqref{eq:no_price_manipulation} ensures $\pos^*\manip(\pos^*) \geq 0$, so that  $\val_0^* \ge \tfrac{1}{1+r}\mathbb{E}^{\mathbb{Q}} [V(P_T^{\pos^*})]$. Moreover, 
it is clear from \eqref{eq:pos_fixed_point} that $\pos^* \geq 0$ if $V$ is nondecreasing and $\pos^* \leq 0$ if $V$ is nonincreasing.  Therefore, $V(P_T^{\pos^*}(\omega)) = V(S_T(\omega) + I(\pos^*)) \geq V(S_T(\omega))$, for $\omega \in \{u,d\}$, from which we deduce that
$\val_0^* \ge  \frac{1}{1+r} \mathbb{E}^{\mathbb{Q}}[V(S_T)]$,
  completing the proof. 
\end{proof}
\begin{proof}[Proof of Proposition~\ref{prop:call_linear}] The payoff is Lipschitz with $L_V = |N|$, so by Proposition~\ref{thm: Lipschitz implies existence} there exists at least one solution $\pos^*$, and any solution satisfies $|\pos^*| \le |N|$. If $K 
\ge S_T(u)$ or $N=0$, then $\pos^* = 0$ is a solution, which is clearly minimal. For $N \ne 0$ and $K < S_T(u)$, the right-hand side of \eqref{eq:pos_fixed_point} is nonzero when $\pos = 0$, so zero is not a solution. Additionally, if both states are out of the money then the right-hand side of \eqref{eq:pos_fixed_point} is zero, so no solution can lead to both states being out of the money. If the call is in the money in both states then we have
\[\pos =  N \frac{S_T(u) + \lambda \pos - K - (S_T(d) + \lambda \pos -K)}{S_T(u) - S_T(d)} = N,\]
so that the replicating strategy is $\pos^* = N$, which is consistent with both states being in the money when $K < S_T(d) + \lambda N$. It just remains to study the case when state $u$ is in the money and state $d$ is out of the money.
This leads to the linear equation
$\pos = N \frac{S_T(u) + \lambda \pos - K }{S_T(u) - S_T(d)}$, which has the unique solution 
\begin{equation} \label{eq:pos*_proof} 
\pos^* =  \frac{N(S_T(u) -K)}{S_T(u)-S_T(d) - \lambda N}.
\end{equation}
 A direct computation gives
\[S_T(\omega) + \lambda \pos^* - K = \frac{(S_T(u) - S_T(d))}{S_T(u) - S_T(d) -\lambda N}\Big(\big(S_T(u)-K\big)1_{\{\omega = u\}} - \big(K-S_T(d)-\lambda N\big)1_{\{\omega= d\}}\Big)
\]
for $\omega \in \{u,d\}$. 
Hence, the solution \eqref{eq:pos*_proof} is consistent with state $u$ being in the money and state $d$ being out of the money if and only if  $S_T(u) > K  \ge S_T(d) + \lambda N$, in which case the denominator in \eqref{eq:pos*_proof} is positive. This shows that there is actually a unique solution for $K < S_T(u)$, from which the formula for $\pos^*_{\min}$ follows. 
\end{proof}

\begin{proof}[Proof of Proposition~\ref{prop:call_sqrt}]
    If $K \ge S_T(u)$ or $N = 0$, then $\pos = 0$ satisfies \eqref{eq:pos_fixed_point} and is clearly the minimal solution in this case.

   Next, we assume $N \ne 0$ and $K < S_T(u)$. If the call
were out of the money in both states, the right-hand side of
\eqref{eq:pos_fixed_point} would vanish. This would lead to $\pos=0$ and consequently
$S_T(u)<K$, so this case cannot happen. If  
the call is in the money in both outcomes then \eqref{eq:pos_fixed_point} becomes $\pos = N$, which is a valid, consistent solution when $ K < S_T(d) +\lambda\, \mathrm{sign}(N) \sqrt{|N|}$. The remaining case is when the call is in the money in state $u$, but out of the money in state $d$. In this case, \eqref{eq:pos_fixed_point} becomes
    \begin{equation} \label{eq:square_root_case2} 
    \pos =  \frac{N(S_T(u) + \lambda\, \mathrm{sign}(\pos)\sqrt{|\pos|}-K)}{S_T(u)-S_T(d)},
    \end{equation}
    from which it is clear that any solution $\pos^*$ must satisfy $\mathrm{sign}(\pos^*) = \mathrm{sign}(N)$. Multiplying \eqref{eq:square_root_case2} by $\mathrm{sign}(N)$ and setting $x = \sqrt{|\pos|} \ge 0$, we see that the fixed point equation becomes the quadratic equation $\Phi(x) = 0$, where
    \[\Phi(x) = \big(S_T(u) - S_T(d)\big)x^2 - \lambda Nx - |N|\big(S_T(u)-K\big).\]
    Since $\Phi$ opens upwards and $\Phi(0) < 0$, we have that $\Phi$ has one positive and one negative root, with the former being equal to $c_*$. Since $x \ge 0$, the unique solution is $\pos^* = \mathrm{sign}(N)c_*^2$.
  
    Using that $\pos^*$ and $N$
have the same sign gives 
\[S_T(u)+\lambda\,\mathrm{sign}(N)c_*-K=\frac{(S_T(u)-S_T(d))c_*^{2}}{|N|}>0,\]
which confirms that $\pos^*$ is consistent with the assumption we made that the call is in the money in state $u$. Next, subtracting $S_T(u)-S_T(d)$ from both sides in the above display gives 
\[
  S_T(d)+\lambda\,\mathrm{sign}(N)c_*-K=\frac{(S_T(u)-S_T(d))(c_*^{2}-|N|)}{|N|},\]
 which shows that the call is out of the money in state
$d$ if and only if $c_*^{2}\le|N|$.

The proof will be complete once we show that 
\begin{equation} \label{eq:c*_equivalence} 
c_*^2 \le |N| \iff 
    K \ge S_T(d) +\lambda\, \mathrm{sign}(N) \sqrt{|N|}, 
\end{equation}
which guarantees that the ranges appearing in \eqref{eq:pos_call_square_root} are comprehensive and pick out the correct solution in each case. Since $c_*$ is a positive root of the quadratic $\Phi$ and $\Phi(0)  <0$, we have for $x \ge 0$, that $\Phi(x) \ge 0$ if and only if $x \ge c_*$. Taking $x = \sqrt{|N|}$ gives $\Phi(\sqrt{|N|}) = |N|(K-S_T(d) - \lambda\, \mathrm{sign}(N)\sqrt{|N|})$ from which \eqref{eq:c*_equivalence} follows.  Since for any $K < S_T(u)$ we found a unique solution, 
this establishes the formula for $\pos^*_{\min}$. 
\end{proof}

\begin{proof}[Proof of Proposition~\ref{prop:digital_call}]
     If $K \leq S_T(d)$ then, at $\pos = 0$, the digital call is in the money in both states, while if $K >S_T(u)$ then it is out of the money in both states. In both of these cases it is easy to see that $\pos = 0$ solves \eqref{eq:pos_fixed_point} so that it is the minimal-in-absolute-value solution.

     Therefore, we restrict our attention to $S_T(d) < K \le S_T(u)$. Then we only need to consider the case when the call is in the money in state $u$, but out of the money in state $d$. Indeed, for all other cases the right-hand side of \eqref{eq:pos_fixed_point} collapses to zero, forcing $\pos = 0$, which clearly does not replicate the payoff since here $P_T^\pos = S_T$ and $S_T(d) < K \le S_T(u)$. Now, under the assumption that the digital call is in the money in state $u$ and out of the money in state $d$, \eqref{eq:pos_fixed_point} becomes
     \[\pos = \frac{N}{S_T(u)-S_T(d)}.\]
     With this unique candidate solution $\pos$, we see that the assumptions we made for the call to be in the money in state $u$, and out of the money for state $d$ are satisfied if and only if 
     \[S_T(d) + \frac{\lambda N}{S_T(u) - S_T(d)} < K \le S_T(u) + \frac{\lambda N}{S_T(u) - S_T(d)}.\]
     From here we readily deduce that \eqref{eq:pos_digital} is a valid solution whenever the strike $K$ is not in $J$. Since this exhausts all possibilities we see that no solution to \eqref{eq:pos_fixed_point} exists for the strike values in \eqref{eq:bad_strikes}. 
 \end{proof}

\section{Proofs and derivations for Section~\ref{sec:multi_period}} \label{sec:multi_period_proofs}

We start this section by using \eqref{eq:value_process_defn} to compute for any $m = 1,\dots,M$,
\begin{align} \val_{t_m} & = \cash_{t_m} + \pos_{t_m}(P_{t_m}^\pos - \varphi \pos_{t_m})  = \cash_{t_{m}-}(1+rh) + \pos_{t_m}(P_{t_{m}}^\pos - \varphi \pos_{t_{m}}) \nonumber\\
& = \cash_{t_{m-1}+}(1+rh) + \pos_{t_{m}}(P_{t_{m}}^\pos - \varphi \pos_{t_{m}}) \nonumber \\
& = \Big(\cash_{t_{m-1}} - (\pos_{t_{m}}- \pos_{t_{m-1}})\big(P_{t_{m-1}}^\pos + \varphi (\pos_{t_{m}} - \pos_{t_{m-1}})\big)\Big)(1+rh) + \pos_{t_{m}}(P_{t_{m}}^\pos - \varphi \pos_{t_{m}}) \nonumber  \\
& = \Big(\val_{t_{m-1}} - \pos_{t_{m}}\big(P_{t_{m-1}}^\pos + \varphi (\pos_{t_{m}} - \pos_{t_{m-1}})\big) +\varphi \pos_{t_{m-1}}\pos_{t_{m}}\Big)(1+rh)  + \pos_{t_{m}}(P_{t_{m}}^\pos - \varphi \pos_{t_{m}}) \label{eq:value_one_step_alt} \\
& = \val_{t_{m-1}}(1+rh) +\pos_{t_{m}}\big(S_{t_{m}}-S_{t_{m-1}}(1+rh)\big) -(2\varphi-\lambda)(1+rh) \pos_{t_{m}}(\pos_{t_{m}}-\pos_{t_{m-1}})  \nonumber \\
& \hspace{10.74cm}- rh(\lambda-\varphi)\pos_{t_{m}}^2. \label{eq:value_one_step}
\end{align}

\begin{proof}[Proof of Proposition~\ref{prop:linear_no_manip}]
Throughout we write $\beta := (1+rh)^{-1} \in (0,1]$ for the one-period discount factor and let $\pos$ be an admissible inventory process. Set $t_{M+1} = T+$, and note that admissibility gives $\pos_{t_0} = \pos_{t_{M+1}} = 0$.
Taking conditional expectation $\mathbb{E}^{\mathbb{Q}}_{t_{m-1}}[\cdot]$ in \eqref{eq:value_one_step}, and using that $\pos_{t_{m}}$ is $\mathcal{F}_{t_{m-1}}$-measurable by the left-continuity convention, while the discounted fundamental price is a $\mathbb{Q}$-martingale (so that $\mathbb{E}^{\mathbb{Q}}_{t_{m-1}}[S_{t_{m}}] = (1+rh)S_{t_{m-1}}$), we obtain the one-step relation
\begin{equation}\label{eq:one_step_cond}
  \val_{t_{m-1}}
  = \beta\mathbb{E}^{\mathbb{Q}}_{t_{m-1}}[\val_{t_{m}}]
    + (2\varphi-\lambda)\pos_{t_{m}}(\pos_{t_{m}}-\pos_{t_{m-1}})
    + rh(\lambda-\varphi)\beta\pos_{t_{m}}^2 .
\end{equation}
Multiplying \eqref{eq:one_step_cond} by $\beta^{m-1}$, taking full expectations and using the tower property gives, after rearranging,
\[
  \beta^{m}\mathbb{E}^{\mathbb{Q}}[\val_{t_{m}}] - \beta^{m-1}\mathbb{E}^{\mathbb{Q}}[\val_{t_{m-1}}]
  = -(2\varphi-\lambda)\mathbb{E}^{\mathbb{Q}}[\beta^{m-1}\pos_{t_{m}}(\pos_{t_{m}}-\pos_{t_{m-1}})]
    - rh(\lambda-\varphi)\mathbb{E}^{\mathbb{Q}}[\beta^{m}\pos_{t_{m}}^2].
\]
Summing over $m=1,\dots,M+1$, and using that no interest accrues between $T$ and $T+$, the left-hand side telescopes leading to
\begin{equation}\label{eq:disc_master_id}
  \beta^{M}\mathbb{E}^{\mathbb{Q}}[\val_T]
  = \val_0 -(2\varphi-\lambda)\mathbb{E}^{\mathbb{Q}}\bigg[\sum_{m=1}^{M+1}\beta^{m-1}\pos_{t_{m}}(\pos_{t_{m}}-\pos_{t_{m-1}})\bigg]
    - rh(\lambda-\varphi)\mathbb{E}^{\mathbb{Q}}\bigg[\sum_{m=1}^M\beta^m\pos_{t_m}^2\bigg].
\end{equation}
Next, we use the discrete-time integration by parts identity
\begin{equation}\label{eq:disc_ibp}
  \pos_{t_{m}}(\pos_{t_{m}}-\pos_{t_{m-1}})
  = \tfrac12(\pos_{t_{m}}^2 - \pos_{t_{m-1}}^2)
    + \tfrac12(\pos_{t_{m}}-\pos_{t_{m-1}})^2.
\end{equation}
Multiplying the difference of squares term by $\beta^{m-1}$ and summing leads to 
\begin{align*}
  \sum_{m=1}^{M+1}\beta^{m-1}(\pos_{t_{m}}^2-\pos_{t_{m-1}}^2) & = \sum_{m=1}^{M+1}(\beta^m\pos^2_{t_{m}} - \beta^{m-1}\pos^2_{t_{m-1}}) + \sum_{m=1}^{M+1}(\beta^{m-1} - \beta^{m})\pos^2_{t_{m}} \\
  & 
  =  rh\sum_{m=1}^{M}\beta^m\pos^2_{t_m},
\end{align*}
where in the last step we used the fact that $\beta^{m-1} - \beta^m = rh\beta^m$. 
Substituting this and \eqref{eq:disc_ibp} into \eqref{eq:disc_master_id} yields
\begin{align} \beta^{M}\,\mathbb{E}^{\mathbb{Q}}[\val_T]
  & = \val_0 -\tfrac{2\varphi-\lambda}{2}\mathbb{E}^{\mathbb{Q}}\bigg[ rh\sum_{m=1}^{M}\beta^m\pos^2_{t_m} + \sum_{m=1}^{M+1}\beta^{m-1}(\pos_{t_{m}}-\pos_{t_{m-1}})^2\bigg] \nonumber \\
  & \qquad \,\, - rh(\lambda-\varphi)\,\mathbb{E}^{\mathbb{Q}}\bigg[\sum_{m=1}^{M}\beta^m\pos^2_{t_m}\bigg] \nonumber \\
  & = \val_0 -\tfrac{2\varphi-\lambda}{2}\,\mathbb{E}^{\mathbb{Q}}\bigg[ \sum_{m=1}^{M+1}\beta^{m-1}(\pos_{t_{m}}-\pos_{t_{m-1}})^2\bigg]  - \tfrac{rh\lambda }{2}\E^\Q\bigg[\sum_{m=1}^{M} \beta^m \pos^2_{t_m}\bigg], \label{eq:discounted_liquidation_value}
\end{align} 
where we used the fact that $-\frac{1}{2}(2\varphi - \lambda)rh - rh(\lambda-\varphi) = -\frac{1}{2}rh\lambda$ 
to combine like terms. 
If $\val_0 = 0$, the condition \eqref{eq:linear_price_manip} ensures that the right-hand side is nonpositive completing the proof. 
\end{proof}
\begin{proof}[Proof of Corollary~\ref{cor:BS_bound}]
 The expression \eqref{eq:multiperiod_value_decomposition} follows from \eqref{eq:discounted_liquidation_value} and the replication condition \eqref{eq: fixed point equation}, which holds courtesy of Theorem~\ref{thm:binomial_replication} proved below. The estimate $\val_0^* \ge \E^\Q[V(P_T^{\pos^*})]/(1+rh)^M$ follows from the fact that the final two terms in \eqref{eq:multiperiod_value_decomposition} are nonnegative under the absence of price manipulation condition \eqref{eq:linear_price_manip}. For monotone $V$, the estimate $\val_0^* \ge \E^\Q[V(S_T)]/(1+rh)^M$ now follows in the same way as in the proof of Proposition~\ref{prop:replicating_value}. 
\end{proof}
\begin{proof}[Proof of Theorem~\ref{thm:binomial_replication}]
We start by deriving \eqref{eq:cash_value_discrete} before introducing the notation in the theorem statement.
To this end, we take $m \in \{1,\dots,M\}$ and suppose positions $\pos^*_{t_m}$ and liquidation values $\val^*_{t_m}$ have previously been solved for. We fix an outcome $\omega^{m-1}$ observed up to time $t_{m-1}$. To obtain a concise expression for the liquidation value at the previous trading time $\val_{t_{m-1}}^*(\omega^{m-1})$, we proceed in a similar fashion as in the proof of Proposition~\ref{prop:replicating_value}. That is, we multiply \eqref{eq:value_one_step_alt} when the next outcome of chance is $\omega_{m} = u$ by $\alpha \in \R$, and we add it to $1-\alpha$ multiplied by \eqref{eq:value_one_step_alt} evaluated at the outcome $\omega_{m} = d$, to obtain 
\begin{align*} 
\val_{t_{m-1}}^*(\omega^{m-1}) =&  \tfrac{1}{1+rh}\bigg(\alpha  \val_{t_{m}}^*(\omega^{m-1}u) + (1-\alpha)\val_{t_{m}}^*(\omega^{m-1} d) \\
 & \qquad \qquad  -\pos_{t_{m}}^*(\omega^{m-1})\Big(\alpha P_{t_{m}}^{\pos^*}(\omega^{m-1}u) + (1-\alpha) P_{t_{m}}^{\pos^*} (\omega^{m-1} d) - \varphi 
\pos^*_{t_{m}}(\omega^{m-1}) \\
&\qquad \qquad  + \big(2\varphi \pos^*_{t_{m-1}}(\omega^{m-2}) -\varphi\pos_{t_{m}}^*(\omega^{m-1}) - P_{t_{m-1}}^{\pos^*}(\omega^{m-1})\big)(1+rh)\Big)\bigg).
\end{align*}
The choice of $\alpha$ that makes the final two lines vanish is 
\begin{align*} \alpha & = \frac{S_{t_{m-1}}(\omega^{m-1})(1+rh) - S_{t_{m}}(\omega^{m-1} d)}{S_{t_{m}}(\omega^{m-1}u) - S_{t_{m}}(\omega^{m-1}d)} \\
& \qquad + \frac{ (2\varphi-\lambda +\varphi rh)\pos^*_{t_{m}}(\omega^{m-1})+(\lambda - 2\varphi)(1+rh)\pos^*_{t_{m-1}}(\omega^{m-2})}{S_{t_{m}}(\omega^{m-1}u) - S_{t_{m}}(\omega^{m-1}d)}. 
\end{align*}
 As such, substituting back in yields
\begin{align*}
&\val_{t_{m-1}}^*(\omega^{m-1})\\
& \quad = \tfrac{1}{1+rh}\E^\Q\big[\val_{t_{m}}^*\mid \omega^{m-1}\big]  \\
 & \qquad + \frac{\val_{t_{m}}^*(\omega^{m-1}u) - \val_{t_{m}}^*(\omega^{m-1}d)}{S_{t_{m}}(\omega^{m-1}u) - S_{t_{m}}(\omega^{m-1}d)}\bigg(\frac{(2\varphi-\lambda +\varphi rh)}{1+rh}\pos^*_{t_{m}}(\omega^{m-1}) +(\lambda- 2\varphi)\pos^*_{t_{m-1}}(\omega^{m-2})\bigg)  \\
 & \quad  = \frac{\E^\Q\big[\val_{t_{m}}^*\mid \omega^{m-1}\big]}{1+rh} +\pos^*_{t_{m}}(\omega^{m-1})\bigg(\frac{(2\varphi-\lambda +\varphi rh)}{1+rh}\pos^*_{t_{m}}(\omega^{m-1}) -( 2\varphi-\lambda)\pos^*_{t_{m-1}}(\omega^{m-2})\bigg),
\end{align*}
where in the final equality we used that $\pos^*_{t_{m}}(\omega^{m-1})$ satisfies \eqref{eq:multiperiod_position_equation}. This establishes \eqref{eq:cash_value_discrete}. We now substitute \eqref{eq:cash_value_discrete} into \eqref{eq:multiperiod_position_equation} at the previous time-index $t_{m-1}$, which gives
\begin{align*}
\pos^*_{t_{m-1}}(\omega^{m-2}) & = \frac{1}{1+rh}\frac{\E^\Q[\val^*_{t_{m}}\mid \omega^{m-2}u] - \E^\Q[\val^*_{t_{m}}\mid \omega^{m-2}d]}{S_{t_{m-1}}(\omega^{m-2}u)  - S_{t_{m-1}}(\omega^{m-2}d)} \\
& \qquad + \frac{2\varphi-\lambda +\varphi rh}{1+rh}\frac{(\pos^*_{t_{m}})^2(\omega^{m-2}u) - (\pos^*_{t_{m}})^2(\omega^{m-2}d)}{S_{t_{m-1}}(\omega^{m-2}u)  - S_{t_{m-1}}(\omega^{m-2}d)} \\
& \qquad - (2\varphi-\lambda) \pos^*_{t_{m-1}}(\omega^{m-2})\frac{\pos^*_{t_{m}}(\omega^{m-2}u) -\pos^*_{t_{m}}(\omega^{m-2}d)}{S_{t_{m-1}}(\omega^{m-2}u)  - S_{t_{m-1}}(\omega^{m-2}d)},
\end{align*}
with the understanding that $\omega^{m-2}$ for $m=1$ is the empty outcome before any randomness has been observed.
This is a linear equation for $\pos^*_{t_{m-1}}(\omega^{m-2})$ and has a unique solution if and only if \eqref{eq:nonzero_denom} holds, in which case it is given by 
\begin{equation} \label{eq:strategy_discrete}
\begin{split} 
\pos^*_{t_{m-1}}(\omega^{m-2}) = &  \frac{1}{1+rh}\bigg(\frac{\E^\Q[\val^*_{t_{m}}\mid \omega^{m-2}u] - \E^\Q[\val^*_{t_{m}}\mid \omega^{m-2}d]}{S_{t_{m-1}}(\omega^{m-2}u)  - S_{t_{m-1}}(\omega^{m-2}d) + (2\varphi-\lambda)(\pos^*_{t_{m}}(\omega^{m-2}u) - \pos^*_{t_{m}}(\omega^{m-2}d))}\\
&  + \frac{(2\varphi-\lambda +\varphi rh)((\pos^*_{t_{m}})^2(\omega^{m-2}u) - (\pos^*_{t_{m}})^2(\omega^{m-2}d))}{S_{t_{m-1}}(\omega^{m-2}u)  - S_{t_{m-1}}(\omega^{m-2}d) + (2\varphi-\lambda)(\pos^*_{t_{m}}(\omega^{m-2}u) - \pos^*_{t_{m}}(\omega^{m-2}d))}\bigg).  
\end{split}
\end{equation}

We now establish Theorem~\ref{thm:binomial_replication}, arguing by backward induction on $m$. For the base case,  Proposition~\ref{thm: Lipschitz implies existence} guarantees a solution $f_M(s)$ for every $s \in \mathcal{S}_{M-1}$. The replication condition \eqref{eq: fixed point equation} then gives $\val^*_{t_M} = V(S_{t_M} + \lambda f_M(S_{t_{M-1}})) = g_M(S_{t_{M-1}},S_{t_M})$ for $\pos^*_{t_M} = f_M(S_{t_{M-1}})$.
For the inductive step, under the condition \eqref{eq:denominator_finite}, the functions $f_m$ and $g_m$ are well-defined for every $m=0,\dots,M$. The induction hypothesis and definition of $f_m,g_m$ show that $(\pos^*,\val^*)$ given by \eqref{eq:cash_value_discrete} and \eqref{eq:strategy_discrete} correspond precisely to $(\pos^*,\val^*)$ given in the statement of the theorem. This establishes the result. 
\end{proof}

\section{Proofs for Section~\ref{sec:cont_time}} \label{sec:proof_cont}
We start by rigorously establishing the limiting dynamics \eqref{eq:cash_value_dynamics}.
\begin{proposition} \label{prop:limit}
        Let $(\pos_t)_{t \in \mathcal{T}}$ be an admissible strategy.
        For each $M \in \mathbb{N}$ set $t_m = \frac{mT}{M}$ for $m=0,\dots,M$, and define the discretized processes 
        \begin{align*}
            S_t^M & = {\textstyle\sum}_{m=0}^{M-1} S_{t_m}1_{[t_m,t_{m+1})}(t) + S_T1_{\{T\}}(t) \qquad \text{for } t \in [0,T], \qquad  S_{T+}^M = S_T^M\\
            \pos^M_t & = {\textstyle \sum}_{m=0}^{M-1}\pos_{t_m}1_{(t_m,t_{m+1}]}(t) \qquad \qquad \qquad \quad \,\,\! \text{for } t \in [0,T],  \qquad \pos^M_{T+} = 0.
        \end{align*}
        Let $\val^M$ denote the corresponding discrete-time liquidation value process; i.e., $\val^M$ is given by the right-hand side of \eqref{eq:value_process_defn} with $\pos^M$ in place of $\pos$. Then, $\val^M_t$ converges in probability at each fixed time $t$ as $M \to \infty$. These limits admit a c\`agl\`ad version $(\val_t)_{t \in \mathcal{T}}$, which satisfies \eqref{eq:cash_value_dynamics}.
\end{proposition}
 \begin{proof}[Proof of Proposition~\ref{prop:limit}]
 Throughout we write $\beta := (1+rh)^{-1} \in (0,1]$ for the one-period discount factor and perform similar manipulations as in the proof of Proposition~\ref{prop:linear_no_manip}.
 Indeed, recalling the expression \eqref{eq:value_one_step} for the process $\val^M$ and rearranging allows us to write
\[
  \beta\val^M_{t_{m}} -\val^M_{t_{m-1}}= \pos_{t_{m}}(\beta S_{t_{m}}-S_{t_{m-1}})
    - (2\varphi-\lambda)\pos_{t_{m}}(\pos_{t_{m}}-\pos_{t_{m-1}})
    - rh(\lambda-\varphi)\beta\pos_{t_{m}}^2 .\]
Now fix $t \in [0,T]$ and let $\overline m^M = \lfloor tM/T\rfloor$, which satisfies $t_{\overline m^M} \to t$ as $M \to \infty$.
We now multiply by $\beta^{m-1}$ on both sides and sum from $m = 1$ to $\overline m^M$ to obtain
\begin{equation} \label{eq:discrete_disc_val}
\begin{aligned} 
\beta^{\overline m^M} \val^M_{t_{\overline m^M}} - \val^M_0 & = \sum_{m=1}^{\overline m^M} \pos_{t_{m}}(\beta^{m}S_{t_{m}} - \beta^{m-1} S_{t_{m-1}}) - (2\varphi- \lambda)\sum_{m=1}^{\overline m^M}\beta^{m-1} (\pos_{t_{m}} - \pos_{t_{m-1}})^2 \\
& \qquad - (2\varphi- \lambda) \sum_{m=1}^{\overline m^M}\beta^{m-1} \pos_{t_{m-1}}(\pos_{t_{m}} - \pos_{t_{m-1}}) - r(\lambda - \varphi)\sum_{m=1}^{\overline m^M} \beta^{m}\pos_{t_{m}}^2 h,
 \end{aligned}
 \end{equation}
 where we used the identity 
 \[\pos_{t_{m}}(\pos_{t_{m}} - \pos_{t_{m-1}}) = (\pos_{t_{m}} - \pos_{t_{m-1}})^2 + \pos_{t_{m-1}}(\pos_{t_{m}} - \pos_{t_{m-1}}).\]
 Recalling that $h = t_{m+1} - t_m $ we see that $\beta^m = (1+\frac{rT}{M})^{-m} = (1+\frac{rt_m}{m})^{-m}$. As such, if $t_m \to u$ as $M \to \infty$ then $\beta^m \to e^{-ru}$. Standard convergence results for discretized Riemann sums and quadratic variations imply, for each fixed $t$, that the right-hand side of \eqref{eq:discrete_disc_val} converges to  
 \begin{equation} \int_0^t \pos_u \d(e^{-ru}S_u) -r(\lambda-\varphi)\int_0^te^{-ru}\pos_u^2 \d u - (2\varphi-\lambda)\int_{(0,t]} e^{-ru}(\pos_u \d\pos^+_u + \d[\pos^+]_u) \label{eq:disc_val}
 \end{equation}
  as $M \to \infty$ (see, e.g., \cite[Proposition~I.4.44 and Theorem~I.4.47]{jacod2003limit}). As such, the left-hand side of \eqref{eq:discrete_disc_val} converges as well to $e^{-rt}\widetilde \val_t - \widetilde \val_0$, where we  call $\widetilde \val_t$ the limit of $\val^M_t$, which is c\`adl\`ag since the right-hand side of \eqref{eq:disc_val} is. 
 To simplify the expression \eqref{eq:disc_val} note by It\^o's formula and \eqref{eq:local_vol} that 
 \begin{equation} \label{eq:discounted_price}
     \d(e^{-ru}S_u) = e^{-ru}(-rS_u \d u + \d S_u) = e^{-ru}\sigma(u,S_u)S_u \d W_u.
 \end{equation} Next, we compute
 \[\d\big(e^{-ru}(\pos^+_u)^2\big) = e^{-ru}(-r\pos_u^2 \d u + 2\pos_u \d \pos^+_u + \d[\pos^+]_u), \] so that 
     \begin{equation} \label{eq:QV_strategy_identity}
     \int_{(0,t]} e^{-ru}\pos_u \d\pos_u^+ + \frac{1}{2}\int_{(0,t]} e^{-ru}\d[\pos^+]_u = \frac{1}{2}e^{-rt}\pos_{t+}^2 - \frac{1}{2}\pos_{0+}^2 + \frac{r}{2}\int_0^t e^{-ru}\pos_u^2 \d u.
     \end{equation}
Substituting \eqref{eq:discounted_price} and \eqref{eq:QV_strategy_identity} into \eqref{eq:disc_val}, and recalling that \eqref{eq:disc_val} was equal to $e^{-rt} \widetilde \val_t - \widetilde \val_0$, yields 
\begin{align*} 
e^{-rt}\widetilde \val_t &  = \widetilde \val_0 +\int_0^t e^{-ru}\pos_u\sigma(u,S_u)S_u \d W_u - \frac{r\lambda}{2}\int_0^te^{-ru}\pos_u^2 \d u  \\
& \hspace{6cm}-\frac{2\varphi -\lambda}{2}\bigg(e^{-rt}\pos_{t+}^2 - \pos_{0+}^2 +\int_{(0,t]}e^{-rs} \d[\pos^+]_s\bigg).
\end{align*}
We can absorb the term $\frac{2\varphi-\lambda}{2}\pos^2_{0+}$ into the quadratic variation integral by using the convention $\int_{\{0\}}e^{-rs}\d[\pos^+]_s = \pos^2_{0+}$. Passing to the left-continuous modification $\val$ then gives \eqref{eq:cash_value_dynamics}, completing the proof.  \end{proof}

\begin{proof}[Proof of Proposition~\ref{prop:price_manip_cont_time}] First we show sufficiency of the conditions on the parameters. Fixing an admissible strategy $\pos$ and
taking expectation in \eqref{eq:cash_value_dynamics} at time $T$ and with $\val_0 = 0$ leads to
\begin{equation} \label{eq:expected_discounted_value}
\E^\Q[e^{-rT}\val_T] = -\frac{r\lambda}{2}\int_0^T e^{-rt}\E^\Q[\pos_t^2]\d t - \frac{1}{2}(2\varphi-\lambda)\E^\Q\bigg[e^{-rT}\pos_T^2 + \int_{[0,T)} e^{-rt}\d[\pos^+]_t\bigg] \le 0.
\end{equation}
 The stochastic integral vanished in the equality as it is a true martingale, guaranteed by the admissibility condition on $\pos$. The inequality then follows by the fact that $r \ge 0$, handling the time integral term, and the condition $2\varphi \ge \lambda$, which ensures the contribution of the other terms is nonpositive. 

To establish necessity first suppose that $r < 0$. Consider the deterministic, finite variation, continuous and bounded strategy $\pos_t = \min\{t,T-t\}$. This strategy is clearly admissible and from the equality in \eqref{eq:expected_discounted_value} we obtain
\[\E^\Q[e^{-rT}\val_T] = -\frac{r\lambda}{2}\int_0^T e^{-rt}\min\{t^2,(T-t)^2\}\d t > 0,\]
due to negativity of $r$. This is price manipulation.

 Next, suppose that $2\varphi < \lambda$. In this case we take $\pos_t = 1_{(T_0,T_1]}(t)$ for some $0 < T_0 < T_1 < T$. With this choice we have from the equality in \eqref{eq:expected_discounted_value} that
\[\E^\Q[e^{-rT}\val_T] = -\frac{r\lambda}{2}\int_{T_0}^{T_1}e^{-rt}\d t - \frac{1}{2}(2\varphi-\lambda)(e^{-rT_0} + e^{-rT_1}) = (\lambda-\varphi)e^{-rT_1} - \varphi e^{-rT_0}.\]
We see that as $T_0 \to T_1$, the right-hand side converges to $(\lambda - 2\varphi)e^{-rT_1} > 0$. Then, by taking $T_0$ and $T_1$ close enough we can generate price manipulation if $2\varphi < \lambda$. This establishes necessity and completes the proof. 
\end{proof}

\begin{proof}[Proof of Theorem~\ref{thm:verification}] 
The polynomial bound \eqref{eq:polynomial_bound} and Assumption~\ref{ass:sigma} ensure that we have $\E[\int_0^T (\pos_t^*)^2\sigma(t,S_t)^2S_t^2dt] < \infty$. 
 On every compact rectangle contained in $[0,T)\times(0,\infty)$, uniform continuity of $\partial_t u$ and $\partial_{xx}u$ gives
$\partial_x u(t+h,x)-\partial_x u(t,x)=o(\sqrt h)$
uniformly. As such, time increments contribute zero quadratic variation,  giving the usual quadratic variation formula
\[
    [(\pos^*)^+]_t = \int_0^t \big(\partial_{xx} u(s,S_s)\big)^2 \sigma^2(s,S_s)S_s^2 \,\d s,
    \qquad t \in [0,T),
\]
despite $(\pos^*)^+$ not necessarily being a semimartingale.
Applying It\^o's formula to $e^{-rt}u(t,S_t)$ leads to the dynamics  \eqref{eq:Ito_cash_value}. Using that $u$ solves the PDE \eqref{eq:PDE} leads to 
\begin{equation}
\begin{aligned} \label{eq:discounted_u}
\d\big(e^{-rt}u(t,S_t)\big) & = -e^{-rt}\bigg(\frac{r\lambda}{2}\big(\partial_x u(t,S_t)\big)^2 + \frac{2\varphi-\lambda}{2}\big(\partial_{xx}u(t,S_t)\big)^2 \sigma^2(t,S_t)S_t^2\bigg)\d t \\
& \quad + e^{-rt}\sigma(t,S_t)S_t\partial_x u(t,S_t)\d W_t.
\end{aligned}
\end{equation}
If $2\varphi > \lambda$, integrating \eqref{eq:discounted_u} and letting $t\uparrow T$, continuity of $u$ and \eqref{eq:polynomial_bound} imply that the nonnegative integral involving $(\partial_{xx}u)^2$ has a finite limit.
From here it follows that $\pos^*$ is an admissible trading strategy. 
From \eqref{eq:modified_cash_value_dynamics} we see for $t \in (0,T)$ that 
\begin{align} 
e^{-rt} \valm_t^* & = e^{-rt}\Big(\val_t^* + \tfrac{2\varphi-\lambda}{2}\big(\partial_x u(t,S_t)\big)^2\Big) \nonumber \\
& = \val^*_{0+} + \tfrac{2\varphi-\lambda}{2}\big(\partial_x u(0,S_0)\big)^2 + \int_0^t e^{-rs}\sigma(s,S_s)S_s\partial_x u(s,S_s)\d W_s \nonumber \\
& \qquad \quad \, -  \int_0^t  e^{-rs}\Big(\tfrac{r\lambda}{2}\big(\partial_x u(s,S_s)\big)^2 + \tfrac{2\varphi-\lambda}{2}\big(\partial_{xx}u(s,S_s)\big)^2 \sigma^2(s,S_s)S_s^2\Big)\d s \nonumber \\
& = \val^*_{0+} + \tfrac{2\varphi-\lambda}{2}\big(\partial_x u(0,S_0)\big)^2 +  e^{-rt}u(t,S_t) - u(0,S_0), \label{eq:matching_values}
\end{align}
where we substituted \eqref{eq:discounted_u} in the final step.
By continuity of $\val^*$ and $\pos^*$ at time $T$ we have that $\valm^*$ is continuous at time $T$ too. As such, since $u$ satisfies property \ref{item:T_continuity}, by sending  $t \uparrow T$ in \eqref{eq:matching_values} we obtain $\Q$-a.s.\ the relationship
\[e^{-rT}\big(\valm_T^* -u(T,S_T)\big) = \val^*_{0+} + \tfrac{2\varphi-\lambda}{2}(\partial_x u(0,S_0))^2  - u(0,S_0) = 0,\]
where the final equality holds since the prescribed initial value $\val_0^* = u(0,S_0) + \frac{2\varphi-\lambda}{2}(\partial_xu(0,S_0))^2$ together with \eqref{eq:cash_value_dynamics} at $t=0$ and $t=0+$ gives $\val^*_{0+} = \val^*_0 - (2\varphi-\lambda)(\partial_xu(0,S_0))^2$. We have shown that $\valm_T^* = u(T,S_T)$, $\Q$-a.s. Using now the definition \eqref{eq:modified_cash_value} of $\valm^*$, the fact that $S_T$ has a density and that $u(T,\cdot)$ satisfies \eqref{eq:terminal_condition} for a.e.\ $x$ we obtain
\begin{align*} 
0 = \valm_T^* - u(T,S_T) & = 
\val_T^* + \tfrac{2\varphi - \lambda}{2}\big(\partial_x u(T,S_T)\big)^2 - V\big(S_T + \lambda \partial_x u(T,S_T)\big) - \tfrac{2\varphi-\lambda}{2}\big(\partial_x u(T,S_T)\big)^2 \\
& = \val_T^* - V(P_T^{\pos^*}),
\end{align*}
where in the last step we recalled that $\pos_T^* = \partial_x u(T,S_T)$.
This completes the proof.  \end{proof}

\subsection{Terminal ODE} \label{sec:terminal_ODE}

The goal of this section is to prove Theorem~\ref{thm:terminal}. We start with a technical well-posedness lemma, studying an auxiliary ODE from which we will construct the amended payoff $\widetilde V$.

\begin{lemma} \label{lem:ODE}
     Assume that $\varphi > 0$. Let a Lipschitz continuous, monotone and convex payoff $V: [0,\infty) \to \R$ be given. The right-derivative $V'$ is defined on $[0,\infty)$, right-continuous, nondecreasing and uniformly bounded by $L_V$, the Lipschitz constant of $V$. Set $\varepsilon = 1$ if $V$ is nondecreasing or constant, and $\varepsilon = -1$ if $V$ is nonincreasing and nonconstant. Extend $V$ to all of $\R$ via $V(x) = V(0)$ for $x < 0$.  
    
    It follows that there exists a unique absolutely continuous and bounded function $ P: \R \to \R$ satisfying 
    \begin{equation} \label{eq:P_ODE} 2\varphi P'(\xi)P(\xi) = P(\xi) - V'(\xi) \qquad \text{for a.e.\ } \xi \in \R.
    \end{equation} The function $P$ satisfies $ 0 \le \varepsilon P \le L_V$. 
    
    \begin{enumerate}
        \item If $2\varphi - \lambda > 0$ then the map $\Xi(\xi) = \xi - \lambda P(\xi)$ is a strictly increasing continuous bijection of $\R$, and the function
    \begin{equation} \label{eq:tilde_V}
        \widetilde V(x) = V\big(\Xi^{-1}(x)\big)+\tfrac{2\varphi-\lambda}{2}P^2\big(\Xi^{-1}(x)\big), \qquad x \in \R
    \end{equation} belongs to $C^1(\R)$ with $\widetilde V'(x) = P(\Xi^{-1}(x))$.
    \item If $2\varphi = \lambda$ then $\Xi$ is nondecreasing, surjective on $\R$, and $V(\xi)$ gives the same value for any $\xi \in \Xi^{-1}(x)$ and any $x \in \R$. Thus, \eqref{eq:tilde_V} is well-defined. In this case $\widetilde V$ is Lipschitz continuous, and at points of differentiability satisfies $\widetilde V'(x) = P(\Xi^{-1}(x))$, where the right-hand side is well-defined. Moreover, if $\widetilde V$ is differentiable at $x$ then $\Xi^{-1}(x)$ is a singleton.
    \end{enumerate} 
Moreover, $\varepsilon \widetilde V$ is nondecreasing, the Lipschitz constant of $\widetilde V$ is at most $L_V$, and in the case $\varepsilon = -1$, we have $\widetilde V(x) = \widetilde V(0)$ for $x < 0$ and $\widetilde V'(x) \ge -x/\lambda$ for a.e.\ $x > 0$.
\end{lemma}
\begin{proof}
The proof is broken up into ten steps.    \paragraph{Step 1: Nonnegativity of $\varepsilon P$.} First, we show that any bounded solution to \eqref{eq:P_ODE} must satisfy $\varepsilon P \ge 0$. We start with the $\varepsilon = 1$ case and argue by contradiction, assuming that $P(\xi_0) < 0$ for some $\xi_0 \in \R$. Denote by $I$ the maximal interval containing $\xi_0$ on which $P$ is negative. 
    Dividing the equation \eqref{eq:P_ODE} by $2\varphi P(\xi)$ and using nonnegativity of $V'$ we see for any $\xi \in I$ that 
    \[P'(\xi) = \frac{1}{2\varphi}\bigg(1-\frac{V'(\xi)}{P(\xi)}\bigg) \ge \frac{1}{2\varphi}. \]
    It follows that $P$ is strictly increasing on $I$ and since $\xi_0$ is in the interior of $I$, it follows that $\inf I = -\infty$. Now for any $\xi \le \xi_0$ we obtain
    \[P(\xi_0) - P(\xi)  = \int_{\xi}^{\xi_0} P'(z)\d z \ge \frac{1}{2\varphi}(\xi_0 - \xi).\]
    Sending $\xi \to - \infty$ shows that $P(\xi) \to -\infty$, contradicting boundedness of $P$. The case $\varepsilon = -1$ is symmetric; that is,  if $P(\xi_0) > 0$, then $V' \le 0$ gives $P' \ge \frac{1}{2\varphi}$ on the corresponding maximal interval $I$, so that $\sup I = \infty$ and $P(\xi) \to \infty$. This is again a contradiction. 
   \paragraph{Step 2: The nonincreasing case $ \varepsilon = -1$.} This case can be handled directly via the substitution $R = \varphi P^2$. Here, the ODE \eqref{eq:P_ODE} becomes 
\begin{equation} \label{eq:R_ODE}
R'(\xi) = -\sqrt{R(\xi)/\varphi} - V'(\xi).
\end{equation}
    Since $V' = 0$ on $(-\infty,0)$ we have on this interval that \eqref{eq:R_ODE} reads $R' = -\sqrt{R/\varphi} \le 0$. As such, if $R(\xi_1) > 0$ for some $\xi_1 < 0$ then $R \ge R(\xi_1)$ on $(-\infty,\xi_1]$, so that $R' \le - \sqrt{R(\xi_1)/\varphi} < 0$. Since $R$ is strictly decreasing with a strictly negative rate it follows that $\lim_{\xi \to -\infty} R(\xi) = \infty$. This contradicts boundedness of $P$ and shows that $R$ must be identically zero on $(-\infty,0]$.
    
    Next, we note by Carath\'eodory's Existence Theorem \cite[Theorem~I.5.1]{hale2009ordinary} that \eqref{eq:R_ODE} on $[0,\infty)$ with the initial condition $R(0) = 0$, has a local solution. When $R(\xi) = 0$ the right-hand side of \eqref{eq:R_ODE} is $-V'(\xi) \ge 0$, whereas when $R(\xi) = \varphi L_V^2$ we have that the right-hand side is $-L_V - V'(\xi) \le 0$. Since the right-hand side points inwardly at these extreme values we see that the solution is nonexplosive and satisfies $R(\xi) \in [0,\varphi L_V^2]$ for all $\xi \ge 0$. From here, we see that $P$ is globally defined and $-L_V \le P(\xi) \le 0$ for all $\xi \in \R$.

    Since the right-hand side of \eqref{eq:R_ODE} is not Lipschitz continuous at $R = 0$, uniqueness is not immediate from Carath\'eodory's Theorem. Instead we use monotonicity of $R \mapsto - \sqrt{R/\varphi}$ by taking two solutions $R_1,R_2$ and directly observing that $\frac{\d}{\d\xi}(R_1 - R_2)^2 \le 0$, which implies uniqueness. Since $x \mapsto \varphi x^2$ is a bijective map from $(-\infty,0]$ to $[0,\infty)$, this implies uniqueness of \eqref{eq:P_ODE} as well. Finally, it remains to verify that $P$ is absolutely continuous on a right-neighbourhood of zero. If $V'(0+) = 0$ then, since $V' \le 0$ and nondecreasing, it must be globally constant, contradicting the definition of $\varepsilon = -1$.
    Hence, we focus on the case $R'(0) = -V'(0+) > 0$. Fixing $\delta > 0$ and using that $|V'|$ is nonincreasing, we see that $R' \ge |V'(\delta)|/2$ whenever $R(\xi) \le \varphi (V'(\delta))^2/4$ from which we deduce that $R(\xi) \ge \frac{|V'(\delta)|}{2}\xi$ on a sufficiently small right-neighbourhood of zero. It then follows that $|P'(\xi)| = \frac{R'(\xi)}{2\sqrt{\varphi R(\xi)}} \le C \xi^{-1/2}$ for some $C > 0$ on this interval establishing integrability of $P'$, and hence absolute continuity of $P$.

    \paragraph{Step 3: Upper bound on $P$.}  With the construction of $P$ completed in the $\varepsilon = - 1$ case, we restrict to  $\varepsilon = 1$ until step 8 below. 
Since $V$ is convex, $V'$ is nondecreasing and so it has a limiting value $c_+ := \lim_{\xi \to \infty} V'(\xi) \le L_V$. We now argue that any bounded solution to \eqref{eq:P_ODE} must satisfy $0 \le P(\xi) \le c_+$ for all $\xi \in \R$. Indeed, suppose by contradiction that there exists $\xi_0$ such that $P(\xi_0) > c_+$. Then, note that
    \[P'(\xi_0) = \frac{1}{2\varphi} \bigg(1 - \frac{V'(\xi_0)}{P(\xi_0)}\bigg) \ge \frac{1}{2\varphi} \bigg(1 - \frac{c_+}{P(\xi_0)}\bigg) > 0. \]
    Arguing in a similar fashion as the nonnegativity of $P$ case, we can establish that $P$ is increasing on $[\xi_0,\infty)$ and that $P'(\xi) \ge \frac{1}{2\varphi} (1 - c_+/P(\xi_0))$ for $\xi \ge \xi_0$. This shows that $P(\xi) \to \infty$ as $\xi \to \infty$; a contradiction.

     \paragraph{Step 4: $P$ dominates $V'$.} Next, we claim that any bounded solution satisfies $P \ge V'$. Suppose by way of contradiction that $P(\xi_0) < V'(\xi_0)$ for some $\xi_0$. By continuity of $P$ and nondecreasingness of $V'$ there exists an interval $J$ containing $\xi_0$ such that $P(\xi) < V'(\xi)$ on $J$ and, as such, $V'$ is positive on this interval. It follows that $P \equiv 0$ is not possible on any neighbourhood contained in $J$, since \eqref{eq:P_ODE} would imply that $V' = 0$ there. As such, we can find a point in $J$, again labelled $\xi_0$, for which the inequality  $0 < P(\xi_0) < V'(\xi_0)  =: v_0$ holds. Since $V'$ is nondecreasing, we have that $V'(\xi) \ge v_0 > 0$ for all $\xi \ge \xi_0$. Let $I = [\xi_0,\xi_1)$ be the maximal interval on which $0 < P < V'$.
     On this interval we have $V'(\xi)/P(\xi) > 1$, so that \eqref{eq:P_ODE} gives $P' < 0$ a.e.\ on $I$ and hence $P(\xi) \le P(\xi_0)$ for all $\xi \in I$. Since $V'$ is nondecreasing, it follows that
\[P(\xi) \le P(\xi_0) < V'(\xi_0) \le V'(\xi), \qquad \text{for all } \xi \in I,\]
from which we deduce that
\[P'(\xi) \le \frac{1}{2\varphi}\bigg(1-\frac{v_0}{P(\xi_0)}\bigg) < 0 \qquad \text{a.e.\ on } I;\]
that is, $P$ decreases at a rate bounded away from zero.
     Consequently, the right end point $\xi_1$ of $I$ cannot correspond to $P(\xi_1) = V'(\xi_1) > 0$. Since $P$ decreases at a rate bounded away from zero, it hits zero in finite time so we must have that $\xi_1 < \infty$ and $P(\xi_1) = 0$.
     
     We claim $P \equiv 0$ on a right-neighbourhood of $\xi_1$. Indeed, for $\xi > \xi_1$ close to $\xi_1$ we have, by continuity of $P$ and right-continuity of $V'$, that $P < V'$. As such, wherever $0 < P < V'$, the ODE \eqref{eq:P_ODE} gives $P'(\xi) = \frac{1}{2\varphi}(1 - V'(\xi)/P(\xi)) \le 0$. Since $P \ge 0$ and $P < V'$ on this neighbourhood, this implies that $P' \le 0$ a.e.\ on it. Hence,
\[
0 \le P(\xi) = P(\xi_1) + \int_{\xi_1}^{\xi} P'(z)\d z = \int_{\xi_1}^{\xi} P'(z)\d z \le 0,
\]
and therefore $P \equiv 0$ on this right-neighbourhood. But then $P' = 0$ a.e.\ there, so \eqref{eq:P_ODE} forces $V' = 0$ a.e.\ on this right-neighbourhood of $\xi_1$ contradicting $V'(\xi) \ge v_0 > 0$ for $\xi \ge \xi_0$. As such, we must have $P \ge V'$.

    \paragraph{Step 5: The case $V$ is constant.} Next, note that if $V$ is constant so that $V' = 0$ everywhere then \eqref{eq:P_ODE} reduces to $2\varphi P'(\xi)P(\xi) = P(\xi)$. The only bounded solution is $P \equiv 0$. Indeed, on each connected component of $\{P \ne 0\}$ the equation forces $2\varphi P' = 1$, and since $P$ is increasing and vanishes at any finite endpoint of such a component, the only nonzero solutions are $P(\xi) = \frac{1}{2\varphi}(\xi - C)^+$ for some $C \in \R$, which are unbounded.

    \paragraph{Step 6: General $V$ case; existence.} We may assume that there exists $\xi_0 \in \R$ so that $V'(\xi_0) > 0$. Because $V'$ is nondecreasing this implies that $V'> 0$ on $[\xi_0,\infty)$. Clearly, from the ODE \eqref{eq:P_ODE} we have that $P(\xi_0) \ne 0$, so we can formulate the ODE in the more standard form
    \[P'(\xi) = g(\xi,P(\xi)), \qquad \text{where} \quad  g(\xi,p) = \frac{1}{2\varphi}\bigg(1 - \frac{V'(\xi)}{p}\bigg).\]
    We append this ODE with the initial condition $P(n) = c_+$ for any $n \in \N$.
    Note that on any set of the form $\{p \ge \epsilon\}$ for some $\epsilon > 0$ (and, in particular,  near the point $(n,c_+)$), $g(\xi,p)$ is measurable in $\xi$ for every fixed $p$, is continuous in $p$ for every fixed $\xi$ and we have the estimate
    $|\partial_p g(\xi,p)| \le \frac{L_V}{2\varphi\epsilon^2}$. As such, Carath\'eodory's Theorem \cite[Theorem~I.5.3]{hale2009ordinary} guarantees a unique local solution $P_n$ to \eqref{eq:P_ODE} satisfying $P_n(n) = c_+$. From the ODE \eqref{eq:P_ODE}, it is clear that $P_n$ is nondecreasing on a neighbourhood of $n$. As such, by continuous extension we can extend $P_n$ to $(\widetilde \xi,n]$, where $\widetilde\xi < n$ is the largest point where $P_n(\widetilde \xi) = 0$. Such a point $\widetilde \xi$ exists since $V' = 0$ on $(-\infty,0)$, where any positive solution satisfies $P_n'(\xi) = \frac{1}{2\varphi} > 0$, and hence must hit zero sufficiently far to the left. By setting $P_n(\xi) = 0$ for $\xi \le \widetilde \xi$ we have defined $P_n$ on $(-\infty,n]$.  It is easy to see that $P_n$ is nondecreasing, valued in $[0,c_+]$ and satisfies $P_n \ge V'$. In particular, $V' = 0$ on $(-\infty, \widetilde \xi]$, so that $P_n = 0$ solves \eqref{eq:P_ODE} on this interval as well, and hence is a  solution to \eqref{eq:P_ODE} on $(-\infty,n]$.
    Additionally, by nonnegativity of $V'$ we have on $\{P_n > 0\}$ that $P_n'(\xi) = \frac{1}{2\varphi}(1- V'(\xi)/P_n(\xi)) \le \frac{1}{2\varphi}$ so that $P_n$ is $\frac{1}{2\varphi}$-Lipschitz.
        
    Next, we will show that whenever $m > n$, we have that $P_m \le P_n$ on $(-\infty,n]$. Suppose by way of contradiction, that $P_m(\overline \xi) >P_n(\overline \xi)$ for some $\overline \xi \le n$. Let \[\xi_* = \sup\{\xi \in (-\infty,n]: P_m(\xi) > P_n(\xi)\},\]
    and note that by assumption $\xi_*$ is well-defined. By continuity, we have that $P_m(\xi_*) = P_n(\xi_*)$, and we call this common value $v$. If $v >0$, then $P_m$ and $P_n$ both solve \eqref{eq:P_ODE} with the same initial data $P(\xi_*) = v$. Since $v > 0$,  Carath\'eodory's Theorem applies and local uniqueness guarantees that $P_m = P_n$ on a neighbourhood of $\xi_*$, which would contradict the definition of $\xi_*$. Conversely, if $v = 0$ then since $P_m$ and $P_n$ are nonnegative and nondecreasing we have that $P_m(\xi) = P_n(\xi) = 0$ for $\xi \le \xi_*$, again contradicting the definition of $\xi_*$. As such, it follows that $P_m \le P_n$ on $(-\infty,n]$ for all $m \ge n$. 
    Now, fix $\xi \in \R$ and note that because $P_n(\xi) \ge P_m(\xi)$ for $m \ge n \ge \xi$, by monotonicity it follows that $P(\xi) := \lim_{n \to \infty} P_n(\xi)$ exists. 
    Because the $P_n$ are uniformly $\frac{1}{2\varphi}$-Lipschitz, the limit $P$ inherits this property together with the bounds $V' \le P \le c_+$. 
    Now on any interval $[\xi_1,\xi_2] \subset \{P > 0\}$ we have that 
    \[P_n(\xi_2) - P_n(\xi_1) = \int_{\xi_1}^{\xi_2} \frac{1}{2\varphi}\bigg(1 - \frac{V'(s)}{P_n(s)}\bigg)\d s,\]
    where we note that $P_n \ge P \ge \min_{\xi \in [\xi_1,\xi_2]}P(\xi) >  0$. Hence by dominated convergence we deduce that 
    \[P(\xi_2) - P(\xi_1) = \int_{\xi_1}^{\xi_2} \frac{1}{2\varphi}\bigg(1 - \frac{V'(s)}{P(s)}\bigg)\d s.\]
    We see that $P$ is absolutely continuous on this interval, and differentiating the above identity shows it satisfies \eqref{eq:P_ODE} a.e.\ Additionally, on the interval $\{P = 0\}$ we have $V' = 0$, due to the bound $V' \le P_n$, from which we deduce that \eqref{eq:P_ODE} also holds on the interval's interior. Since $P \le P_n \le c_+$ is bounded we have constructed a bounded, absolutely continuous (even Lipschitz) solution to \eqref{eq:P_ODE}.

\paragraph{Step 7: General $V$ case; uniqueness.} Suppose $P$ and $Q$ are two bounded and absolutely continuous solutions. If $P \equiv  Q \equiv 0$ there is nothing to prove. Moreover, the zero function is a solution if and only if $V$ is constant so we can assume without loss of generality that neither $P$ nor $Q$ is identically zero. Since $c_+ > 0$ there exists a smallest $\xi_1$ such that $P$ and $Q$ are both nondecreasing and strictly positive on $(\xi_1,\infty)$; in particular, by continuity $P(\xi_1) = 0$ or $Q(\xi_1) = 0$. Defining $D = P-Q$ it is easy to see that $D$ satisfies
\[D' = \frac{1}{2\varphi}\bigg(\frac{V'}{Q} - \frac{V'}{P}\bigg) = \frac{V'}{2\varphi PQ} D \qquad \text{on }(\xi_1,\infty).\] Note that
\[\lim_{\xi \to \infty} D(\xi) = \lim_{\xi \to \infty} P(\xi) - \lim_{\xi \to \infty} Q(\xi) = c_+ - c_+ = 0. \]
Moreover, away from the set $\{D = 0\}$, the function $|D|$ is differentiable and satisfies
\[\frac{\d}{\d \xi} |D(\xi)| = \frac{V'(\xi)}{2\varphi  P(\xi)Q(\xi)}|D(\xi)| \ge 0\] so that $|D|$ is nondecreasing. But since $|D(\xi)| \to 0$ as $\xi \to \infty$ this shows that actually $D = 0$ on $[\xi_1,\infty)$. For values $\xi < \xi_1$ we have that either $P(\xi) = 0$ or $Q(\xi) = 0$. Suppose, without loss of generality that $P(\xi) = 0$. Then sending $\xi \uparrow \xi_1$ and invoking continuity we see that both $P(\xi_1) = 0$ and $D(\xi_1) = 0$, which implies that $P(\xi_1) = Q(\xi_1) = 0$. Hence $P(\xi) = Q(\xi) = 0$ for $\xi \le \xi_1$ establishing that $P = Q$ everywhere; that is, uniqueness holds.
   \paragraph{Step 8: Monotonicity of $\Xi$.} We now return to the general $\varepsilon$ case.
 On $\{P = 0\}$, we have that $\Xi(\xi) = \xi$ is clearly strictly increasing. Otherwise, on $\{P \ne 0\}$ we use \eqref{eq:P_ODE} to compute that 
    \[\Xi'(\xi) = 1 - \lambda P'(\xi) = 1- \frac{\lambda}{2\varphi}\bigg(1-\frac{V'(\xi)}{P(\xi)}\bigg) \ge \frac{2\varphi - \lambda}{2\varphi},\]
    where we used that $V'/P \ge 0$, which is immediate from the previously established bound $\varepsilon P \ge 0$. 
    Hence, when $2\varphi > \lambda$, we see that $\Xi$ is strictly increasing, while for $2\varphi = \lambda$ it is nondecreasing. In the latter case, on $\{P \ne 0\}$ we see that $\Xi'(\xi) = 0 \iff V'(\xi) = 0$. As such, $V$ is flat on these regions, so that $V(\Xi^{-1}(x))$ is well-defined for any $x \in \R$. Moreover, there are at most countably many $x \in \R$ such that $\Xi^{-1}(x)$ is not a singleton. Since $P$ is bounded, it is clear that $\lim_{\xi \to \pm \infty} \Xi(\xi) = \pm \infty$ so that it is surjective.

    \paragraph{Step 9: $\widetilde V$ is $C^1$.}
    Note, by the chain rule and derivative of inverse function formula we have for a.e.\ $x$ such that $\Xi^{-1}(x)$ is a singleton,
    \begin{align}
        \widetilde V'(x) & = \frac{V'(\Xi^{-1}(x))}{\Xi'(\Xi^{-1}(x))} + \frac{2\varphi- \lambda}{2}\frac{(P^2)'(\Xi^{-1}(x))}{\Xi'(\Xi^{-1}(x))} = \frac{V'(\Xi^{-1}(x)) + (2\varphi-\lambda)P'(\Xi^{-1}(x))P(\Xi^{-1}(x))}{1-\lambda P'(\Xi^{-1}(x))} \nonumber \\
        & = \frac{V'(\Xi^{-1}(x)) + P(\Xi^{-1}(x)) - V'(\Xi^{-1}(x)) -\lambda P'(\Xi^{-1}(x))P(\Xi^{-1}(x))}{1-\lambda P'(\Xi^{-1}(x))} \nonumber \\
        & = P\big(\Xi^{-1}(x)\big), \label{eq:tildeV'}
    \end{align}
    where we used \eqref{eq:P_ODE} in the penultimate equality. If $2\varphi > \lambda$ then $\Xi^{-1}(x)$ is a singleton for all $x$, and since $V$ and $P$ are absolutely continuous, and $\Xi^{-1}$ is Lipschitz, we see from \eqref{eq:tilde_V} and \eqref{eq:tildeV'} that $\widetilde V$ is $C^1$ as required. If $2\varphi = \lambda$ then from \eqref{eq:P_ODE} we see that $V'(\xi) = P(\xi)\Xi'(\xi)$ a.e.\ Since $V$ and $\Xi$ are absolutely continuous, $\Xi$ is nondecreasing and $|P|$ is bounded by $L_V$, it follows that
    \[|V(\xi_1) - V(\xi_2)| \le L_V (\Xi(\xi_1) - \Xi(\xi_2)), \qquad \forall \xi_1 \ge \xi_2.\]
    Taking $x_i\in \R$  and $\xi_i \in \Xi^{-1}(x_i)$ for $i=1,2$ with $x_1 \ge x_2$ we obtain
    \[|\widetilde V(x_1) - \widetilde V(x_2)| = |V(\xi_1) - V(\xi_2)| \le L_V(\Xi(\xi_1) - \Xi(\xi_2)) =  L_V(x_1-x_2), \]
    establishing Lipschitz continuity of $\widetilde V$. Finally, it remains to establish that if $\widetilde V$ is differentiable at $x$ then $\Xi^{-1}(x)$ is a singleton.  Suppose $\Xi^{-1}(x)$ is an interval $[a,b]$ and note by definition of $\Xi$ that we have 
    \[P(\xi) = \frac{\xi - \Xi(\xi)}{\lambda} = \frac{\xi - x}{\lambda}, \qquad \text{for } \xi \in [a,b],\]
    so that $P(a) < P(b)$. However, from calculations similar to those in \eqref{eq:tildeV'} we see that $\widetilde V'_-(x) = P(a) < P(b) = \widetilde V'_+(x)$, which contradicts differentiability of $\widetilde V$. 
    
    \paragraph{Step 10: The remaining claims.} From the identity \eqref{eq:tildeV'}, and the bound $0 \le \varepsilon P \le L_V$, we see that $\varepsilon \widetilde V$ is nondecreasing and Lipschitz with Lipschitz constant at most $L_V$. 
    
    Now we consider the case $\varepsilon = -1$. Since $P$ is zero on the negative half-line, we have $\Xi(\xi) = \xi$, whereas for $\xi > 0$ we have $\Xi(\xi) = \xi - \lambda P(\xi) \ge \xi > 0$, so that $\Xi^{-1}(x) = \{x\}$ for $x \le 0$. As such, we have from \eqref{eq:tilde_V} that
    \[\widetilde V(x) = V(x) + \tfrac{2\varphi-\lambda}{2}P(x)^2 = V(0) = \widetilde V(0), \qquad x \le 0, \]
    establishing constancy of $\widetilde V$ on the negative half-line. Finally for a.e.\ $x > 0$ and $\xi \in \Xi^{-1}(x)$, we have $\xi  >0$ by the previous discussion, from which we deduce the bound
    \[ x  = \xi - \lambda P(\xi) \ge -\lambda P(\xi) = - \lambda \widetilde V'(x),\]
    where we again used \eqref{eq:tildeV'}. This completes the proof. 
\end{proof}

We now prove Theorem~\ref{thm:terminal} with the function $\widetilde V$ being the one constructed in Lemma~\ref{lem:ODE}.
\begin{proof}\!\![Proof of Theorem~\ref{thm:terminal}] 

The case that $V$ is constant is trivial so we focus on the nonconstant situation, in which case $L_V > 0$.
    First, let $x$ be a point of differentiability of $\widetilde V$ and set $\xi = \Xi^{-1}(x)$, which is a singleton by Lemma~\ref{lem:ODE}. Then, again by Lemma~\ref{lem:ODE}, $\widetilde V'(x) = P(\xi)$ and from the definition of $\Xi$ we see that 
    \begin{equation} \label{eq:xi} 
    x = \xi - \lambda P(\xi)  \implies \xi = x+ \lambda\widetilde V'(x).
    \end{equation} As such,
    \begin{equation} \label{eq:tildeV_solves_ODE}
    \widetilde V(x) = V(\xi) + \tfrac{2\varphi-\lambda}{2}P^2(\xi) = V\big(x+\lambda \widetilde V'(x)\big) + \tfrac{2\varphi-\lambda}{2}\big(\widetilde V'(x)\big)^2,
    \end{equation}
    which shows that $\widetilde V$ solves \eqref{eq:implicit_ODE}.

    From here we can deduce that $\widetilde V$  dominates $V$. Indeed, since $\widetilde V'(x) = P(\xi)$ for a.e.\ $x$ and $\varepsilon P \ge 0$, it follows that $\varepsilon \widetilde V' \ge 0$. Hence, by monotonicity of $V$ we have that $V(x + \lambda \widetilde V'(x)) \ge V(x)$, which together with \eqref{eq:tildeV_solves_ODE} and continuity of both functions yields $\widetilde V(x) \ge V(x)$ for all $x \in \R$.

    Next, we let a Carath\'eodory solution $f$ to \eqref{eq:implicit_ODE} be given.
    To start, we will establish that in the $2\varphi - \lambda > 0$ case, $f$ satisfies
    \begin{equation} \label{eq:f_lower_bound}
    f(x) \ge - L_V |x| - C, \qquad \forall x \in \R,
    \end{equation}
    where $L_V$ is the Lipschitz constant of $V$ and $C \ge 0$ is some constant. To see this, note that $V(x) \ge V(0) - L_V|x|$ so that 
    \[f(x) = V(x+\lambda f'(x)) + \tfrac{2\varphi-\lambda}{2} (f'(x))^2 \ge V(0) - L_V|x| - \lambda L_V |f'(x)| + \tfrac{2\varphi-\lambda}{2} (f'(x))^2. \]
    The quadratic term dominates and we see that the lower bound holds with $C = |V(0)| + \frac{L_V^2\lambda^2}{2(2\varphi-\lambda)}$.
    This establishes \eqref{eq:f_lower_bound} a.e., and hence everywhere by continuity.
    
    Now set $w = \widetilde V - f$. In the case $V$ is nonincreasing we have $w = 0$ on $(-\infty,0]$, since $\widetilde V(x) = f(x) = V(0)$ here, so it suffices to consider $x > 0$. Moreover, in this case, Lemma~\ref{lem:ODE} and the hypotheses on $f$ imply that $\widetilde V'(x),f'(x) \ge -x/\lambda$ for a.e.\ $x > 0$. Returning to the general case, we let $x$ be a common point of differentiability for $\widetilde V$ and $f$.   Set \[\Phi_x(q) = V(x+\lambda q) + \tfrac{2\varphi - \lambda}{2}q^2, \qquad q \in D_x,\] where $D_x = \R$ if $V$ is nondecreasing and $D_x =  [-\frac{x}{\lambda},\infty)$, when $V$ is nonincreasing.
    Then, $\Phi_x$ is convex and has subdifferential $\partial \Phi_x(q) = \lambda \partial V(x+\lambda q) + (2\varphi-\lambda)q$.  Letting $\xi$ be given by the final expression in \eqref{eq:xi}, and noting that $V'(\xi) \in \partial V(\xi)$ (where we recall the notation $V'$ denotes the right-derivative of $V$), we see that 
    \[s:= \lambda V'(\xi) + (2\varphi-\lambda) \widetilde V'(x) \in \partial \Phi_x(\widetilde V'(x)). \]
    Since $\varepsilon V'(\xi) \in [0,L_V]$ and,  by Lemma~\ref{lem:ODE}, $\varepsilon \widetilde V'(x) = \varepsilon P(\xi) \in [0,L_V]$, we see that $\varepsilon s \in [0,2\varphi L_V]$. Using that $f$ solves \eqref{eq:implicit_ODE}, that $f$ is differentiable at $x$, and $\Phi_x$ is convex, we have that
    \[f(x) = \Phi_x(f'(x)) \ge \Phi_x(\widetilde V'(x)) + s(f'(x) - \widetilde V'(x)) = \widetilde V(x) + s(f'(x) - \widetilde V'(x)).\] Bringing $f$ to the other side establishes the bound
    \begin{equation} \label{eq:w_bound} 
    w(x) \le s w'(x).
    \end{equation}

    We now assume by way of contradiction that $w(x_0) > 0$ for some $x_0 \in \R$. First we consider the $\varepsilon = 1$ case. Then, by continuity, $w > 0$ on some maximal neighbourhood $(a,b)$ of $x_0$ so the bound \eqref{eq:w_bound} implies that $s  > 0$, so that $w'(x) \ge w(x)/s \ge w(x)/(2\varphi L_V) >0$ a.e.\ on this neighbourhood. This implies that the function $x \mapsto w(x)\exp(-\frac{x-x_0}{2\varphi L_V})$ is increasing on $[x_0,b)$ so that
    \begin{equation} \label{eq:w_lower_bound} 
    w(x) \ge w(x_0)\exp\bigg(\frac{x-x_0}{2\varphi L_V}\bigg), \qquad x_0 \le x < b.
    \end{equation}
    We see that we must have $b = \infty$; indeed, otherwise by continuity we would have $w(b) = 0$, which would contradict the bound \eqref{eq:w_lower_bound}. On the other hand, boundedness of $\widetilde V'$ guarantees that $\widetilde V$ has at most affine growth, so together with \eqref{eq:w_lower_bound} this implies that $f(x) \to -\infty$ as $x \to \infty$.  When $2\varphi - \lambda > 0$ this contradicts \eqref{eq:f_lower_bound}, while if $2\varphi = \lambda$ then this contradicts $f(x) = V(x + \lambda f'(x)) \ge V(0)$, which follows by nondecreasingness of $V$. 
    
    The case  $\varepsilon = -1$ is similar, with the preceding argument leading to an exponential lower bound 
    $w(x) \ge w(x_0) \exp(\frac{x_0-x}{2\varphi L_V})$ for $ a <  x \le x_0$. Sending $x \downarrow a \ge 0$ contradicts $w(a) = 0$. As such, we must have $w = \widetilde V - f \le 0$.
    
    It just remains to prove that $|\widetilde V'(x)|\le |f'(x)|$ for any common point of differentiability $x$. If $\widetilde V'(x) = 0$ there is nothing to prove so we assume by way of contradiction that $0 \le |f'(x)| < |\widetilde V'(x)|$. If $2\varphi - \lambda > 0$ then using \eqref{eq:implicit_ODE} and monotonicity of $V$ we have that
    \begin{align*} 
    0 \le f(x) - \widetilde V(x) & = V(x+\lambda f'(x)) - V(x + \lambda \widetilde V'(x)) + \tfrac{2\varphi-\lambda}{2}\Big(\big(f'(x)\big)^2 -  \big(\widetilde V'(x)\big)^2\Big)   \\
    & \le \frac{2\varphi-\lambda}{2}\Big(\big(f'(x)\big)^2 -  \big(\widetilde V'(x)\big)^2\Big),
    \end{align*}
    which is a contradiction. 
    
    Next we consider the case $2\varphi = \lambda$.  Using the bound $f(x) \ge \widetilde V(x)$ together with \eqref{eq:tildeV_solves_ODE} gives $V(x+\lambda \widetilde V'(x)) \le V(x+\lambda f'(x))$. However, since $ \varepsilon f'(x) \le |f'(x)| < |\widetilde V'(x)| = \varepsilon \widetilde V'(x)$,  we have that $V$ is constant on $[x + \lambda f'(x), \xi]$ (resp., $[\xi,x+\lambda f'(x)]$) if $\varepsilon = 1$ (resp., $\varepsilon = -1$), where $\xi = x + \lambda \widetilde V'(x) = \Xi^{-1}(x)$. Note that $\widetilde V'(x) \ne 0$, so that $P(\xi) \ne 0$ and hence, in the case $\varepsilon = -1$ we have $\xi > 0$ by Lemma~\ref{lem:ODE}. As $\varepsilon f'(x) < \varepsilon \widetilde V'(x)$, this interval is nondegenerate, so $V' = 0$ on a one-sided neighbourhood of $\xi$. Since $P(\xi) = \widetilde V'(x) \ne 0$ and $P$ is continuous, after possibly shrinking the neighbourhood we can assume that $P \ne 0$ on it. The ODE \eqref{eq:P_ODE} reads $2\varphi P' P = P$, from which we deduce that $P' = \frac{1}{2\varphi}$ at its points of differentiability on this interval. As $2\varphi = \lambda$ this gives $\Xi' = 1 - \lambda P' = 0$, so $\Xi$ is constant on a one-sided neighbourhood of $\xi$. But then $\Xi^{-1}(x)$ is not a singleton, contradicting Lemma~\ref{lem:ODE}, which guarantees that $\Xi^{-1}(x)$ is a single point at every point of differentiability of $\widetilde V$. This establishes $|f'(x)| \ge |\widetilde V'(x)|$ and completes the proof. 
    \end{proof}

Next we verify that \eqref{eq:tildeVcall} is indeed the function $\widetilde V$ of Theorem~\ref{thm:terminal} in the case of a call option. To this end, it is easy to see that the function $P$ of Lemma~\ref{lem:ODE} is 
\[P(\xi) = \begin{cases}
0, & \xi < K- 2\varphi N, \\
N + \frac{\xi-K}{2\varphi}, & K - 2 \varphi N  \le \xi < K, \\
N, & \xi \ge K.
\end{cases}
\]
Indeed, this follows from the uniqueness statement of the lemma since this function is Lipschitz continuous, bounded and satisfies \eqref{eq:P_ODE}. A direct computation shows that $\Xi$ is strictly increasing everywhere for $2\varphi - \lambda > 0$ and strictly increasing outside of the region $(K-2\varphi N, K)$ on which it is constant for $2\varphi = \lambda$. Hence,
\[\Xi^{-1}(x) = \begin{cases}
    x, & x \le  K - 2\varphi N, \\
    \frac{2\varphi}{2\varphi-\lambda} x - \frac{\lambda}{2\varphi-\lambda}(K-2\varphi N), & K - 2\varphi N < x < K - \lambda N, \\
    x+\lambda N, & x \ge K - \lambda N.
\end{cases}\]
Note that the middle region is empty when $2\varphi = \lambda$. Finally substituting into \eqref{eq:tilde_V} and simplifying leads to
\eqref{eq:tildeVcall}.

\subsection{Proof of Theorem~\ref{thm:PDE}}

    	It will be convenient to define $\Sigma(t,x) = x\sigma(t,x)1_{\{x > 0\}}$ for $(t,x) \in [0,T] \times \R$ and $\widetilde \sigma(t,y) = \sigma(t,e^y)$ for all $(t,y) \in [0,T] \times \R$. We use the convention throughout that $V(x) = V(0)$ for $x < 0$. Assumption~\ref{ass:sigma} guarantees that $\Sigma(t,\cdot)$ is globally Lipschitz continuous uniformly in $t$ with Lipschitz  constant $L_\Sigma = \overline \sigma + C_\sigma$, since $\partial_x \Sigma(t,x) = \sigma(t,x) + \partial_y \widetilde \sigma(t,\log x)$.

	\begin{proof}  For $r = 0$ the PDE \eqref{eq:PDE} is linear and the solution explicit with $u = \widetilde u^{\mathrm{BS}}$, so the claims are verified with classical arguments. As such, we assume $r >0$ going forward. The proof first establishes uniqueness in step 1 and then existence in step 2. As part of the existence proof we also establish the stochastic representation \eqref{eq:u_stoch_rep} and the bounds \eqref{eq:u_bounds}. Throughout, we set $M = r\lambda L_V$.
		
		\paragraph{\underline{Step 1: Uniqueness.}}
		
		\paragraph{Step 1a: Comparison principle.} 
		For a feedback-form control characterized by a measurable function $a:[0,T] \times (0,\infty) \to \R$, consider the operator
		\begin{equation} \label{eq:Lcal_operator} \mathcal{L}^a w = - rw + \partial_t w + (rx + a)\partial_x w + \tfrac{1}{2}\sigma^2 x^2 \partial_{xx}w. 
		\end{equation}
		We assume $\varepsilon a(t,x) \in [0,M]$ and that $a^{-}(t,x) \le \Lambda' x$ for some $\Lambda' > 0$ and all $(t,x) \in [0,T] \times (0,\infty)$. We now establish the following comparison principle, which will be crucial in the proof of uniqueness: if $w \in C^{1,2}([0,T) \times (0,\infty)) \cap C([0,T] \times (0,\infty))$ has linear growth $|w(t,x)| \le C(1+x)$, and satisfies
		\begin{equation} \label{eq:w_comparison}
			\mathcal{L}^a w  \ge 0 \quad \text{on } [0,T) \times (0,\infty) \quad \text{and} \quad w(T,\cdot) \le 0, 
		\end{equation} 
		then $w \le 0$. In particular, by applying this to $w$ and $-w$ we see that if $\mathcal{L}^a w = 0$ and $w(T,\cdot) = 0$ then $w= 0$.
		
		To establish this, we consider the explicit function $\psi(t,x) = \eta e^{K(T-t)}(x^p + x^{-1})$ for some $\eta,K >0$ and $p > 1$. A direct computation shows that 
		\[\frac{\mathcal{L}^a\psi}{\eta e^{K(T-t)}} =  
		\big(-K + r(p-1) + \tfrac{1}{2}\sigma^2 p(p-1)\big)x^p + (\sigma^2 - K - 2r)x^{-1} + a(px^{p-1}-x^{-2}).\]
		We now write $a = a^+  - a^-$ and note that the four terms this produces can be bounded above as follows:
		\[a^+ px^{p-1} \le Mp x^{p-1}, \quad -a^- px^{p-1} \le 0, \quad -a^+x^{-2} \le 0, \quad a^- x^{-2} \le \Lambda' x^{-1}.\]
		This leads to the estimate \[\frac{\mathcal{L}^a \psi}{\eta e^{K(T-t)}} \le c_1 x^p + Mp x^{p-1}  + c_2 x^{-1},\]
		where $c_1 =  -K + r(p-1) + \tfrac{1}{2}\overline \sigma^2 p(p-1)$ and $c_2  =  \overline \sigma^2 - K - 2r + \Lambda'$. For $K > Mp + \Lambda' + r(p-1) + \frac{1}{2}\overline \sigma^2 p(p-1) + \overline \sigma^2$, we have $\mathcal{L}^a \psi \le 0$. Indeed, for $x \ge 1$, $c_1x^p + Mpx^{p-1}  \le x^{p-1}(c_1 + Mp) < 0$ and $c_2x^{-1} \le 0$, whereas for $x < 1$, we have $c_1x^p \le 0$ and $Mpx^{p-1} + c_2x^{-1} \le (Mp + c_2)x^{-1} < 0$. As such, for $w$ satisfying \eqref{eq:w_comparison} we have that $\mathcal{L}^a(w - \psi) \ge 0$.
		
		Consider the compact domain  $[0,T] \times [1/m,m]$  for any $m > 1$. Suppose by way of contradiction that $w-\psi$ attains a positive maximum. If it happens on the parabolic interior of the domain; i.e., at some point $(t_0,x_0) \in [0,T) \times  (1/m,m)$, then at that point we have $\partial_x (w-\psi) = 0$, $\partial_{xx}(w-\psi) \le 0$ and $\partial_t(w-\psi) \le 0$, which leads to $\mathcal{L}^a (w-\psi)(t_0,x_0) \le - r(w(t_0,x_0) - \psi(t_0,x_0))  < 0$, contradicting the previously established nonnegativity of $\mathcal{L}^a (w-\psi)$. The maximum also cannot be achieved at $t_0 = T$ because $w \le 0$ there, so that $w-\psi < 0$. Hence, it must occur at either $x = 1/m$ or $x = m$. For the spatial boundary points we have that $\psi(\cdot, m)  \ge \eta m^p$ and $\psi(\cdot,1/m) \ge \eta m$. Since $w$ satisfies $|w(t,x)| \le C(1+x)$, and consequently $|w(t,x)| \le 2C$ on $\{x \le 1\}$, we see that for $m$ large enough we have  $w-\psi < 0$ on the spatial boundary; a contradiction. This establishes that $w \le \psi$ on $[0,T] \times [1/m,m]$. Sending $m\to \infty$ and  then $\eta \downarrow 0$ we establish $w \le 0$ on $[0,T] \times (0,\infty)$.

		\paragraph{Step 1b: Uniqueness in the stated class.} We now establish that the class of functions
		\begin{equation}\label{eq:Ucal}
			\mathcal{U} = \left \{\begin{aligned} u \in C^{1,2}([0,T) \times (0,\infty)) \cap C([0,T] \times (0,\infty)): u \text{ solves } \eqref{eq:PDE}, \,
				u(T,\cdot) = \widetilde V,\\  0 \le \varepsilon\partial_x u \le L_V \text{ and } \text{there exists } \Lambda > 0 \text{ s.t.\ }   (\partial_x u)^{-} \le \Lambda x \end{aligned} \right\}
		\end{equation}
		has at most one element. To this end, suppose $u_1,u_2 \in \mathcal{U}$ are given and define $w = u_1-u_2$. Then set $a(t,x) = \frac{r\lambda}{2}(\partial_x u_1(t,x) + \partial_x u_2(t,x))$ and note that $\varepsilon a(t,x) \in [0,M]$ and $a^{-}(t,x) \le \frac{r\lambda}{2} (\Lambda_1 + \Lambda_2) x$ by the derivative bounds, where $\Lambda_i$ is the constant in the definition of \eqref{eq:Ucal} associated with $u_i$.  A direct calculation shows that $\mathcal{L}^a w = 0$ and $w(T,\cdot) = 0$. Since $a$ satisfies the hypothesis of the comparison bound of the previous paragraph, and $w$ has linear growth (because it is continuous and $|\partial_x w| \le 2L_V$), we obtain $u_1 = u_2$, which shows uniqueness.

\paragraph{\underline{Step 2: Existence.} }  We now work on establishing that the function $u$ of \eqref{eq:u_stoch_rep} solves \eqref{eq:PDE}, and satisfies the bounds \eqref{eq:dxu_bounds} and \eqref{eq:u_bounds}.

\paragraph{Step 2a: The bounds (\ref{eq:u_bounds}).} 
We start by establishing \eqref{eq:u_bounds}. The lower bounds $u(t,x) \ge \widetilde u^{\mathrm{BS}}(t,x) \ge u^{\mathrm{BS}}(t,x)$ are immediate by taking $\alpha = 0$ and recalling $\widetilde V \ge V$. For the upper bound, let $\alpha \in \mathcal{A}$ be given and note that because $\varepsilon \alpha \ge 0$, we have by comparison that $\varepsilon (S_T^\alpha - S_T) \ge 0$. It then follows from the dynamics of  $e^{-r(s-t)}(S_s^\alpha - S_s)$ that we have $e^{-r(T-t)}\E_{t,x}[S^\alpha_T- S_T] = \E_{t,x}[\int_t^T e^{-r(s-t)}\alpha_s \d s]$. As such, the Lipschitz property of $\widetilde V$ yields 
\[\widetilde V(S_T^\alpha) - \widetilde V(S_T) \le \varepsilon L_V(S_T^\alpha - S_T). \] It then follows that
			\begin{align*}
				J(t,x,\alpha) - \widetilde u^{\mathrm{BS}}(t,x) \le \E_{t,x}\bigg[\int_t^T e^{-r(s-t)}\bigg(\varepsilon L_V \alpha_s - \frac{\alpha_s^2}{2r\lambda}\bigg)\d s\bigg] & \le \frac{r\lambda L_V^2}{2}\int_t^T e^{-r(s-t)}\d s \\
				& = \frac{\lambda L_V^2}{2}(1-e^{-r(T-t)}), 
			\end{align*}
			where we maximized the quadratic function $z \mapsto \varepsilon L_V z - \frac{z^2}{2r\lambda}$ in the penultimate step. Now taking supremum over $\alpha \in \mathcal{A}$ yields the upper bound in \eqref{eq:u_bounds}. In particular $u$ is finite and has linear growth in $x$, since $\widetilde u^{\mathrm{BS}}$ does as a consequence of $\widetilde V$ having linear growth and the identity $\E_{t,x}[S_T] = xe^{r(T-t)}$.

\paragraph{Step 2b: A constrained control problem.}  Define
    \[ \Lambda_u = \tfrac{1}{\lambda} e^{\frac{1}{2}(2r + \overline \sigma^2 + L^2_\Sigma)T}, \qquad \rho = r\lambda \Lambda_u, \qquad A(x) = \{a \in \R: \varepsilon a \in [0,M] \text{ and } a \ge -\rho x\},\]
set $\mathcal{A}_\rho(t,x)$ to be all controls $\alpha \in \mathcal{A}$ with $\alpha_s \in A(S^\alpha_s)$ for every $s$, and define the restricted value function
\[u_\rho(t,x) = \sup_{\alpha \in \mathcal{A}_\rho(t,x)} J(t,x,\alpha) \le u(t,x).\]
For $\varepsilon = 1$, it is clear that $\mathcal{A}_\rho(t,x) = \mathcal{A}$, while for $\varepsilon = -1$, there is a genuine restriction.
The reason for introducing $\mathcal{A}_\rho(t,x)$ is that for any control $\alpha$ in that set, we have that $rS^\alpha + \alpha \ge (r-\rho)S^\alpha$ so that by comparison, $S^\alpha$ is a strictly positive process. The penultimate step in the proof will establish that actually $u_\rho = u$, so that this constraint is not restrictive.

\paragraph{Step 2c: Showing $u_\rho$ is Lipschitz and monotone.}
Next, we establish that $u_\rho$ is monotone and uniformly Lipschitz in $x$. For clarity we will write $S^{t,x,\alpha}$ to emphasize the initial conditions. Now fix $x' > x> 0$, $\alpha \in \mathcal{A}_\rho(t,x)$ and define $D = S^{t,x',\alpha} - S^{t,x,\alpha}$. We then see that $D$ has dynamics
\[\d D_s = rD_s \d s + \gamma_s D_s \d W_s, \qquad D_t = x'-x > 0,\]
where $\gamma_s =( \Sigma(s,S^{t,x',\alpha}_s)  - \Sigma(s,S^{t,x,\alpha}_s))/D_s$. By Assumption~\ref{ass:sigma} we have that  $|\gamma_s| \le L_\Sigma$. As such, we get the representation
\[D_s = (x'-x)e^{r(s-t)}\mathcal{E}_s, \quad \text{where} \quad \mathcal{E}_s = \exp\bigg( \int_t^s \gamma_u \d W_u - \tfrac 12 \int_t^s \gamma_u^2 \d u\bigg)\]
is a martingale, guaranteed by boundedness of $\gamma$. As such, $D$ is positive and because $\alpha \in \mathcal{A}_\rho(t,x)$ we see that $\alpha_s \ge -\rho S_s^{t,x,\alpha} \ge -\rho S_s^{t,x',\alpha}$ which shows that $\mathcal{A}_\rho(t,x) \subset \mathcal{A}_\rho(t,x')$.

   Next, we compute that 
\begin{equation} \label{eq:J_diff} 
	\varepsilon\big(J(t,x',\alpha) - J(t,x,\alpha)\big) = \varepsilon e^{-r(T-t)}\E[\widetilde V(S_T^{t,x',\alpha}) - \widetilde V(S_T^{t,x,\alpha})].
\end{equation}  By the Lipschitz property we obtain
\[\varepsilon \big(J(t,x',\alpha) - J(t,x,\alpha)\big) \le  e^{-r(T-t)}  L_V \E[D_T]  = L_V(x'-x),\]
where in the final step we used that $\E[D_T] = (x'-x)e^{r(T-t)}$, which follows by martingality of $\mathcal{E}$.  We now take supremum over $\alpha \in \mathcal{A}_\rho(t,x)$ to obtain
\[\varepsilon \big(u_\rho(t,x') - u_\rho(t,x)\big) \le L_V(x'-x),\]
which is immediate for $\varepsilon = 1$ because $\mathcal{A}_\rho(t,x) = \mathcal{A}_\rho(t,x')$, while for $\varepsilon = -1$ we use the inclusion $\mathcal{A}_\rho(t,x) \subset \mathcal{A}_\rho(t,x')$ to conclude.  
Next, from \eqref{eq:J_diff} when $\varepsilon  = 1$ we obtain $\varepsilon (u_\rho(t,x') - u_\rho(t,x))  \ge 0$ by nondecreasingness of $\widetilde V$ and the fact that $D_T > 0$. For $\varepsilon= -1$, we use a comparison argument. Given $\alpha' \in \mathcal{A}_\rho(t,x')$ set $\beta_s = \alpha_s'/S_s^{t,x',\alpha'} \in [-\rho,0]$ and let  $S^{t,x,\alpha}$ be the solution of the SDE $dZ_s = (r+ \beta_s)Z_sds + \Sigma(s,Z_s)dW_s$ with $Z_t = x$ and $\alpha_s = \beta_s S_s^{t,x,\alpha}$. Note that $S^{t,x',\alpha'}$ also solves this SDE, but with initial condition $x'$.  As such, by comparison we have that $S^{t,x',\alpha'} \ge S^{t,x,\alpha}$. From this observation and from the bounds on $\beta_s$ it follows that $\alpha \in \mathcal{A}_\rho(t,x)$. Since $\widetilde  V$ is nonincreasing  and $\alpha^2 \le (\alpha')^2$ we have that $J(t,x,\alpha) \ge J(t,x',\alpha')$, which shows that $u_\rho(t,x) \ge u_\rho(t,x')$ in this case. In summary, we have shown that
\begin{equation} \label{eq:u_rho_bounds}
	0 \le \varepsilon \big(u_\rho(t,x') - u_\rho(t,x)\big) \le L_V(x'-x).
\end{equation}
By \eqref{eq:u_rho_bounds} and \eqref{eq:u_bounds}, $u_\rho(t,\cdot)$ is monotone and bounded near zero, so we can extend it via $u_\rho(t,0) = \lim_{x \downarrow 0} u_\rho(t,x)$.

\paragraph{Step 2d: Establishing (\ref{eq:dxu_bounds}) for $u_\rho$.}
At every point of differentiability of $u_\rho(t,\cdot)$, dividing the estimate \eqref{eq:u_rho_bounds} by $x'-x$ and sending $x' \downarrow x$ shows  $0 \le \varepsilon \partial_x u_\rho(t,x) \le L_V$.  The second estimate in \eqref{eq:dxu_bounds} is automatic for $\varepsilon = 1$ so we focus on $\varepsilon = -1$. Here, for $x' > x > 0$ and for a control $\alpha \in \mathcal{A}_\rho(t,x)$ we see that  
\begin{equation} \label{eq:V_diff}
	\widetilde V(S_T^{t,x',\alpha}) - \widetilde V(S_T^{t,x,\alpha})  = \int_{S_T^{t,x,\alpha}}^{S_T^{t,x',\alpha}}\widetilde V'(z)\d z \ge - \tfrac{1}{\lambda}(S_T^{t,x',\alpha})^+D_T.
\end{equation}
Using that $\alpha \le 0$, SDE comparison and standard $L^2$ estimates yield
\[\|(S_T^{t,x',\alpha})^+\|_2 \le  
\|S_T^{t,x',0}\|_2  \le x'\exp\big(\tfrac{1}{2}(2r + \overline \sigma^2)(T-t)\big), \qquad \|D_T\|_2 \le (x'-x)\exp\big(\tfrac{1}{2}(2r + L_\Sigma^2)(T-t)\big).
\]
Hence, from \eqref{eq:J_diff}, \eqref{eq:V_diff} and Cauchy--Schwarz we obtain
\[J(t,x,\alpha) \le J(t,x',\alpha) + x'(x'-x) \Lambda_u.\]
Taking supremum over $\alpha \in \mathcal{A}_\rho(t,x)$ and using the inclusion $\mathcal{A}_\rho(t,x) \subset \mathcal{A}_\rho(t,x')$ again yields $u_\rho(t,x) \le u_\rho(t,x') + \Lambda_u x'(x'-x)$. Dividing by $x'-x$ and sending $x' \downarrow x$ then establishes the second inequality in \eqref{eq:dxu_bounds} at every point of differentiability of $u_\rho(t,\cdot)$.

 \paragraph{Step 2e: Policy iteration.} We start by defining the Hamiltonian $H:(0,\infty) \times \R \to \R$ via
 \begin{equation} \label{eq:Hamiltonian}
 	H(x,p) = \sup_{a \in A(x)}\{ap - \tfrac{a^2}{2r\lambda}\}.
 \end{equation}
 It is easy to see that $\partial_p H(x,p)$ is the projection of $r\lambda p$ onto $A(x)$; in particular,  $H(x,p) = \frac{r\lambda}{2}p^2$ and $\partial_p H(x,p) = r\lambda p$ whenever $r\lambda p \in A(x)$. More generally, we have that 
 \begin{equation} \label{eq:projection} \partial_p H(x,p) = \Pi_{A(x)}(r\lambda p)=
 \begin{cases}
 	\ \ \big(r\lambda p^{+}\big)\wedge M, & \text{if }\varepsilon=+1, \\
 	-\big((r\lambda p^{-})\wedge M\wedge \rho x\big), & \text{if }\varepsilon=-1.
 \end{cases}
 \end{equation}
 We now iteratively define the feedback-form controls $a_n$ and functions $u_{n+1}$ via $a_0 = 0$,
 \begin{equation} \label{eq:u_recursion} 
 	\mathcal{L}^{a_n} u_{n+1} =  \frac{a_n^2}{2r\lambda}, \quad u_{n+1}(T,\cdot) = \widetilde V, \qquad a_{n+1}(t,x) = \Pi_{A(x)}\big(r\lambda \partial_x u_{n+1}(t,x)\big),
 \end{equation} 
 where $\mathcal{L}^{a_n}$ is given by \eqref{eq:Lcal_operator}.
 Formally, the function $u_{n+1}$ solves a linear PDE resembling \eqref{eq:PDE}, except that the nonlinear term $(\partial_x u)^2$ is replaced by the previously defined control $a_n^2$. 
 
 We inductively show that this is well-defined. When $a_0 = 0$, the PDE for $u_1$ is simply the Black--Scholes PDE with terminal condition $\widetilde V$, which is well known to have unique solution $u_1 = \widetilde u^{\mathrm{BS}} \in C^{1,2}([0,T) \times (0,\infty))$. Next, fix $n$ for the inductive step  and note that by construction $a_{n}(t,x) \in A(x)$.
We define 
\[u_{n+1}(t,x) = J(t,x,\alpha^n), \quad \text{where} \quad \alpha_s^n = a_n(s,S_s^{\alpha^n}). \]
By standard interior regularity, $\partial_x u_n$ is locally H\"older continuous, and hence so is $a_n$, since $(x,p) \mapsto \Pi_{A(x)}(r\lambda p)$ is Lipschitz  
by \eqref{eq:projection}.
Making the change of variables $y = \log x$ and localizing to compact sets, the equation for $u_{n+1}$ transforms into a uniformly parabolic Cauchy problem with H\"older continuous coefficients. Approximating the Lipschitz terminal condition $\widetilde V$ by smooth functions, the classical result \cite[Theorem~IV.5.1]{ladyzhenskaya1968linear}, together with the Feynman--Kac formula and standard interior estimates, yields
 a unique linear-growth solution $u_{n+1} \in C^{1,2}([0,T) \times (0,\infty))$ to the PDE \eqref{eq:u_recursion}. 
 
Additionally $u_{n+1}$ has linear growth and $u_{n+1}(t,x) \le u_\rho(t,x)$ for all $n$, since $\alpha^n \in \mathcal{A}_\rho(t,x)$. 
Next, note that 
\[\mathcal{L}^{a_n}u_n - \frac{a_n^2}{2r\lambda} = (a_n-a_{n-1})\partial_x u_n +  \mathcal{L}^{a_{n-1}}u_n - \frac{a_n^2}{2r\lambda}  = a_n \partial_x u_n - \frac{a_n^2}{2r\lambda}  - \bigg(a_{n-1}\partial_x u_n  - \frac{a_{n-1}^2}{2r\lambda}\bigg) \ge 0,
 \] 
 where the last inequality is due to the fact that $a_n$ given by \eqref{eq:projection} maximizes the Hamiltonian \eqref{eq:Hamiltonian}. Since $\frac{a_n^2}{2r\lambda} = \mathcal{L}^{a_n}u_{n+1}$ we see that $\mathcal{L}^{a_n}(u_n - u_{n+1}) \ge 0$. Since $u_n(T,\cdot) - u_{n+1}(T,\cdot) = 0$ we see from the comparison principle that $u_n \le u_{n+1}$. As such, $u_\infty = \lim_{n \to \infty} u_n$ exists, has linear growth and satisfies $\widetilde u^{\mathrm{BS}} \le u_\infty \le u_\rho$.
 
 To establish that $u_\infty$ solves a limiting PDE, we need to obtain uniform in $n$ estimates on $a_n$, $u_n$ and its derivatives.  Since $\|a_n\|_\infty \le M$ and the bounds $\widetilde u^{\mathrm{BS}} \le u_n \le u_\rho$ hold uniformly in $n$, standard interior estimates for linear parabolic equations (see \cite[Ch~IV, Theorem~9.1 and Corollary~9.2]{ladyzhenskaya1968linear}) yield for each compact $K \subset [0,T) \times (0,\infty)$ and each $q > 3$, uniform in $n$ Sobolev bounds for $\|u_n\|_{W^{1,2}_q(K)}$ and uniform in $n$ H\"older bounds on $\|\partial_x u_n\|_{C^{\delta/2,\delta}(K)}$ for $\delta = 1 - 3/q\in (0,1)$ (here $W^{1,2}_q$ denotes the Sobolev space corresponding to one derivative in time, two in space and $L^q$ integrability of the function and its derivatives). As such, by taking the limit we see that $a_n \to a_\infty =  \Pi_{A(x)}(r\lambda \partial_x u_\infty)$ locally uniformly, and $u_\infty \in W^{1,2}_{q,\mathrm{loc}}([0,T) \times (0,\infty))$ solves the limiting version of \eqref{eq:u_recursion} given by $\mathcal{L}^{a_\infty}u_\infty = \frac{a^2_\infty}{2r\lambda}$. This is a linear parabolic equation with H\"older coefficients and right-hand side, so standard interior Schauder estimates guarantee that $u_\infty \in C^{1,2}([0,T) \times (0,\infty))$.
  Moreover, the bounds $\widetilde u^{\mathrm{BS}} \le u_\infty \le u_\rho \le u$ and the estimates \eqref{eq:u_bounds} imply that $u_\infty(t,x) \to \widetilde u^{\mathrm{BS}}(T,x) = \widetilde V(x)$ locally uniformly as $t \uparrow T$, so $u_\infty \in C([0,T] \times (0,\infty))$ with the correct terminal condition.

 \paragraph{Step 2f: Identification of the limit.}
 
 We now demonstrate that actually $u_\infty = u_\rho$. We have already shown that $u_\infty \le u_\rho$, so it just remains to show the reverse inequality. To this end, we fix $\alpha \in \mathcal{A}_\rho(t,x)$ and apply It\^o's formula to $e^{-r(s-t)}u_\infty(s,S_s^\alpha)$ leading to
 \begin{align*}
 	\d\big(e^{-r(s-t)}u_\infty(s,S_s^\alpha)\big) & =  e^{-r(s-t)}\big(-ru_\infty + \partial_t u_\infty + (rS_s^\alpha + \alpha_s)\partial_xu_\infty + \frac{1}{2}\sigma^2(S_s^\alpha)^2\partial_{xx}u_\infty\big)(s,S_s^\alpha)\d s  \\
 	& \qquad + e^{-r(s-t)}\partial_x u_\infty(s,S_s^\alpha)\Sigma(s,S_s^\alpha)\d W_s   \\
 	& =e^{-r(s-t)} \Big(\mathcal{L}^{a_\infty} u_\infty(s,S_s^\alpha) + \big(\alpha_s - a_\infty(s,S_s^\alpha)\big) \partial_x u_\infty(s,S_s^\alpha)\Big)\d s \\
 	& \qquad + e^{-r(s-t)}\partial_x u_\infty(s,S_s^\alpha)\Sigma(s,S_s^\alpha)\d W_s  \\
 	& = e^{-r(s-t)}\Big(\alpha_s\partial_x u_\infty(s,S_s^\alpha) - H\big(S_s^\alpha,\partial_x u_\infty(s,S_s^\alpha)\big)\Big)\d s \\
 	& \qquad+ e^{-r(s-t)}\partial_xu_\infty(s,S_s^\alpha)\Sigma(s,S_s^\alpha)\d W_s ,
 	\end{align*}  
 	where we used that $\mathcal{L}^{a_\infty}u_\infty = \frac{a_\infty^2}{2r\lambda}$ and that $a_\infty(s,x)$ achieves the maximum in the Hamiltonian at $p = \partial_x u_\infty(s,x)$. Since $\alpha_s\in A(S_s^\alpha)$, we can bound the drift from above by $e^{-r(s-t)}\frac{\alpha_s^2}{2r\lambda} $ using the definition of the Hamiltonian in \eqref{eq:Hamiltonian}. Next, defining the localizing sequence of stopping times  $\tau_m = \inf\{s \ge t: S_s^\alpha \not \in (1/m,m)\} \land T-\frac{1}{m}$ we  have
 	\[\E[e^{-r(\tau_m - t)}u_\infty( \tau_m,S_{\tau_m}^\alpha)] \le u_\infty(t,x) +\E\bigg[ \int_t^{\tau_m} e^{-r(s-t)}\frac{\alpha_s^2}{2r\lambda}\d s\bigg].\] 
 	Since $S^\alpha$ is nonexplosive, does not hit zero, has finite second moments, and $u_\infty$ has linear growth we can send $m \to \infty$ and use the terminal condition $u_\infty(T,\cdot) = \widetilde V$ to obtain $u_\infty(t,x) \ge J(t,x,\alpha)$. Taking supremum over $\alpha$ now yields $u_\infty(t,x) \ge u_\rho(t,x)$ and shows that $u_\infty = u_\rho$. Since $u_\rho$ satisfies \eqref{eq:dxu_bounds}, we have that $0 \le \varepsilon \partial_x u_\infty = \varepsilon \partial_x u_\rho \le L_V$ and $(\partial_x u_\infty)^{-} \le \Lambda_u x$, from which we deduce that $ r\lambda \partial_x u_\infty(t,x) \in A(x)$, so that $a_\infty(t,x) = r\lambda \partial_x u_\infty(t,x) $, and $H(x,\partial_x u_\infty) =\frac{r\lambda}{2}(\partial_x u_\infty)^2 $. As such, the PDE that $u_\infty$ solves is \eqref{eq:PDE}. 
 	
 	\paragraph{Step 2g: Removing the constraint.} 
 	The next step is to remove the constrained set  $\mathcal{A}_\rho$; that is, to show $u_\rho = u$. Clearly $u_\rho \le u$, so we again just need to show the reverse inequality. Taking $\alpha \in \mathcal{A}$, repeating the It\^o formula computations for $u_\infty = u_\rho$ of the previous paragraph, and using the expression $H(S_s^\alpha,\partial_x u_\rho) = \frac{r\lambda}{2}(\partial_x u_\rho(s,S_s^\alpha))^2$ leads to
 	\begin{align*}
 		d\big(e^{-r(s-t)}u_\rho(s,S_s^\alpha)\big)   & = e^{-r(s-t)}\Big(\alpha_s\partial_x u_\rho(s,S_s^\alpha) - \tfrac{r\lambda}{2}\big(\partial_x u_\rho(s,S_s^\alpha)\big)^2\Big)\d s \\
 		& \qquad + e^{-r(s-t)}\partial_xu_\rho(s,S_s^\alpha)\Sigma(s,S_s^\alpha)\d W_s \\
 		& \le e^{-r(s-t)}\tfrac{\alpha_s^2}{2r\lambda}\d s   + e^{-r(s-t)}\partial_xu_\rho(s,S_s^\alpha)\Sigma(s,S_s^\alpha)\d W_s,
 	\end{align*}
 	where we used the fact that $ap - \frac{r\lambda}{2}p^2 \le \frac{a^2}{2r\lambda}$ for every $a,p \in \R$.
 	For $\varepsilon = 1$, the state process $S^\alpha$ remains positive so the previous localization argument applies here verbatim and establishes $u \le u_\rho$. For $\varepsilon = -1$, we localize to the time $\tau_m$ of step 2f, take expectation and send $m \to \infty$ to obtain
 	\begin{equation} \label{eq:u_rho_lower_bound}
 		u_\rho(t,x) \ge \E\bigg[e^{-r(T\land \tau_0 - t)}u_\rho(T \land \tau_0,S^\alpha_{T \land \tau_0}) - \int_t^{T \land \tau_0}e^{-r(s-t)}\frac{\alpha_s^2}{2r\lambda}\d s\bigg], 
 		\end{equation}
    where $\tau_0 = \inf\{s \ge t: S_s^\alpha = 0\}$.
 	On the set $\{\tau_0 < T\}$ we see  that \[u_\rho(T\land \tau_0, S^\alpha_{T \land \tau_0}) = u_\rho(\tau_0,0+) \ge \widetilde u^{\mathrm{BS}}(\tau_0,0+) = e^{-r(T-\tau_0)}\widetilde V(0).\]
 	Moreover, the drift $rS_s^\alpha + \alpha_s$ is nonpositive once $S_s^\alpha \le 0$, and since $\Sigma(s,S_s^\alpha) = 0$ for $S_s^\alpha \le 0$, we have that $S_T^\alpha \le 0$; i.e., once $S^\alpha$ becomes nonpositive it stays nonpositive. By constancy of $\widetilde V$ on the negative half-line we deduce that $\widetilde V(S_T^\alpha) = \widetilde V(0)$.
 	Additionally the integrand appearing in the running cost is nonnegative, so from these observations and the estimate \eqref{eq:u_rho_lower_bound} we obtain $u_\rho(t,x) \ge  J(t,x,\alpha)$ for every $\alpha \in \mathcal{A}$. Taking supremum now shows that $u_\rho = u$.

    \paragraph{Step 2h: Convergence of the strategy at maturity.}
    Finally, we  establish that $\partial_xu(t,S_t)\to\widetilde V'(S_T)$ almost surely as $t\uparrow T$. Set $z=u-\widetilde u^{BS}$ and define the operator $\mathcal L_t=rx\partial_x+\frac12\sigma^2(t,x)x^2\partial_{xx}$. From \eqref{eq:dxu_bounds}–\eqref{eq:u_bounds}, we have
\[
(\partial_t+\mathcal L_t-r)z=-\tfrac{r\lambda}{2}(\partial_xu)^2,
\qquad
0\le z(t,x)\le \tfrac{r\lambda L_V^2}{2}(T-t).
\]
Since the right-hand side is uniformly bounded, the local parabolic $W^{1,2}_p$ estimate \cite[Theorem~7.22]{lieberman1996second}, followed by parabolic Sobolev embedding with $p>3$, yields on every compact interval $K\subset(0,\infty)$,
\begin{equation} \label{eq:partial_z_estimate}
\sup_{x\in K}|\partial_xz(t,x)|\le C_K\sqrt{T-t}
\end{equation}
for a constant $C_K  > 0$ and for all $t$ sufficiently close to $T$.

By Lemma~\ref{lem:ODE}, $\widetilde V'$ has a bounded representative $P\circ\Xi^{-1}$ that is continuous at every point where $\Xi^{-1}$ is a singleton. The discontinuity set $D$ is at most countable, since $\Xi$ is continuous and nondecreasing. Fix a continuity point $x_0$ of the bounded representative of $\widetilde V'$. Given $\epsilon>0$, choose an interval $I$ around $x_0$ on which $|\widetilde V'(x)-\widetilde V'(x_0)|\le\epsilon$ for a.e.\ $x \in I$. For $x$ sufficiently close to $x_0$ and sufficiently small $h>0$, absolute continuity of $\widetilde V$ and Cauchy–Schwarz give for some $C_I > 0$ the bound
$$
\left|
\frac{\widetilde u^{BS}(t,x+h)-\widetilde u^{BS}(t,x)}{h}
-\widetilde V'(x_0)
\right|
\le \epsilon+C_I\sqrt{T-t}.
$$
Indeed, the Feynman--Kac representation gives $\widetilde u^{\mathrm{BS}}(t,x) = \E[e^{-r(T-t)}\widetilde V(S_T)|S_t=x]$, so we can start the fundamental price process $S$ at the initial values $S_t^x = x$ and $S_t^{x+h} = x+h$ driven by the same Brownian motion. The quotient $(S_T^{x+h}-S_T^x)/h$ is nonnegative, has expectation
$e^{r(T-t)}$, and has uniformly bounded second moment by the
Lipschitz property of $\Sigma$. Moreover, for $x$ and $x+h$ in a
fixed smaller interval compactly contained in $I$, the usual
second-moment estimate for diffusion increments gives
$\mathbb Q(S_T^x\notin I\text{ or }S_T^{x+h}\notin I)
\le C_I(T-t)$.
Absolute continuity of $\widetilde V$ and Cauchy--Schwarz then
establish the displayed bound.
Sending $h\downarrow0$, then $(t,x)\to(T,x_0)$  and finally $\epsilon\downarrow0$, proves the claimed joint limit
\[
\lim_{\substack{t\uparrow T\\y\to x}}
\partial_x\widetilde u^{\mathrm{BS}}(t,y)=\widetilde V'(x),
\qquad x\notin D.
\]
Combining this with the estimate \eqref{eq:partial_z_estimate} for $\partial_xz$, the continuity of $S$, and the fact that $S_T$ has a density proves the asserted almost-sure convergence. 
  \end{proof}

  \begin{proof}[Proof of Corollary~\ref{cor:replication}] We just need to check that the assumptions of Theorem~\ref{thm:verification} are satisfied.
The function $u$ of Theorem~\ref{thm:PDE} satisfies item \ref{item:C12} by construction and item \ref{item:absolute_continuity} by Theorem~\ref{thm:terminal}, since $u(T,\cdot) = \widetilde V$. Item \ref{item:polynomial} follows from the bound
$|\partial_x u| \le L_V$ of \eqref{eq:dxu_bounds}, and item \ref{item:T_continuity} follows from the convergence conclusion of Theorem~\ref{thm:PDE}\ref{item:stochastic_rep}.
Theorem~\ref{thm:verification} then gives the claim, with
$\val_0^* = u(0,S_0)$ as the quadratic term in \eqref{eq:initial_cash} vanishes. 
\end{proof}

\end{document}